\documentclass[onecolumn,draftclsnofoot]{IEEEtran}

\pdfoutput=1

\usepackage{cite}
\usepackage{amsmath, amsthm, amsfonts, amssymb, multirow}

\usepackage{amsthm,amsfonts,amsmath,amssymb,dsfont,color}
\usepackage{commath,enumerate}
\usepackage{mathtools}

\usepackage{hyperref}
\usepackage{subcaption, multirow} 
\usepackage{makecell}

\definecolor{Milad}{RGB}{17, 117, 12}
\definecolor{Reviewer1}{RGB}{42, 0, 208}
\definecolor{Reviewer2}{RGB}{192, 9, 9}

\newenvironment{milad}{ }{}
\newenvironment{r1}{ }{}
\newenvironment{r2}{ }{}

\newcommand{\p}[2][{}]{\mathbb{P}_{#1} \intoo{#2}}
\newcommand{\e}[2][{}]{\mathbb{E}_{#1} \sbr{#2}}
\newcommand{\expp}[1]{\exp \intoo{#1}}
\newcommand{\define}{\triangleq}

\theoremstyle{theorem}
\newtheorem{thm}{{Theorem}}[]
\newtheorem{coro}{\textbf{Corollary}}[]
\newtheorem{lem}{Lemma}

\theoremstyle{remark}
\newtheorem{rem}{Remark}
\theoremstyle{definition}

\global\long\def\P#1{\operatorname{P}\left\{  #1\right\}  }

\def\F{\mathcal{F}}

\def\P{\mathcal{P}}
\def\E{\mathcal{E}}
\def\D{\mathcal{D}}
\def\C{\mathcal{C}}
\def\N{\mathcal{N}}

\def\Q{\mathcal{Q}}
\def\G{\mathcal{G}}
\def\ee{{\rm e}}

\newcommand{\mvec}[1]{{\boldsymbol #1}}

\begin{document}

\title{Using  Black-box Compression Algorithms for Phase Retrieval }

% u%se a multiple column layout for up to three different
%\author{
%\IEEEauthorblockN{Milad Bakhshizadeh}
%\IEEEauthorblockA{Department of Statistics\\
%Columbia University\\
% mb4041@columbia.edu}
%\and
%\IEEEauthorblockN{Shirin Jalali}
%\IEEEauthorblockA{Nokia-Bell labs\\ New Jersey \\ shirin.jalali@nokia-bell-labs.com}
%\and
%\IEEEauthorblockN{Arian Maleki}
%\IEEEauthorblockA{Department of Statistics\\
%Columbia University\\
%arian@stat.columbia.edu}}

\author{
Milad Bakhshizadeh, Arian Maleki, Shirin Jalali}


%Telephone: (800) 555--1212, Fax: (888) 555--1212}


\maketitle 


\begin{abstract}
\footnote{This paper was presented in part at ISIT 2018 and SPARS 2019.}
Compressive phase retrieval refers to the problem of recovering a structured $n$-dimensional complex-valued vector  from its phase-less  under-determined linear measurements. The non-linearity of measurements makes   designing theoretically-analyzable  efficient phase retrieval algorithms challenging. As a result, to a great extent, algorithms designed in this area are developed to take advantage of  simple structures such as sparsity and its convex generalizations. The goal of this paper is to move beyond  simple models through employing compression codes. Such codes are typically  developed  to take advantage of complex signal models to represent the signals as efficiently as possible.  In this work, it is shown how an existing compression code can be treated as a black box and integrated into an efficient solution for phase retrieval. First, COmpressive PhasE Retrieval (COPER) optimization, a computationally- intensive compression-based phase retrieval method,  is proposed. COPER provides  a theoretical framework for studying compression-based phase retrieval.   The number of measurements required by COPER is connected to $\kappa$, the  $\alpha$-dimension (closely related to the rate-distortion dimension) of the given family of compression codes.  To   finds the solution of COPER, an efficient iterative algorithm called gradient descent for COPER (GD-COPER) is proposed. It is  proven that under some mild conditions on the initialization and the compression, if  the number of measurements is larger than $ C \kappa^2 \log^2 n$, where $C$ is a  constant,  GD-COPER obtains an accurate estimate of the input vector in  polynomial time.  In the simulation results, JPEG2000 is integrated in GD-COPER  to confirm the superb performance of the resulting algorithm on real-world images. 

\end{abstract}


% \fi
\IEEEpeerreviewmaketitle

\section{Introduction}
\label{sec:CPR}

\subsection{Motivation}\label{ssec:motive}
Consider  the problem of recovering $\mvec{x} \in\mathcal{Q}$ from $m$ noisy phase-less linear observations 
\[
\mvec{y} = |A\mvec{x} |+\mvec{\epsilon},
\]
where $A\in\mathds{C}^{m\times n}$ and $\mvec{\epsilon}\in\mathds{R}^m$ denote the sensing matrix and the measurement noise,  respectively. Here $\mathcal{Q}$  denotes a compact subset of $\mathds{C}^n$ and $| \cdot|$ denotes the element-wise absolute value operator.   Assume that the  class of signals denotes by $\mathcal{Q}$  is ``structured'', but instead of the set $\mathcal{Q}$, or its underlying structure, for recovering $\mvec{x}$ from $\mvec{y}$,  we have access to   a compression code that takes advantage of the  structure of signals in  $\mathcal{Q}$ to compress them efficiently. For instance, consider  the class of images or videos for which compression algorithms, such as JPEG2000 or MPEG4, capture complicated structures within such signals and encode them efficiently. Employing such structures in a phase retrieval algorithm can reduce the number of measurements or equivalently increase the quality of the recovered signals. This raises the following questions:
 \begin{enumerate} 
 \item Is it possible to use a given compression algorithm for the recovery of $\mvec{x}$ from its undersampled set of phaseless observations? 
 \item What is the required number of observations (in terms of the rate-distortion performance of  the code), for almost zero-distortion recovery of $\mvec{x}$?
 \item Can we find polynomial time algorithms to use a given compression algorithm to recover $\mvec{x}$ from its undersampled set of phaseless observations? If so, how does the answer to the second question change if we want to approximate the solution in the polynomial time?
 \end{enumerate} 

By answering these questions we can immediately use the structures that are employed by the state-of-the-art compression algorithms, such as JPEG2000 or MPEG4, to improve the quality of the recovered signals or decrease the required number of measurements for a given quality. Furthermore, if the image or video compression communities discover new compression algorithms that are capable of employing new and more complicated structures, then the framework we develop in this paper obtains a phase retrieval algorithm, with no extra effort, that uses such complicated structures. 
  
In the remainder of this section, we first review  the formal definitions of compression algorithms and their rate-distortion performance measures. We will then briefly sates  our responses to the above three questions. Finally, we compare our contribution with the existing work in the literature. 

\subsection{Background on compression algorithms}

A rate-$r$ compression code is composed of an encoder mapping $\mathcal{E}$  and a decoder mapping $ \mathcal{D}$, where 
\begin{gather*}
\mathcal{E} : \mathds{C}^n \to \cbr{0,1}^r,\; {\rm and}\;\; \quad \mathcal{D} :  \cbr{0,1}^r \to \mathds{C}^n. 
\end{gather*}
The distortion performance of   the  compression code defined by mappings $(\E,\D)$ on   set $\mathcal{Q}$ is measured  as
\[
{\delta \triangleq \sup_{\mvec{x} \in \mathcal{Q}} \enVert{\mvec{x} - \mathcal{D} (\mathcal{E} (\mvec{x}))}. }
\]
Throughout the paper sometimes we use subscript $r$ for the encoder and decoder mappings as $(\E_r,\D_r)$ to highlight the rate of the code.   The codebook of  compression  code $(\mathcal{E}_r,\mathcal{D}_r)$ operating at rate $r$ is defined as  
\[
\mathcal{C}_r \triangleq \mathcal{D}_r (\mathcal{E}_r (\mathcal{Q}))=\{ \mathcal{D}_r (\mathcal{E}_r (\mvec{x})): \;\mvec{x}\in\mathcal{Q}\}.
\]
It is straightforward to confirm that $|\mathcal{C}_r| \leq 2^r$. 

In many application areas, the user has access to a family of compression codes. For instance, in image processing, a user can tune the rate in JPEG2000. Given a family of compression codes $\mathcal{F}=\cbr{(\mathcal{E}_r, \mathcal{D}_r)}_{r \in \mathds{N}}$  for set $\mathcal{Q}$ indexed by their rate $r$, let $\delta(r)$ denote the distortion performance of the code operating at rate $r$, i.e., $(\mathcal{E}_r, \mathcal{D}_r)$. Then, the rate-distortion function of this family of codes is defined as 
$$
r(\delta) \triangleq \inf \{r: \delta(r) < \delta\}.
$$
%It can be shown that, for a class of structured signals, $r(\delta) \ll n \log(1/\delta)$ \cite{com2com}. 
%Note that, up to a constant, $n \log (1/\delta)$ is the rate distortion of an optimal code for an $\ell_2$ ball in $\mathds{R}^n$.  
Define the $\alpha$-dimension of this family of codes as 
\begin{equation} \label{def alpha dim}
\dim_\alpha (\mathcal{F}) \triangleq \lim \sup_{\delta \rightarrow 0} \frac{r(\delta)}{\log \frac{1}{\delta}}.
\end{equation}
We will later show that this quantity is closely connected to  the number of measurements our proposed recovery methods require for  accurate phase retrieval. To offer some insight on  this quantity and what it measures consider the following well-known example. Let 
\[
\mathcal{B}_n = \cbr{\mvec{x} \in \mathds{R}^n \sVert[2] \enVert{\mvec{x}} \leq 1 },
\]
 and 
 \[
 \mathcal{S}_{n,k}= \cbr{ \mvec{x} \in \mathcal{B}_n \sVert[2] \enVert{\mvec{x}}_0 \leq k }
 \] 
 denote the unit $n$-dimensional  ball and the set of $k$-sparse signals in the unit ball, respectively.  It is straightforward to show that the  $\alpha$-dimension of any family of compression codes for $\mathcal{B}_n$ and $\mathcal{S}_{n,k}$ are lower-bounded  by $n$ and $k$, respectively. As shown in \cite{com2com}, there exist  compression codes that achieve these lower bounds in both cases.  A straightforward generalization of this result implies that for $k$-sparse signals in the unit ball in $\mathds{C}^n$, the $\alpha$-dimension of any  family of compression codes is lower-bounded by $2k$, and this bound is achievable.

%It is reasonable to think of $\dim_\alpha (\mathcal{Q}) \triangleq \inf\limits_{\mathcal{F}} \dim_\alpha(\mathcal{F})$ as the true dimension of compact set $\mathcal{Q}$.  Thus, the smaller $\dim_\alpha(\mathcal{Q}) - \dim_\alpha(\mathcal{F})$, the more structure of $\mathcal{Q}$ is captured by compression family $\mathcal{F}$.  Similarly, it can be shown that, for $k$-sparse signals in the unit ball of $\mathds{C}^n$, the $\alpha$-dimension of any  family of compression codes is above $2k$, and this bound is achievable.

\subsection{Summary of our contributions}\label{ssec:contribute}
Consider the problem of  noiseless phase retrieval, i.e., recovering $\mvec{x}$ up to its phase from   $\mvec{y} = |A\mvec{x} |$.  To answer the first two questions we raised in Section \ref{ssec:motive}, we propose COmpressible PhasE Retrieval (COPER) that employs  a given  compression code to solve  the described phase retrieval problem. Given measurement matrix $A\in\mathds{C}^{m\times n}$, define the distortion measure $d_A : \mathds{C}^n\times  \mathds{C}^n  \to \mathds{R}^+$ as follows
\begin{align}  \label{def d_A(Ax,Ac)}
d_A(\mvec{x}, \mvec{c}) \triangleq \sum_{k = 1}^m \intoo{ \envert{{\mvec{a}_k}^* \mvec{x}}^2 - \envert{{\mvec{a}_k}^* \mvec{c}}^2 }^2 =
\sum_{k = 1}^m \intoo{{\mvec{a}_k}^* (\mvec{x} \mvec{x}^* -\mvec{ c}\mvec{ c}^*){\mvec{a}_k}}^2,
\end{align}
where ${\mvec{a}_k}^*$ denotes the $k^{\text{th}}$ row of $A$.  When there is no ambiguity about the signal on interest $\mvec{x}$, we use $d_A(\mvec{c})$ instead of $d_A(\mvec{x}, \mvec{c})$.  Throughout the paper, for complex matrix $A$,  $A^*$ and $\bar{A}$ denote  its transposed-conjugate, and conjugate, respectively. Based on the defined  distance  measure,  we define COPER, a non-convex optimization problem for recovering $\mvec{x}$ from measurements $\mvec{y}$, as follows:
\begin{equation} \label{def x hat}
 \hat{\mvec{x}} = \arg\min _{\mvec{c} \in \C_r} d_A(\mvec{x},\mvec{c}).
\end{equation}
In other words, among all elements of the codebook, COPER finds the one for which  $\envert{A \mvec{c}}$ is closest to measurements $\mvec{y}$. Note that since $y_k = \envert{{\mvec{a}_k}^* \mvec{x}}$, to calculate $d_A(\mvec{x}, \mvec{c})$, we do not need to know $\mvec{x}$. 

In phase retrieval, since the measurements are phaseless,  the recovery of $\mvec{x}$ can never be exact; if $\mvec{x}$ satisfies $ 
\mvec{y} =  |A \mvec{x} |$, then so does ${\rm e}^{i \theta} \mvec{x}$, for any $\theta \in \mathds{R}$. Hence, following the standard procedure in the phase retrieval literature, we measure the quality of our estimate $\hat{\mvec{x}}$ as
\begin{equation*}
\inf_{\theta \in \intco{0,2 \pi}} \enVert{{\rm e}^{i \theta} \mvec{x}- \hat{\mvec{x}}}^2.
\end{equation*}
In Section \ref{sec:main}, we will bound $\inf\limits_{\theta} \enVert{{\rm e}^{i \theta} \mvec{x}- \hat{\mvec{x}}}^2$ in terms of the number of measurements and the rate-distortion function of the code. We will then show that $m> \dim_\alpha (\mathcal{F})$ observations suffice for an accurate recovery of $\mvec{x}$ by COPER.   For the aforementioned set of $k$-sparse signals that lie in the unit ball in $\mathds{C}^n$, using  a family of compression codes with an $\alpha$-dimension  of $2k$, our results imply that  COPER requires slightly more than $2k$ noise-free phase-less measurements for an accurate recovery.

Despite the nice theoretical properties of COPER, it is not directly useful in practice as  it is based on an exhaustive search over the set of all codewords, which is  exponentially   large. This leads us to the third question  asked in Section \ref{ssec:motive}. In response to this question, we introduce an iterative algorithm called gradient descent for COPER (GD-COPER). Let $ \mvec{z}_0 $ denote some selected  initial point, and define gradient of real-valued function $d$ as $\nabla d_A(\mvec{z}) \triangleq  \intoo{\frac{\partial d_A}{\partial \mvec{z}}}^* $, where 
\begin{equation*}
\frac{\partial d_A}{\partial \mvec{z}} \triangleq \frac{\partial d_A(\mvec{z},\bar{\mvec{z}})}{\partial \mvec{z}} \sVert[2]_{\bar{\mvec{z}} = \rm constant},
\end{equation*} is the Wirtinger derivative \cite{Wirtinger Flow}. 
The iterations of GD-COPER proceed as follows:
\begin{align} \label{eq: projective gradient step}
\mvec{s}_{t + 1} &\triangleq \mvec{z}_t - \mu \nabla d_A( \mvec{z}_t ),\nonumber\\
\quad \mvec{z}_{t + 1}& \triangleq  \mathcal{P}_{\C_r} ( \mvec{s}_{t + 1} ),
\end{align}
where $t$ represents the iteration index. Moreover, here, for $\mvec{z}\in\mathds{C}^n$, 
\[
d_A(\mvec{z}) =  d_A(\mvec{x}, \mvec{z}) = \sum_{k = 1}^m\intoo{ \envert{\mvec{a}_k^* \mvec{z}}^2 - \envert{ \mvec{a}_k^* \mvec{x}}^2 }^2=\sum_{k=1}^m \intoo{{\mvec{a}_k}^* (\mvec{x} \mvec{x}^* -\mvec{z}\mvec{z}^*){\mvec{a}_k}}^2,
\]
and therefore, 
\[
\nabla d_A(\mvec{z})  = 2 \sum_{k = 1}^m \intoo{ \envert{\mvec{a}_k^* \mvec{z}}^2 - \envert{ \mvec{a}_k^* \mvec{x}}^2 } \mvec{a}_k \mvec{a}_k^* \mvec{z}.
\] 

Here, $\C_r$, as defined earlier, is the set of codewords of the code, and $ \P_{\C_r} : \mathds{C}^n \to \C_r $ denotes the projection operator  on this set. That is, for $\mvec{s}\in\mathds{C}^n$,
\[
\P_{\C_r}(\mvec{s})=\arg\min_{\mvec{c}\in \C_r}\enVert{\mvec{c}-\mvec{s}}^2.
\]  
We show that, under some mild conditions on the initialization, given  $ m > C \dim_\alpha (\mathcal{F})^2 \log^2n $ phase-less measurements, GD-COPER finds an accurate estimate of $\mvec{x}$. Note that the number of measurements GD-COPER requires is considerably larger than what is needed by the combinatorial COPER optimization; in addition to the extra log factor, the number of measurements GD-COPER requires is proportional to $\dim_\alpha (\mathcal{F})^2$, unlike for COPER which only requires $\dim_\alpha (\mathcal{F})$ observations. While it might be  the case that the difference is due to our proof techniques and the gap is not something fundamental, based on our study of the problem, it seems  more plausible  to us  that the difference is the cost paid for having a polynomial time algorithm and cannot be closed  (except for probably removing the $\log$ factors).  

Finally, we perform  extensive numerical experiments to understand  the algorithmic properties of  GD-COPER, and evaluate the amount of gain a compression algorithm can offer for a simple `gradient descent'-type algorithm.

%\textcolor{blue}{Our analysis of GD-COPER seems to be sharp (except for probably the $\log$ factors), and hence this extra number of measurements seems to be the cost we pay for obtaining  polynomial time algorithms. Whether we can find other polynomial time algorithms that can recover $\mvec{x}$ with $\dim_\alpha (\mathcal{F})$ measurements remains open for future research. }

\subsection{Related work}
The problem of phase retrieval has been extensively studied  in the literature \cite{FineUp, CandesWirtingerFlow, CandesPhaseLift, YoninaGESPER, AmplitudeFlow, JaElHa15, ChCa15:nearly, DhLu17, DhThLu17, BaRo15, GoSt18, GhLaGoSt18, MaXuMa19, ChChFaMa18, WaZhGiCh18, RuSoHaLe18, SaLuLu19}. (Refer to \cite{JaElHa15} for a comprehensive review of the literature.) Since, unlike compressed sensing, in  phase retrieval, the measurements are a non-linear function of the input, even if the number of measurements is more than the ambient dimension of the signal, the recovery problem is still challenging. Hence, the primary focus of the field has  been on  developing and analyzing efficient recovery algorithms for {\em general} input signals. However, similar to compressed sensing, in most applications, the input signals are in fact structured. Therefore, taking such structures into account can lead to more efficient recovery algorithms with a lower number of required measurements or smaller reconstruction error. Hence, in more recent years, there has been work on phase retrieval of structured sources. In this domain,  most papers are concerned with standard structures, such as sparsity.  Assuming the signal is $k$-sparse, i.e., all of its coordinates but $k$ of  them are $0$, a variety of recovery algorithms have been proposed in the literature.  In the following, we briefly review some of such methods.

It is assumed in \cite{MoRoBa07} that the signal is sparse, or can be approximated well with few non-zero coefficients.  Furthermore, the authors  suppose that $l_1$-norm of the signal is known, and employ an iterative phase retrieval  algorithm. However, no theoretical guarantee is offered for the performance of the proposed  recovery algorithm. The lifting is used in \cite{Ohetal11, Ohetal12} to convexify the problem and take advantage of semidefinite programming (SDP) for signal recovery.  Since $\mvec{x} \in \mathds{C}^n$ is lifted to the space of $\mathds{C}^{n \times n}$ matrices,  the proposed algorithm is computationally demanding. Furthermore, the performance of the algorithm is guaranteed only under the assumption that the linear operator that appears in the SDP satisfies either the restricted isometry property or the coherence condition.  Similarly, \cite{ZaZhXi13} poses the problem of sparse phase retrieval as a non-convex optimization problem and uses the alternating direction method of multipliers (ADMM) to solve the problem. Generalized approximate message passing (GAMP) has been used in \cite{ScRa16} for the recovery of sparse signals. Despite the success of the ADMM and GAMP in simulation results, the theoretical properties of the algorithms are unknown.   Inspired by the Wirtinger flow algorithm, \cite{CaLiMa16} proposes a projected gradient descent for the recovery of $k$-sparse signals that resembles GD-COPER, proposed in this paper. However, GD-COPER uses a generic compression algorithm, while the projected gradient descent of  \cite{CaLiMa16} uses the projection on the set of all $k$-sparse vectors. Also, by combining the alternating minimization idea with Compressive Sampling Matching Pursuit (CoSaMP) \cite{JaHe17} has obtained another theoretically-supported algorithm for sparse phase retrieval with sample complexity of $O(k^2 \log n)$.  In  a more general setting, \cite{SaAbHa18, hand2016compressed} consider the regularized PhaseMax formulation, proposed in \cite{BaRo15, GoSt18}, and show that if a good anchor is available, then the algorithm is capable of recovering the signal from a number of measurements  proportional to the minimum required number of measurements. 
 
More recently, a few papers have used more sophisticated structures that are present in images to improve the performance of the recovery algorithms \cite{MeMaBa16, MeScVeBa18, HeMaRe14, soltanolkotabi2019structured}. For instance, by integrating a generic image denoiser in the iterations of the approximate message passing, similar to the approach of \cite{D_AMP}, \cite{MeMaBa16} improved the performance of the approximate message passing for the recovery of images. Since the message passing framework works mainly for  measurement matrices drawn independent and identically distributed (i.i.d.),  \cite{MeMaBa16} used the RED formulation, proposed originally in \cite{RoElMi17},  for the phase retrieval. The simulation results presented in \cite{MeMaBa16} suggest that the algorithms that are based on the RED formulation  (and a neural net denoiser) work well on Gaussian as well as coded diffraction and Fourier measurement matrices. Similarly, \cite{HeMaRe14} adds a total variation penalty to the non-convex formulation of phase retrieval problem and uses the ADMM approach for finding a local minimizer. Finally, \cite{ShAh18} uses a deep generative network to model images and then uses the learned model as a prior to help the phase retrieval recovery algorithms.

Finally, using  generic compression algorithms for compressed sensing and image restoration problems has been investigated before in \cite{com2com, BeigiEtAl, DaElBr17, RaJaEr16}. However, given the nonlinear nature of the measurement process in phase retrieval, similar to   other compressed sensing methods, neither  the theoretical nor the algorithmic tools and techniques  developed in the area of compression-based compressed sensing are directly applicable to  phase retrieval.  

In this paper, we develop a  theoretical framework for phase retrieval, i.e., recovering a signal from its  under-determined noise-free phase-less measurements,  that is applicable to  general structures employed by compression codes. This allows developing theoretically-analyzable algorithms that employ  structures much beyond those that  have been studied so far in the phase retrieval literature. We first  propose an idealistic compression-based phase retrieval recovery method that guides us on  the potential of such recovery methods. We then  propose a computationally-efficient and theoretically-analyzable algorithm that given enough measurements is  guaranteed to convergence to the desired solution.  We also obtain an upper bound on the gap between the performance of the efficient algorithm and that of the idealistic computationally-infeasible method. %Such an approach can reveal the weaknesses of the existing algorithm and proposes new ways for improving them. 

\subsection{Organization of the paper}
The organization of the paper is as follows.
Sections \ref{sec:main} and \ref{sec:gd-coper} state and prove our main theoretical contributions regarding the performance of COPER and GD-COPER, respectively.  Section \ref{sec: simulation} summarizes our simulation results.  Finally, section \ref{sec:proofs} gathers lemmas and theorems we have used to obtain the main results and proves them.

\section{Theoretical Guarantees of COPER} \label{sec:main}

Consider a class of signals $\mathds{Q}\subset\mathds{C}^n$ and a compression code with encoding and decoding mappings $(\mathcal{E}_r,\mathcal{D}_r)$ and codebook $\mathcal{C}_r$. Using the given compression code,   COPER   recovers $\mvec{x}\in\mathds{Q}$ from measurements $\mvec{y} = |A\mvec{x} |$ by solving the following combinatorial optimization: 
\[
 \hat{\mvec{x}} = \arg\min _{\mvec{c} \in \C_r} d_A(\mvec{x},\mvec{c}).
\]
The main goal of this section is to analyze the performance of this optimization. Toward this goal, we make the following assumptions:
\begin{enumerate}
\item For every $ \mvec{x} \in \Q$, we have $\|\mvec{x}\|_2 \leq 1$.\footnote{Given the fact that we need the rate-distortion function to be finite for every $\delta >0$, we expect $\Q$ to be a subset of $\{\mvec{x} \in \Re^n | \|\mvec{x}\|_2 \leq R\}$ for a given $R$. Without any loss of generality and for notational simplicity we have set $R=1$. }

\item The elements of $A$ are i.i.d.~ drawn from $\mathcal{N}(0,1) + i \mathcal{N}(0,1)$, where $i$ denotes the square root of $-1$.
\end{enumerate}
The following theorem obtains an upper bound on the accuracy of the COPER's estimate. 
\begin{thm} \label{theorem_main}
Let $\intoo{\mathcal{E}_r,\mathcal{C}_r}$ be a rate-$r$ compression code with distortion $\delta$.  Let $\mvec{x} \in \mathcal{Q}$ denotes the desired signal, and define sensing matrix $A$, as above.  Let $\hat{\mvec{x}}$ denotes the solution of COPER optimization. That is, $\hat{\mvec{x}} =  \arg\min\limits_{\mvec{c} \in C_r} d_A(\mvec{x},\mvec{c})$.  Then, we have
\begin{align} \label{Main theorem ineq}
\inf\limits_{\theta} \enVert{{\rm e}^{i \theta}\mvec{x} - \hat{\mvec{x}}}^2 
\leq 16 \sqrt{3}  \frac{1 + \tau_2}{\sqrt{\tau_1}} m \delta,
\end{align}
with probability at least
\begin{equation}
1 -  2^r {\rm e}^{\frac{m}{2} \intoo{K + \ln \tau_1 - \ln m} } - {\rm e}^{-2m \intoo{\tau_2 - \ln (1 + \tau_2)}}, 
\end{equation}
where
$K = \ln 2 \pi e$, and $\tau_1,\tau_2 $ are arbitrary positive real numbers.
\end{thm}
The general form of this theorem enables us to set $\tau_1, \tau_2$, and  $\delta$, and obtain different types of performance guarantees. Hence, before proving this theorem, we mention one specific choice that connects this result to the $\alpha$-dimension of the compression code in the next corollary. 

\begin{coro} \label{cor_main}
For large enough $r$, we have 
\begin{equation} 
\p{
\inf_\theta \enVert{{\rm e}^{i \theta} \mvec{x} - \hat{\mvec{x}}}^2 \leq C \delta^\epsilon} \geq 1 - {2}^{- c_{\eta} r} - {\rm e}^{-0.6 m},
\end{equation}
where $C = 32 \sqrt{3}$, and $m = \eta \frac{r}{\log \frac{1}{\delta}} $.
Given 
$ \eta > \frac{1}{1 - \epsilon}$, 
$c_{\eta}$ is a positive number less than $\eta(1 - \epsilon) - 1$.
\end{coro}
\begin{proof}
Given $\epsilon > 0$, $\eta > 0$, in Theorem \ref{theorem_main}, let
$\tau_1= m^2 \delta^{2 - 2 \epsilon},$
and 
$\tau_2 = 1.$
It follows that,
\begin{align}
 \inf\limits_{\theta} \enVert{{\rm e}^{i \theta} \mvec{x} - \hat{\mvec{x}} }^2 \leq 32 \sqrt{3}  \delta^\epsilon,
\end{align}
with probability
\begin{align} \label{10}
1 - {\rm e}^{r \intoo{\ln 2 + \frac{\eta \ln 2}{2 \ln \frac{1}{\delta}} \intoo{K + \ln m^2 \delta^{2 - 2 \epsilon} - \ln m} }} - {\rm e}^{-2m(1 - \ln 2)}.
\end{align}
Note that
$1 - \ln 2 > 0.3$,
and
\begin{gather}
\nonumber
1 + \frac{\eta }{2 \ln \frac{1}{\delta}} \intoo{K + \ln m^2 \delta^{2 - 2 \epsilon} - \ln m}
= 1 + \frac{\eta (K + \ln m)}{2 \ln \frac{1}{\delta}} - \eta(1 - \epsilon).
\end{gather}
Since $K, \eta$ are constants, and $m \to \eta \dim_\alpha (\mathcal{F})$  as  $\delta \to 0$.  Therefore, 
$$ \frac{\eta (K + \ln m)}{2 \ln \frac{1}{\delta}} \xrightarrow{\delta \to 0} 0. $$
Set any positive number $c_{\eta}$ such that  $0 < c_{\eta} < \eta (1 - \epsilon) - 1$, so for large enough $r$ we have
\begin{equation*}
1 + \frac{\eta (K + \ln m)}{2 \ln \frac{1}{\delta}} - \eta(1 - \epsilon) < - c_{\eta},
\end{equation*}
Thus
\begin{equation*}
{\rm e}^{r \ln 2 \intoo{1 + \frac{\eta}{2 \ln \frac{1}{\delta}} \intoo{K + \ln m^2 \delta^{2 - 2 \epsilon} - \ln m} }} < {2}^{- c_{\eta} r}.
\end{equation*}
\end{proof}

We would like to emphasize on a few points about this corollary:

\begin{rem} \label{remark: connect with alpha dim}
Corollary \ref{cor_main} shows that COPER  recovers the signal $\mvec{x}$ from $\eta \dim_\alpha(\Q)$ measurements  for any $\eta > 1$ with desired small distortion.  This happens with very high probability as $r \to \infty$. One simple implication of this result is that, in the case of $k$-sparse complex signals, COPER needs $2 \eta  k$ measurements for almost accurate recovery. Even if we had access to the sign of $A \mathbf{x}$, we could not recover $\mathbf{x}$ accurately with less than $2k$ measurements. Hence, in some sense this result is sharp. 
\end{rem}

\begin{rem}
This theorem guarantees the minimizer of  the COPER optimization. However, note that the COPER optimization is highly non-convex (optimization of a non-convex function over a discrete set). Hence, it is still not clear how we can get a good approximation of $\mathbf{\hat{x}}$ in polynomial time. This issue will be discussed in the next section. 
\end{rem}

Next we briefly review the main steps of the proof of Theorem \ref{theorem_main}. 

\begin{proof}[Roadmap of the proof of Theorem \ref{theorem_main}]
Here we mention the roadmap of the proof to help the readers understand the main ideas. The details are presented in section \ref{Properties of d}. Let
\begin{equation*} \label{def x tilde}
\tilde{\mvec{x}} = \mathcal{D} (\mathcal{E}(\mvec{x})).
\end{equation*}
Clearly, $\tilde{\mvec{x}} \in \mathcal{C}_r$. 
Note that by definition of $\delta(r)$, $\enVert{\mvec{x} - \tilde{\mvec{x}}} \leq \delta(r)$.  Moreover, by definition of $\hat{\mvec{x}}$, we have
\begin{equation}\label{eq:prooffirstCOPER}
d_A \intoo{\envert{A \mvec{x}}, \envert{A \hat{\mvec{x}}}} \leq d_A \intoo{\envert{A \mvec{x}}, \envert{A \tilde{\mvec{x}}}}.
\end{equation}
For a complex vector $\mvec{c}$, let $\lambda_1(\mvec{c}), \lambda_2(\mvec{c})$ denote the two non-zero eigenvalues of  $\mvec{x} \mvec{x}^* - \mvec{c} \mvec{c}^*$. Furthermore, let  $\lambda_{\max}(\mvec{c})$ denote the one with the largest absolute value. In Theorem \ref{thm:concentration of d} (proved in Appendix B) we prove that for any positive $\tau_1$ and $\tau_2$ we  have
\begin{equation} \label{eq 1:thm:concentrationfirst}
\p{d_A(\envert{A \mvec{x}} , \envert{A \mvec{c}}) > \lambda_{\max}^2(\mvec{c}) \tau_1, \; \forall \mvec{c} \in C_r}
\geq 1 -  2^r {\rm e}^{\frac{m}{2} \intoo{K + \ln \tau_1 - \ln m} },
\end{equation}
where $  K = \ln 2 \pi {\rm e}$ and
\begin{align} \label{eq 2:thm:concentrationfirst}
\p{d_A( \envert{A \mvec{x} },\envert{A \tilde{\mvec{x}}} ) < \lambda_{\max}^2(\tilde{\mvec{x}}) \intoo{4m (1 + \tau_2)}^2}
 \geq 1 - {\rm e}^{-2m \intoo{\tau_2 - \ln (1 + \tau_2)}}.
\end{align}
Combining \eqref{eq:prooffirstCOPER}, \eqref{eq 1:thm:concentrationfirst}, and \eqref{eq 2:thm:concentrationfirst}, we obtain
\begin{align}
\nonumber
\lambda_{\max}^2(\hat{\mvec{x}}) \tau_1
& < d_A \intoo{\envert{A \mvec{x}}, \envert{A \hat{\mvec{x}}}}
\\ & \nonumber
 \leq d \intoo{\envert{A \mvec{x}}, \envert{A \tilde{\mvec{x}}}}
 \\ &  
 < \lambda_{\max}^2(\tilde{\mvec{x}}) \intoo{4m (1 + \tau_2)}^2.
\end{align}
Therefore,
\begin{gather} \label{lamda_max < lamda_max}
\lambda_{\max}^2(\hat{\mvec{x}}) < \frac{16 m^2 (1 + \tau_2)^2}{\tau_1} \lambda_{\max}^2(\tilde{\mvec{x}}),
\end{gather}
with a probability larger than
$1 -  2^r {\rm e}^{\frac{m}{2} \intoo{K + \ln \tau_1 - \ln m} } - {\rm e}^{-2m \intoo{\tau_2 - \ln (1 + \tau_2)}}$. Hence, the main remaining step is to connect  $\lambda_{\max}^2({\hat{\mvec{x}}})$  with $\inf\limits_\theta \enVert{{\rm e}^{i \theta} \mvec{x} - \hat{\mvec{x}}}^2$. According to Lemma \ref{lem:lambdasVSvectors} (proved in the Appendix)
 we have
\begin{align} 
\label{lambda_max > norm}
\lambda_{\max}^2({\hat{\mvec{x}}}) & \geq \frac{1}{2} \intoo{ \lambda_1^2({\hat{\mvec{x}}}) + \lambda_2^2({\hat{\mvec{x}}})} 
\\ \nonumber &
 = \frac{1}{2} \intoo{\enVert{\mvec{x}}^2 - \enVert{{\hat{\mvec{x}}}}^2}^2 +  \intoo{\enVert{\mvec{x}}^2 \enVert{{\hat{\mvec{x}}}}^2 - \envert{\mvec{x}^* {\hat{\mvec{x}}}}^2}.
\end{align}

Recall {$\enVert{\mvec{x} - \tilde{\mvec{x}}} \leq \delta$} and since $\mvec{x},\tilde{\mvec{x}} \in \mathcal{Q}$ we have $\enVert{\mvec{x}},\enVert{\tilde{\mvec{x}}} \leq 1$, thus
\begin{gather} 
 \intoo{\enVert{\mvec{x}} + \enVert{\tilde{\mvec{x}}}}^2 \intoo{\enVert{\mvec{x}} - \enVert{\tilde{\mvec{x}}}}^2 \leq 4 \delta^2.
\end{gather}
Moreover,
\begin{align*}
\nonumber
\delta^2 & \geq \enVert{\mvec{x} - \tilde{\mvec{x}}}^2 
\\ \nonumber &
= \enVert{\mvec{x}}^2 + \enVert{\tilde{\mvec{x}}}^2 - \mvec{x}^* \tilde{\mvec{x}} - \tilde{\mvec{x}}^* \mvec{x}
\\ \nonumber &
\geq \enVert{\mvec{x}}^2 + \enVert{\tilde{\mvec{x}}}^2 - 2 \envert{\mvec{x}^* \tilde{\mvec{x}}} 
\\ \nonumber &
\geq 2 \intoo{\enVert{\mvec{x}} \enVert{\tilde{\mvec{x}}}  - \envert{\mvec{x}^* \tilde{\mvec{x}}}},
\end{align*}
so we have 
$\intoo{\enVert{\mvec{x}} \enVert{\tilde{\mvec{x}}}  - \envert{\mvec{x}^* \tilde{\mvec{x}}}} \leq \frac{\delta^2}{2}, $
which implies
\begin{align}
 \intoo{\enVert{\mvec{x}}^2 \enVert{\tilde{\mvec{x}}}^2 - \envert{\mvec{x}^* \tilde{\mvec{x}}}^2} \leq \delta^2.
\end{align}
Similarly, Lemma \ref{lem:lambdasVSvectors} implies
\begin{align}  \label{lambda_max < norm}
\lambda_{\max}^2(\tilde{\mvec{x}}) & \leq  \intoo{ \lambda_1^2(\tilde{\mvec{x}}) + \lambda_2(\tilde{\mvec{x}})^2} 
\\ \nonumber &
=
 \intoo{\enVert{\mvec{x}}^2 - \enVert{\tilde{\mvec{x}}}^2}^2 +  2 \intoo{\enVert{\mvec{x}}^2 \enVert{\tilde{\mvec{x}}}^2 - \envert{\mvec{x}^* \tilde{\mvec{x}}}^2}
 \\ \nonumber &
\leq
 6 \delta^2.
\end{align}
Therefore,  combining (\ref{lamda_max < lamda_max}),(\ref{lambda_max > norm}),(\ref{lambda_max < norm}), we have
\begin{align}  \nonumber
\frac{1}{2} \intoo{\enVert{\mvec{x}}^2 - \enVert{\hat{\mvec{x}}}^2}^2 +  \intoo{\enVert{\mvec{x}}^2 \enVert{\hat{\mvec{x}}}^2 - \envert{\mvec{x}^* \hat{\mvec{x}}}^2} &\leq \lambda_{\max}^2(\hat{\mvec{x}})
\\ \nonumber &
< \frac{16 m^2 (1 + \tau_2)^2}{\tau_1} \lambda_{\max}^2(\tilde{\mvec{x}})
\\  & \label{30}
\leq \frac{96 m^2 (1 + \tau_2)^2}{\tau_1} \delta^2
 \end{align}
  with probability larger than  $1 -  2^r {\rm e}^{\frac{m}{2} \intoo{K + \ln \tau_1 - \ln m} } - {\rm e}^{-2m \intoo{\tau_2 - \ln (1 + \tau_2)}}$. Finally, Lemma \ref{correct phase} connects the left hand side of \eqref{30}  with $\intoo{\inf\limits_{\theta} \enVert{{\rm e}^{i \theta}\mvec{x} - \hat{\mvec{x}}}^2}^2 $. Hence, using Lemma \ref{correct phase} we have
\begin{align*}
\intoo{\inf\limits_{\theta} \enVert{{\rm e}^{i \theta}\mvec{x} - \hat{\mvec{x}}}^2}^2  \leq \frac{768 m^2 (1 + \tau_2)^2}{\tau_1} \delta^2, 
\end{align*}
which means
\begin{align*}
\p{\inf\limits_{\theta} \enVert{{\rm e}^{i \theta}\mvec{x} - \hat{\mvec{x}}}^2  \leq 16 \sqrt{3} \frac{1 + \tau_2}{\sqrt{\tau_1}} m  \delta} \geq 1 -  2^r {\rm e}^{\frac{m}{2} \intoo{K + \ln \tau_1 - \ln m} } - {\rm e}^{-2m \intoo{\tau_2 - \ln (1 + \tau_2)}}, 
\end{align*}
where  $K = \ln 2 \pi e, \; \tau_1,\tau_2 > 0$.

\end{proof}

\section{Theoretical Guarantees of GD-COPER}
\label{sec:gd-coper}

As discussed before, COPER is based on an exhaustive search over the space of all codewords, and is hence computationally very demanding, if not infeasible. This section aims to prove that with more measurements GD-COPER, introduced in Section \ref{ssec:contribute}, reaches a good approximation of the solution of COPER in polynomial time.
In this section, we assume that
 \begin{equation} \label{asump normalized Q and codeword}
\enVert{ \mvec{x} } = 1, \; \enVert{\mvec{z}} = 1, \quad \forall \; \mvec{z} \in \C_r .
\end{equation}
This assumption enables us to state our theoretical results in a simpler form. We will have a more detailed discussion about this assumption in Section \ref{sec:assumption}.
Recall  that the iterations of the GD-COPER algorithm are given by
\begin{align} \label{eq:GD-COPER}
\mvec{s}_{t + 1} &\triangleq \mvec{z}_t - \mu \nabla d_A ( \mvec{z}_t ),\nonumber\\
\quad \mvec{z}_{t + 1}& \triangleq  \mathcal{P}_{\C_r} ( \mvec{s}_{t + 1}),
\end{align}

\begin{rem}
The projection   step  in  GD-COPER, i.e., $\mvec{z}_{t + 1} \triangleq  \mathcal{P}_{\C_r} ( \mvec{s}_{t + 1})$,  might   seem computationally expensive, as the codebook $\C_r$ is exponentially large. However,  for a good compression code, it is natural to  expect  the projection on the set of codewords to be equivalent to the successive  application of the encoder and the decoder mappings  of the compression code. In other words, we expect $\mathcal{P}_{\C_r} (\cdot) = \mathcal{D}_r (\mathcal{E}_r (\cdot))$ or, at least, $\mathcal{D}_r (\mathcal{E}_r (\cdot))$ to be very close to $\mathcal{P}_{\C_r} (\cdot)$. We will present an example in Section \ref{sec:assumption} to justify this claim. We will also provide theoretical results regarding the robustness of GD-COPER to this assumption in Theorem \ref{thm convergence non-ideal compression}. Hence, in our simulations, we  use this observation  and run the GD-COPER algorithm as follows:
\begin{align*} 
\mvec{s}_{t + 1} &= \mvec{z}_t - \mu \nabla d_A ( \mvec{z}_t ),\nonumber\\
\quad \mvec{z}_{t + 1}& =  \mathcal{D}_r  (\mathcal{E}_r ( \mvec{s}_{t + 1} )).
\end{align*}
%It is straightforward to prove that even if $\mathcal{D}_r  (\mathcal{E}_r ( \cdot ))$ is an `approximate projection', the algorithm is still capable of obtaining an accurate estimate \cite{BeigiEtAl}. 
\end{rem}

We first mention our generic result. We will then, simplify this result in a few corollaries to interpret it and compare with the existing work.

\begin{thm} \label{thm convergence of algorithm}
For a fixed signal $\mvec{x} \in \Q$, define $\mvec{z}_t  \in \C_r$ as in \eqref{eq:GD-COPER} with $ \mu = \frac{1}{8 m} $.  Suppose that for all $  \theta \in \mathds{R}$,   $\ee^ {i \theta} \mvec{x} \in \Q$. Define
$ \theta_t \triangleq \arg \min\limits_{\theta \in \mathds{R}} \enVert{\mvec{z}_t -  {\rm e}^{i \theta} \mvec{x}}$.  For all $ \epsilon \geq C_2 m^{-\frac{1}{3}} $, with probability at least $ 1 - C_3 {\rm e}^{- C_1 \sqrt{m \epsilon} + (3 \ln 2) r} $, where $C_1, C_2, C_3 > 0$ are absolute constants, for  $t=1,2,\ldots$, we have
\begin{equation}\label{eq:theoremGD}
\enVert{ \mvec{z}_{t + 1} -  {\rm e}^{i \theta_t} \mvec{x} } \leq \intoo{ \enVert{ \mvec{z}_t -  {\rm e}^{i \theta_t} \mvec{x}} + \epsilon } \enVert{ \mvec{z}_t -  {\rm e}^{i \theta_t} \mvec{x}} + 3 \delta_r. 
\end{equation}

\end{thm}

Before  proving this theorem, we first simplify the statement of this theorem and compare it with Corollary \ref{cor_main}.  
The following Corollary shows having enough measurements, we may get arbitrary close to the COPER's solution with this algorithm, with exponentially high probability. 

\begin{coro}\label{cor:relaxedinit}
Consider the same setup as in Theorem \ref{thm convergence of algorithm}. Assume that  $ \inf\limits_{\theta \in \mathds{R}} \enVert{{\rm e}^{i \theta} \mvec{x} - \mvec{z}_0}  = 1 - 2 \tau < 1 $, for some $\tau > 0$.  Then, if $\delta \leq \frac{\tau (1 - 2 \tau)}{3}$,  and 
\[
m \geq \max \cbr{ \intoo{\frac{C_2}{\tau}}^3, \frac{C_4}{\tau} (\dim_\alpha (\F) \log \frac{1}{\delta} )^2},
\] 
after $T$ iterations of GD-COPER, 
\begin{equation}
\inf_{\theta \in \mathds{R}} \enVert{ {\rm e}^{i \theta} \mvec{x} - \mvec{z}_T }  \leq (1 - 2 \tau) \intoo{1 - \tau}^T  +\frac{3}{\tau} \delta_r,
\end{equation}
with probability at least
\begin{equation}
1 - C_3 {\rm e}^{ - \frac{C_1 \sqrt{\tau}}{2} \sqrt{m}}.
\end{equation}
Here, $C_1, C_2, C_3$ are the constants introduced in Theorem \ref{thm convergence of algorithm} with $ \epsilon = \tau$, and $C_4$ is an absolute constant.

\end{coro}

\begin{proof}
We apply Theorem \ref{thm convergence of algorithm} with $ \epsilon = \tau $, thus we need $ \tau = \epsilon \geq C_2 m^{- \frac{1}{3}} $, hence
\begin{equation} \label{coro proof m lower bound 1}
m \geq \intoo{\frac{C_2}{\tau}}^3.
\end{equation}
With a probability larger than $1 - C_3 {\rm e}^{- C_1 \sqrt{m \tau} + (3 \ln 2) r} $ at each iteration we have
\begin{equation} \label{rmk conv alg proof 1}
\enVert{ \mvec{z}_{t + 1} - {\rm e}^{i \theta_{t + 1}}  \mvec{x} } \leq \intoo{ \enVert{ \mvec{z}_t - {\rm e}^{i \theta_{t}}  \mvec{x}} + \tau } \enVert{ \mvec{z}_t - {\rm e}^{i \theta_{t}}  \mvec{x}} + 3 \delta,
\end{equation}
hence,
\begin{align} \nonumber
\enVert{ \mvec{z}_{t + 1} - {\rm e}^{i \theta_{t + 1}} \mvec{x} } 
& \leq 
 \intoo{ \enVert{ \mvec{z}_t - {\rm e}^{i \theta_{t}}  \mvec{x}} + \tau } \enVert{ \mvec{z}_t - {\rm e}^{i \theta_{t}}  \mvec{x}} + 3 \delta
\\ & \leq \nonumber
(1 - \tau) (1 - 2 \tau) + 3 \delta
\\ & \leq \label{rmk conv alg proof 3}
1 - 2 \tau,
\end{align}
since $ \delta \leq \frac{\tau (1 - 2 \tau)}{3} $.
Therefore, by \eqref{rmk conv alg proof 1} and \eqref{rmk conv alg proof 3}, we may deduce that
\begin{equation*}
\enVert{ \mvec{z}_{t + 1} - {\rm e}^{i \theta_{t + 1}}  \mvec{x} } \leq \intoo{ 1 - \tau } \enVert{ \mvec{z}_t - {\rm e}^{i \theta_{t}}  \mvec{x}} + 3 \delta.
\end{equation*}
Hence we get,
\begin{align} \nonumber
\enVert{\mvec{x} - {\rm e}^{i \theta_{T}}  \mvec{z}_T} 
& \leq
 \intoo{1 -  \tau}^T \enVert{{\rm e}^{i \theta_{0}}  \mvec{x} - \mvec{z}_0}  + 3 \delta \intoo{1 + 1 -  \tau + \intoo{1 -  \tau}^2 + \ldots +\intoo{1 -  \tau}^{T- 1}}
 \\ & \leq \label{coro proof main ineq}
(1 - 2 \tau) \intoo{1 - 
\tau}^T  +\frac{3 \delta}{\tau}.
\end{align}

Moreover, if $\G$ denotes the event under which Theorem \ref{thm convergence of algorithm} holds, i.e. \eqref{eq:theoremGD} is satisfied, then
\begin{align*}
\p{\G} 
& \geq
1 - C_3 {\rm e}^{- {C_1} \sqrt{m \tau} + (3 \ln 2) r}
\\ & \geq
1 - C_3 {\rm e}^{- \frac{C_1 \sqrt{\tau}}{2}\sqrt{m}},
\end{align*}
once we have $ (3 \ln 2 ) r \leq \frac{C_1 \sqrt{\tau m}}{2} $, or equivalently
\begin{equation*}
m \geq \frac{C_4'}{\tau} r^2, \quad C_4' = \intoo{ \frac{6 \ln 2}{C_1} }^2.
\end{equation*}
Since $ \lim\limits_{r \to \infty} \frac{r}{\log \frac{1}{\delta}} = \dim_\alpha (\F) $, for large enough $r$, we have $ r \leq 1.5 \dim_\alpha (\F) \log \frac{1}{\delta} $. Hence,
\begin{equation} \label{coro proof m lower bound 2}
m \geq \frac{C_4}{\tau} \intoo{\dim_\alpha(\F) \log \frac{1}{\delta}}^2, 
\end{equation}
where $C_4 = 2.25 C_4'$.

Since we assumed 
$
m \geq \max \intoo{\intoo{\frac{C_2}{\tau}}^3, \frac{C_4}{\tau} \intoo{\dim_\alpha(\F) \log \frac{1}{\delta}}^2},
$
both \eqref{coro proof m lower bound 1} and \eqref{coro proof m lower bound 2} are satisfied.  Then by \eqref{coro proof main ineq} we obtain
\begin{equation}
\inf_{\theta \in \mathds{R}} \enVert{ {\rm e}^{i \theta} \mvec{x} - \mvec{z}_T } \leq (1 - 2 \tau) \intoo{1 - 
\tau}^T  +\frac{3 \delta}{\tau}.
\end{equation}
\end{proof}

\begin{rem}  \label{rem: conv with |x - z_0| < 1/2}
Consider the same setup as in Corollary \ref{cor:relaxedinit} and let $\tau = \frac{1}{4}$.  Then, for $ \delta \leq \frac{1}{24}$,  and
$$m \geq \max \cbr{ \intoo{{4 C_2}}^3, {4 C_4} (\dim_\alpha (\F) \log \frac{1}{\delta} )^2},$$
 after $T$ iterations of the GD-COPER algorithm, 
\begin{equation}
\inf_{\theta \in \mathds{R}} \enVert{ {\rm e}^{i \theta} \mvec{x} - \mvec{z}_T } \leq \frac{1}{2} \intoo{\frac{3}{4}}^T  +12 \delta,
\end{equation}
with  a probability greater than $1  - C_3 {\rm e}^{- \frac{C_1}{4} \sqrt{m} },$ 
where $C_i, \; i \in \cbr{1,\ldots, 4}$, are positive constants. 
\end{rem}

\begin{rem}
If $n$ is a large number (which is the case in almost all the applications of the phase retrieval), then we can set $\delta = \frac{1}{n}$ in Remark \ref{rem: conv with |x - z_0| < 1/2}, and  conclude that with $m \geq C'_4 \dim_\alpha (\F)^2 \log^2n$ measurements, GD-COPER can with high probability obtain an accurate estimate of $\mvec{x}$ (with $O(1/n)$ distortion). Hence, the number of measurements GD-COPER requires is substantially more than the number of observations COPER requires. At this stage, it is not clear whether this discrepancy is an artifact of our proof technique, the limitation of the GD-COPER algorithm, or a fundamental limitation of the polynomial time algorithms. We leave the full study of this phenomenon for future research. We should also mention that in the case of sparse phase retrieval \cite{CaLiMa16} observed that even under a good initialization the thresholded Wirtinger flow algorithm can recover the signal exactly with $k^2 \log n$ measurements, which is again consistent with our result. Furthermore, the paper presented other evidences to suggest that to obtain a good initialization $k^2 \log n$ measurements are required. It is also worth mentioning that \cite{LiVl13} has shown that convex relaxation methods will not work if the number of measurements is less than $c k^2/ \log n$ for constant $c$.  
\end{rem}

\begin{rem}
Corollary \ref{cor:relaxedinit} proves the accuracy of GD-COPER under an initialization that satisfies $\inf\limits_{\theta \in \mathds{R}} \enVert{{\rm e}^{i \theta} \mvec{x} - \mvec{z}_0}  = 1 - 2 \tau < 1$. Finding an initialization that theoretically satisfies this condition is a good research direction for future research. However, as will be clarified in our simulation results and has also be discussed elsewhere \cite{MaXuMa19}, the choice of initialization seems to have a minor effect (almost none) on the performance of GD-COPER (and other iterative algorithms). Hence, in our simulation results we have initialized GD-COPER with a white image. 
\end{rem}

\begin{proof}[Roadmap of the proof of Theorem \ref{thm convergence of algorithm}]
Let $\tilde{\mvec{x}} = \P_{\C_r}(  {\rm e}^{i \theta_t} \mvec{x}) $.  Since $ \mvec{z}_{t + 1} = \P_{\C_r} (\mvec{s}_{t + 1})$ and $ \tilde{\mvec{x}} \in \C_r $, we have
\begin{align} \nonumber
\enVert{ \mvec{s}_{t + 1} -  \tilde{\mvec{x}} }^2 
& \geq \enVert{ \mvec{s}_{t + 1} -  {\mvec{z}_{t + 1}} }^2 
\\ & =
\enVert{\mvec{s}_{t + 1} -  \tilde{\mvec{x}} }^2 + \enVert{ \tilde{ \mvec{x}} - \mvec{z}_{t + 1}}^2 + 2 \Re \intoo{ ( \tilde{ \mvec{x}} - \mvec{z}_{t + 1})^* (\mvec{s}_{t + 1} - \tilde{\mvec{x}}) }.
\end{align}
Therefore,
\begin{equation} \label{upper bdd of error by real inner product}
\enVert{ \tilde{\mvec{x}} - \mvec{z}_{t + 1}}^2 \leq 2 \Re \intoo{ ( \tilde{ \mvec{x}} - \mvec{z}_{t + 1})^* (\tilde{\mvec{x}} - \mvec{s}_{t + 1}) }.
\end{equation}
Let 
\begin{align}
 \mvec{v}_{t} \triangleq \frac{ \tilde{\mvec{x}}  - \mvec{z}_{t + 1}}{\enVert{ \tilde{\mvec{x}}  - \mvec{z}_{t + 1} }}.
\end{align}
Using this definition, \eqref{upper bdd of error by real inner product} can be written as
\begin{equation} \label{def and ineq v_z}
\enVert{ \tilde{\mvec{x}} - \mvec{z}_{t + 1}} \leq 2 \Re \intoo{ \mvec{v}_{t}^* ( \tilde{\mvec{x}} -  \mvec{s}_{t + 1}) }.
\end{equation}

Recall that  $ \e{ \nabla d_A (\mvec{z}) } 
= 8m ( \mvec{z} \mvec{z}^* - \mvec{x} \mvec{x}^* ) \mvec{z} $. Hence, 
\begin{align} \nonumber
 \tilde{\mvec{x}} - \mvec{s}_{t + 1} 
& =
 \tilde{\mvec{x}} -  {\rm e}^{i \theta_t} \mvec{x}  + {\rm e}^{i \theta_t} \mvec{x}  - \intoo{ \mvec{z}_t - \frac{1}{8 m} \e{ \nabla d_A( \mvec{z}_t) } + \frac{1}{8m} \intoo{\e{ \nabla d_A(\mvec{z}_t)} - \nabla d_A(\mvec{z}_t)} }
\\ & = \nonumber
 \tilde{\mvec{x}} -  {\rm e}^{i \theta_t} \mvec{x}  + {\rm e}^{i \theta_t} \mvec{x}  - \intoo{\mvec{z}_t - \mvec{z}_t + (\mvec{x}^* \mvec{z}_t) \mvec{x}} + \frac{1}{8m} \intoo{ \nabla d_A (\mvec{z}_t) - \e{ \nabla d_A(\mvec{z}_t)} }
\\ & = \label{thm conv alg proof 1}
\tilde{\mvec{x}} -  {\rm e}^{i \theta_t} \mvec{x}  + (1 - ({\rm e}^{i \theta_t} \mvec{x})^* \mvec{z}_t) {\rm e}^{i \theta_t} \mvec{x}  + \frac{1}{8m} \intoo{ \nabla d_A(\mvec{z}_t) - \e{ \nabla d_A(\mvec{z}_t)} }.
\end{align}
Note that 
$ \enVert{\tilde{\mvec{x}} - {\rm e}^{i \theta_t} \mvec{x}} \leq \delta_r$.
Also,  since by lemma \ref{lem correct phase_a} we have $1 -({\rm e}^{i \theta_t} \mvec{x})^* \mvec{z}_t = \frac{1}{2} \enVert{  {\rm e}^{i \theta_t} \mvec{x} - \mvec{z}_t }^2
$
and $ \enVert{{\rm e}^{i \theta_t} \mvec{x}} = \enVert{\mvec{v}_{t}} = 1 $, by the triangle inequality,  from \eqref{def and ineq v_z} and \eqref{thm conv alg proof 1}, we have
\begin{eqnarray} \nonumber
\lefteqn{
\enVert{  { {\rm e}^{i \theta_t} \mvec{x}} - \mvec{z}_{t + 1} }
  \leq
\enVert{{{\rm e}^{i \theta_t} \mvec{x}} - \tilde{\mvec{x}}} + \enVert{  \tilde{\mvec{x}} - \mvec{z}_{t + 1} } 
} 
\\ \nonumber
& \leq &
\delta_r + 2 \Re \intoo{ \mvec{v}_{t}^* ( \tilde{\mvec{x}} -  \mvec{s}_{t + 1}) }
\\ & \nonumber
\leq & \delta_r + 2 \enVert{\mvec{v}_t} \enVert{\tilde{\mvec{x}} - {\rm e}^{i \theta_t} \mvec{x}} + 2 (1 - ({\rm e}^{i \theta_t} \mvec{x})^* \mvec{z}_t) \enVert{\mvec{v}_t} \enVert{ {\rm e}^{i \theta_t} \mvec{x}} + \frac{1}{4m} \Re \intoo{\mvec{v}_t^* \intoo{\nabla d_A(\mvec{z}_t) - \e{\nabla d_A(\mvec{z}_t)}} }
\\
& \leq & \label{thm conv alg proof 2}
\delta_r + 2 \delta_r + \enVert{{{\rm e}^{i \theta_t} \mvec{x}} - \mvec{z}_{t} }^2 + \frac{1}{4 m} \Re  \intoo{\mvec{v}_{t}^*\intoo{ \nabla d_A(\mvec{z}_t) - \e{ \nabla d_A(\mvec{z}_t)} } }.
\end{eqnarray}
Define event $\cal G$ as follows
\begin{equation} \label{def event G}
{\cal G} \triangleq \cbr{ \frac{1}{4m} \Re \left(\mvec{v}^* \intoo{ \nabla d_A(\mvec{z}) - \e{ \nabla d_A(\mvec{z}) } } \right) \leq \epsilon \inf_{\theta \in \mathds{R}} \enVert{ {\rm e}^{i \theta} \mvec{x} - \mvec{z} }, \quad \mvec{v} = \frac{\tilde{\mvec{x}} - \mvec{z'}}{\enVert{\tilde{\mvec{x}} - \mvec{z'}}}, \qquad \forall \tilde{\mvec{x}}, \mvec{z}, \mvec{z'} \in \C_r }. 
\end{equation}
 
 One difficulty  in bounding  $\p{\cal G}$ is that $\nabla d_A(\mvec{z})$ is summation of heavy-tailed random variables.  To address this issue,  in Lemma \ref{lem concentration of gradient} (stated and proved in Section \ref{ssec:concgradient}), we develop a technique that yields sharp concentration bounds for such summations. Applying Lemma \ref{lem concentration of gradient} with $4 \epsilon$, for a given $\mvec{v} \in \mathds{C}^{n}$ with $\enVert{\mvec{v}} = 1$ and $ \mvec{z} \in \C_r $, we get constants $C_1, C_2, C_3 > 0$ for which, for every $\epsilon \geq C_2 m^{-\frac{1}{3}}$,
\begin{equation}\label{eq:concentrationsinglepart}
\p{ \envert{ \Re \intoo{  \mvec{v}^* \intoo{ \nabla d_A(\mvec{z}) - \e{\nabla d_A(\mvec{z})}} } } > 4 m \epsilon \inf_{\theta \in \mathds{R}} \enVert{{\rm e}^{i \theta} \mvec{x} - \mvec{z} } } \leq  C_3 {\rm e}^{ - C_1 \sqrt{m \epsilon}}. 
\end{equation}
Hence, combining \eqref{eq:concentrationsinglepart} with the union bound, for every $\epsilon \geq C_2 m^{-\frac{1}{3}}$, we have 
\begin{equation}
\p{\cal G} \geq 1 - 2^{3r} C_3 {\rm e}^{- C_1 \sqrt{m \epsilon}}. %\quad \forall \epsilon \geq C_2 m^{-\frac{1}{3}}.
\end{equation}
Therefore, conditioned on  $\cal G$ we have  
\begin{align}
\frac{1}{4 m} \Re  \intoo{\mvec{v}_{t}^*\intoo{ \nabla d_A(\mvec{z}_t) - \e{ \nabla d_A(\mvec{z}_t)} } } & \leq \epsilon \inf_{\theta \in \mathds{R}} \enVert{{\rm e}^{i \theta} \mvec{x} - \mvec{z}_t} 
= \epsilon \enVert{{\rm e}^{i \theta_t} \mvec{x} - \mvec{z}_t},
\end{align}
hence, \eqref{thm conv alg proof 2} implies that, for all  $t\in\{ 1 \cdots,T\}$, 
\begin{equation}
\enVert{ \mvec{z}_{t + 1} -  {\rm e}^{i \theta_t} \mvec{x} } \leq \intoo{ \enVert{ \mvec{z}_t - {\rm e}^{i \theta_t} \mvec{x}} + \epsilon } \enVert{ \mvec{z}_t - {\rm e}^{i \theta_t} \mvec{x}} + 3 \delta_r,
\end{equation}
which in turn leads to \eqref{eq:theoremGD}. 

\end{proof}

\section{Simulation results}
\label{sec: simulation}

The main goal of this section is to experimentally evaluate the performance of our algorithm. Furthermore, comparisons between our algorithm and Wirtinger flow will be presented to empirically evaluate the amount of gain a compression scheme offers. Since the publicly available compression algorithms work with real-valued images, in our simulation results we focus on real-valued signals and measurements only. Note that even though our theoretical results are presented for complex-valued signals, the extension to real-valued signals is straightforward. For the sake of brevity, we did not include such extensions. 

\subsection{Measurement matrices}

 We consider two types of measurement matrices: (i) Gaussian measurement matrices in which $A_{ij} \overset{iid}{\sim} N(0,1)$, and (ii) coded diffraction patterns in which the measurements are constructed in the following way:

\begin{equation}
{y}_{i, l} = \envert{\sum_{k = 1}^n {x}_k \cos \intoo{\frac{i \pi}{n} \intoo{k + \frac{1}{2}}} {M}_{l, k}}.
\end{equation}
In these measurements, $M_{l,k}$ modulates the entries of the signal and is drawn from the following distribution:
\begin{equation} \label{def: mask}
{M}_{l, k} \overset{iid}{\sim} \begin{dcases}
1 & \text{with probability} \quad \frac{1}{4}
\\
-1 & \text{with probability} \quad \frac{1}{4}
\\
0 & \text{with probability} \quad \frac{1}{2}
\end{dcases},
\qquad
1 \leq k \leq n,
\quad
1 \leq l \leq L.
\end{equation}
Coded diffraction patterns have recently received attention in the phase retrieval problem since they can outperform the Fourier matrices. Note that due to the construction of the coded diffraction measurement matrices, the imaging system is over-sampled by the factor $L$. Our simulation results will cover $L \in \{1, 2, \ldots, 15\}$. As we will discuss later, GD-COPER algorithm is capable of performing well even when $L=1$. Note that given that the signs are missing, this can be considered as an under-sampled situation.

\subsection{Setting the parameters}\label{ssec:parametersetting}
\subsubsection{GD-COPER}

 In our simulation results, we will be using natural images, and JPEG2000 compression algorithm. In particular, we have used a python implementation of the JPEG2000 which is a part of the PIL package available at : {\url {https://pillow.readthedocs.io/en/3.1.x/handbook/image-file-formats.html##jpeg-2000}}. The compression algorithm has multiple inputs. The first one is the image itself. The other parameter that is important in our implementation is the parameter ``quality-layer'', denoted by $q$ in our paper, that controls compression ratio (or equivalently the rate). Figure \ref{fig: different quality layers} shows the result of the compression-decompression for three different values of the parameter $q$. It is clear from this figure that as $q$ decreases, the distortion in the reconstructed image reduces. The value $q=0$ corresponds to the lossless compression. 

\begin{figure}
\centering
\begin{subfigure}[b]{.3 \textwidth}
\includegraphics[width = \textwidth]{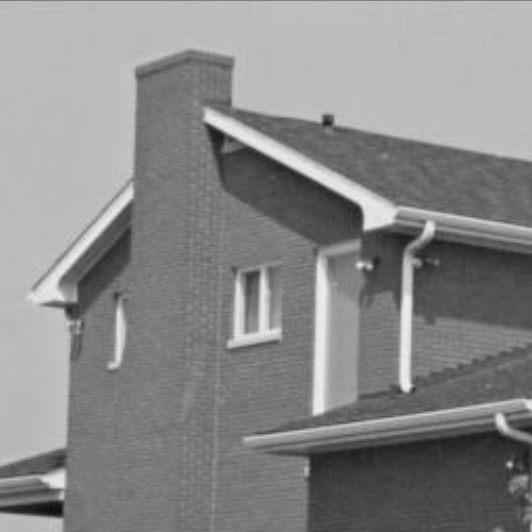}
\caption{q = 0}
\end{subfigure}
\begin{subfigure}[b]{.3 \textwidth}
\includegraphics[width = \textwidth]{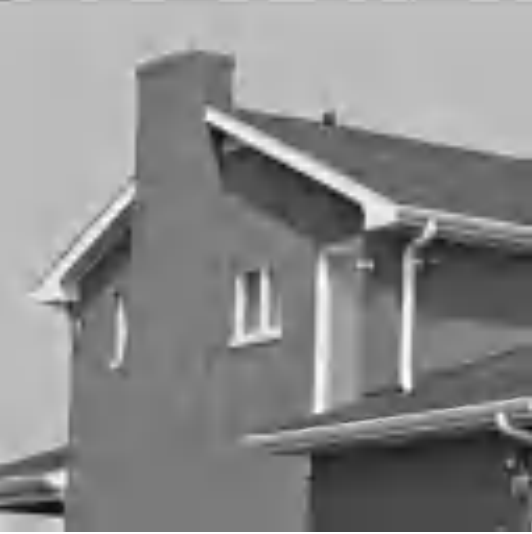}
\caption{q = 50}
\end{subfigure}
\begin{subfigure}[b]{.3 \textwidth}
\includegraphics[width = \textwidth]{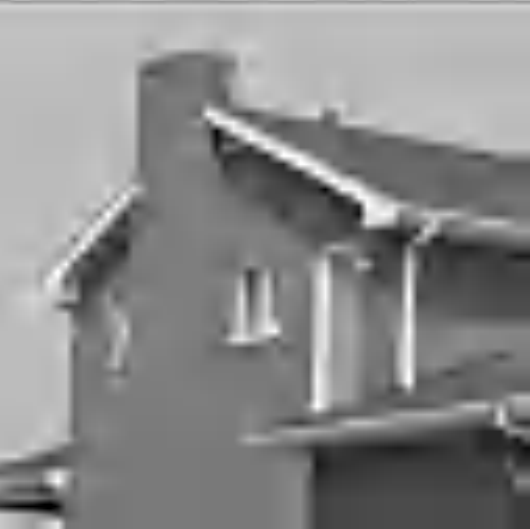}
\caption{q = 100}
\end{subfigure}
\caption{Compression with different quality-layers}
\label{fig: different quality layers}
\end{figure}

The GD-COPER algorithm has three different parameters that require tuning: (i) initialization, (ii) the quality parameter of the compression algorithm $q$ at every iteration, and  (iii) the step size $\mu$. As will be discussed in {Section \ref{sec:stability to initialization}}, our algorithm is not sensitive to the the initialization and in our simulation results, we start the algorithm with the white image. Hence, in this section, we only describe how we set the step-size $\mu_t$ and the quality parameter $q_t$ at every iteration. The problem of parameter tuning for iterative algorithms is a challenging problem that has not been settled properly yet \cite{parameter_tuning}. Hence, after doing multiple runs of the algorithm we have found a set of parameters that work well in practice. Below we summarize the chosen parameters for the Gaussian and coded diffraction patterns. We should emphasize that better tuning are expected to improve the performance of GD-COPER. Below we discuss our choice of parameters for the Gaussian and coded diffraction patterns separately.  

\begin{itemize}

\item Gaussian matrices:
The ``quality-layer" and step-size parameters at step $t$ are set in the following way:

\begin{equation} \label{eq:rate and stepsize setting}
\begin{dcases}
q_t = 40, \quad \mu_t = .2 \times \frac{\enVert{\mvec{z}_t}}{\enVert{\nabla d_A(\mvec{z}_t)}} \qquad & 1 \leq t \leq 10
\\
q_t = 0, \quad \mu_t = .02 \times \frac{\enVert{\mvec{z}_t}}{\enVert{\nabla d_A(\mvec{z}_t)}} \qquad & t \geq 11
\end{dcases}
\end{equation}
Note that $q_t = 0$ means that the algorithm employs an almost-lossless compression. The main reason an almost-lossless compression is used in the final iteration is that we are considering noiseless observations. We run the GD-COPER for 50 iterations, since the error does not decrease much after that.

\item Coded diffraction patterns: The value of parameters we chose for the coded-diffraction patterns is somewhat different from the ones we chose for Gaussian matrices. For such matrices, we adopt the following parameters:

\begin{equation} \label{eq:DCT rate and stepsize setting}
\mu_t = \max \intoo{{\rm e}^{0.7 - 0.41 t},  0.02} \times \frac{\enVert{\mvec{z}_t}}{\enVert{\nabla d_A(\mvec{z}_t)}},
\qquad
\begin{dcases}
q_t =50 & 1 \leq t \leq 5
\\
q_t =20 & 6 \leq t \leq 30
\\
q_t =0 & t > 30
\end{dcases}.
\end{equation}
\end{itemize}

We run the GD-COPER for 50 iterations, since the error remains almost the same for further iterations. Again these parameters are obtained by comparing a number of choices and choosing the one that seems to perform well on a wide range of images and problem instances. 

\subsubsection{Setting the parameters of Wirtinger flow}

The following parameters of the Wirtinger flow require tuning: (i) initialization, (ii) step size. Most of the papers, including \cite{Wirtinger Flow} suggest using the spectral method for the initialization of the algorithm. Our simulation results, some of which are reported in Section \ref{sec: spectral}, show that the algorithm works better when it is initialized with the white image. Hence, in all our simulations, except the ones in Sections \ref{sec: spectral} and \ref{sec:stability to initialization}, we initialize the algorithm with a white image. 

For setting the step size, we follow the suggestions of \cite{Wirtinger Flow}, and adopt the following policy:
\begin{equation} \label{eq: WF step-size}
\mu_t = \min \intoo{1 - {\rm e}^{- \frac{t}{t_o}}, \mu_{\max}}, 
\end{equation}
where $t_o = 330, \; \mu_{\max} = 0.4$.  Moreover, 300 iterations are used in all runs of Wirtinger Flow ( this is the number which is suggested in the simulations of \cite{Wirtinger Flow}) except for the cases that due to the divergence of algorithm the machine terminates the run earlier. Divergence happens when the norm of $\mvec{z}_t$ goes to infinity.

\subsection{Results}

We present our results for Gaussian and coded diffraction patterns separately.

\subsubsection{Gaussian measurement matrices}

\begin{table}[h!]
\centering
\caption{Results for the Gaussian measurement matrices. Both the GD-COPER algorithm and the Wirtinger flow are initialized with a white image. The setting of all the other parameters is described in Section \ref{ssec:parametersetting}. The notation DVG in the table refers to the fact that the algorithm either stops since the norm of $z$ diverges to infinity, or returns a result with negative PSNR.}
\begin{tabular}[t]{| c | *{3}{c|} | *{2}{c|}}
\hline
\multirow{2}{*}{Target} & \multirow{2}{*}{$\frac{m}{n}$}
 & \multicolumn{2}{c||}{GD-COPER} & \multicolumn{2}{c|}{Wirtinger Flow}
\\  \cline{3-6}
& & PSNR & \scriptsize{Run time} & PSNR & \scriptsize{Run time}
\\ \hline
\multirow{4}{*}{\includegraphics[width=.1 \textwidth]{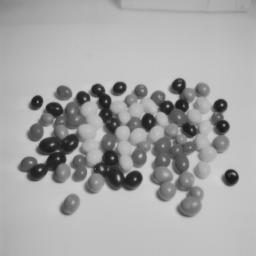}}  
& 0.5 & 23.22 & 11.2 & DVG & 8.68
\\ \cline{2-6}
 & 0.73 & 24.44 & 15.2 & DVG & 15.2
\\ \cline{2-6}
 & 1.0 & 25.63 & 18.9 & DVG & 30.6
\\ \cline{2-6}
 & 2.0 & 31.79 & 29.3 & DVG & 106.
\\  \hline \hline
\multirow{4}{*}{\includegraphics[width=.1 \textwidth]{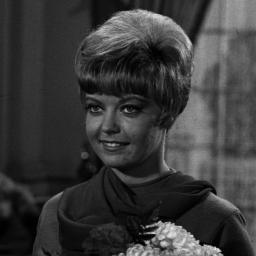}} 
 & 0.5 & 22.58 & 13.1 & 4.83 & 39.3
\\ \cline{2-6}
 & 0.73 & 24.79 & 15.6 & 6.5 & 60.3
\\ \cline{2-6}
 & 1.0 & 26.43 & 17.9 & 8.68 & 79.6
\\ \cline{2-6}
 & 2.0 & 31.91 & 31.3 & 17.71 & 135.
\\  \hline \hline
\multirow{4}{*}{\includegraphics[width=.1 \textwidth]{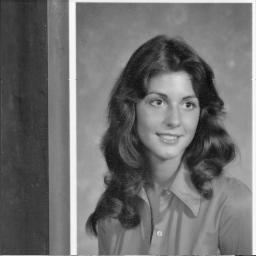}} 
 & 0.5 & 21.42 & 11.9 & DVG & 13.4
\\ \cline{2-6}
 & 0.73 & 23.73 & 15.2 & DVG & 33.1
\\ \cline{2-6}
 & 1.0 & 25.84 & 18.8 & 10.94 & 82.8
\\ \cline{2-6}
 & 2.0 & 32.36 & 30.1 & 29.66 & 136.
\\  \hline \hline
\multirow{4}{*}{\includegraphics[width=.1 \textwidth]{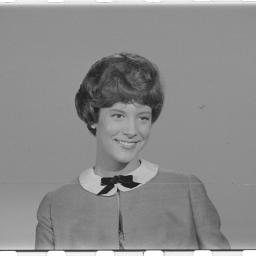}} 
 & 0.5 & 25.5 & 12.2 & DVG & 14.2
\\ \cline{2-6}
 & 0.73 & 27.43 & 13.9 & DVG & 22.1
\\ \cline{2-6}
 & 1.0 & 29.15 & 18.3 & DVG & 41.7
\\ \cline{2-6}
 & 2.0 & 34.76 & 29.6 & 33.36 & 140.
\\  \hline \hline
\multirow{4}{*}{\includegraphics[width=.1 \textwidth]{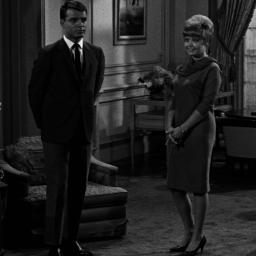}} 
  & 0.5 & 22.03 & 12.4 & 3.92 & 43.1
\\ \cline{2-6}
 & 0.73 & 24.08 & 15.1 & 5.68 & 59.0
\\ \cline{2-6}
 & 1.0 & 26.67 & 17.4 & 7.94 & 74.2
\\ \cline{2-6}
 & 2.0 & 33.07 & 28.4 & 14.35 & 143.
\\  \hline \hline
\multirow{4}{*}{\includegraphics[width=.1 \textwidth]{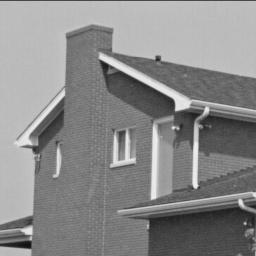}} 
& 0.5 & 21.83 & 11.2 & DVG & 7.64
\\ \cline{2-6}
 & 0.73 & 23.35 & 15.7 & DVG & 20.7
\\ \cline{2-6}
 & 1.0 & 24.52 & 19.9 & DVG & 34.1
\\ \cline{2-6}
 & 2.0 & 32.67 & 28.8 & 35.65 & 135.
\\  \hline \hline
\multirow{4}{*}{\includegraphics[width=.1 \textwidth]{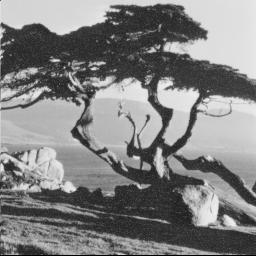}} 
 & 0.5 & 17.49 & 10.9 & DVG & 10.9
\\ \cline{2-6}
 & 0.73 & 18.68 & 14.4 & DVG & 20.9
\\ \cline{2-6}
 & 1.0 & 21.44 & 19.0 & DVG & 37.8
\\ \cline{2-6}
 & 2.0 & 29.04 & 29.8 & 32.74 & 140.
\\  \hline \hline
\end{tabular}
 \label{table:Gaussion results}
\end{table}

In our simulations, we consider seven different images shown in the first column of Table \ref{table:Gaussion results}.
 All these images are chosen from {``The Miscellaneous volume data-set''}, which is publicly available at {\url http://sipi.usc.edu/database/database.php?volume=misc}.
Since the images are colored we have extracted the luminance of the image and all the simulations are performed on gray-scale images. To reduce the computational complexity of our algorithm (in the case of Gaussian measurements only) we downsample images to reduce their size to $128 \times 128$. This size reduction helps us avoid the issues we face in storing i.i.d.~ Gaussian matrices. However, it also reduces the structures that exist in an image. Hence, JPEG2000 loses some of its efficiency. Hence, we expect the GD-COPER to perform better as the image size increases. This will become clearer when we work with larger images in the coded diffraction pattern simulations. 

After the downsampling, the signals' dimensions are $n=16384$. In Table \ref{table:Gaussion results}, we have considered $m = 32786, 16384$, $12000,$ and $8192$. Note that, in most of  these systems, not only the measurements are phaseless, but also they are undersampled.  

In each setup, we compare the performance of our algorithm with that of the Wirtinger flow. In addition to comparing the quality of the reconstruction via evaluating the peak-signal-to-noise-ratio (PSNR),\footnote{PSNR is defined as \begin{equation*} \label{def:PSNR}
{\rm PSNR} = 20 \log_{10} \intoo{\frac{255}{\sqrt{ \rm MSE}}},
\end{equation*}
where ${\rm MSE}$ is the mean squared error obtained from the last iteration of the algorithm.} we report the run time of the algorithms as well.   
 The run times of the algorithm are measured on a laptop computer with 2.8 GHz Intel Core i7 processor and 16 GB RAM. We can draw the following conclusions from the results reported in Table \ref{table:Gaussion results}:

\begin{enumerate}

\item[(i)] As expected, the Wirtinger flow does not do well when $\frac{m}{n}\leq1$. This is in contrast to the performance of GD-COPER that can obtain reasonable estimates even for $m/n \leq 1$. Note that in some cases, the Wirtinger flow can do as well as GD-COPER when $\frac{m}{n} = 2$. This happens because we have downsampled the images to $128 \times 128$ size, and hence we have removed some structures that JPEG2000 could otherwise  employ. In other words, JPEG2000 cannot efficiently reduce the size of such images, and hence GD-COPER is not capable of employing the structures of such images either. In the next section, GD-COPER works with large images (we can do this with coded diffraction patterns), and will outperform the Wirtinger flow with a larger margin. See Figure \ref{fig:final steps} for a visual comparison of the GD-COPER and Wirtinger flow algorithms.

\begin{figure}[h!]
\begin{subfigure}[b]{.19 \textwidth}
\includegraphics[width= \textwidth]{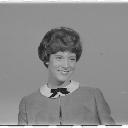}
\caption{Original Image}
\end{subfigure}
\begin{subfigure}[b]{.19 \textwidth}
\includegraphics[width= \textwidth]{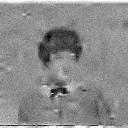}
\caption{$\frac{m}{n} = .5$}
\end{subfigure}
\begin{subfigure}[b]{.19 \textwidth}
\includegraphics[width= \textwidth]{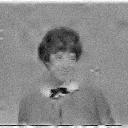}
\caption{$\frac{m}{n} = .73$}
\end{subfigure}
\begin{subfigure}[b]{.19 \textwidth}
\includegraphics[width= \textwidth]{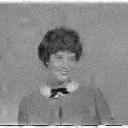}
\caption{$\frac{m}{n} = 1$}
\end{subfigure}
\begin{subfigure}[b]{.19 \textwidth}
\includegraphics[width= \textwidth]{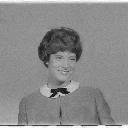}
\caption{$\frac{m}{n} = 2$}
\end{subfigure}

\begin{subfigure}[b]{.19 \textwidth}
\includegraphics[width= \textwidth]{Girl_3_128.jpg}
\caption{Original Image}
\end{subfigure}
\begin{subfigure}[b]{.19 \textwidth}
\includegraphics[width= \textwidth]{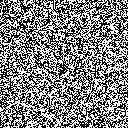}
\caption{$\frac{m}{n} = .5$}
\end{subfigure}
\begin{subfigure}[b]{.19 \textwidth}
\includegraphics[width= \textwidth]{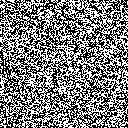}
\caption{$\frac{m}{n} = .73$}
\end{subfigure}
\begin{subfigure}[b]{.19 \textwidth}
\includegraphics[width= \textwidth]{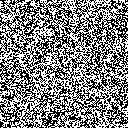}
\caption{$\frac{m}{n} = 1$}
\end{subfigure}
\begin{subfigure}[b]{.19 \textwidth}
\includegraphics[width= \textwidth]{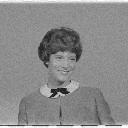}
\caption{$\frac{m}{n} = 2$}
\end{subfigure}
\caption{First row: outcomes of GD-COPER for different values of $m/n$. Second row: outcomes of Wirtinger Flow for different values of $m/n$. The original image is shown in the left column. The measurement matrix is Gaussian.  } \label{fig:final steps}
\end{figure}

\item [(ii)] GD-COPER is faster than the Wirtinger flow. Note that each iteration of GD-COPER is computationally more demanding than that of the Wirtinger flow. However, GD-COPER requires less steps to obtain a good estimate of the signal. {Figure \ref{figure:normalized error} compares the normalized MSE (we have normalized the mean square error, by the energy of the underlying signal) of GD-COPER as a function of the iteration number with that of the Wirtinger flow. We can see that GD-COPER converges in 10 iteration, while Wirtinger flow requires around 200 iterations to get to a comparable error if it does not diverge.  }

 \begin{figure}[h!]
\begin{subfigure}[b]{.5 \textwidth}
\includegraphics[width =  \textwidth]{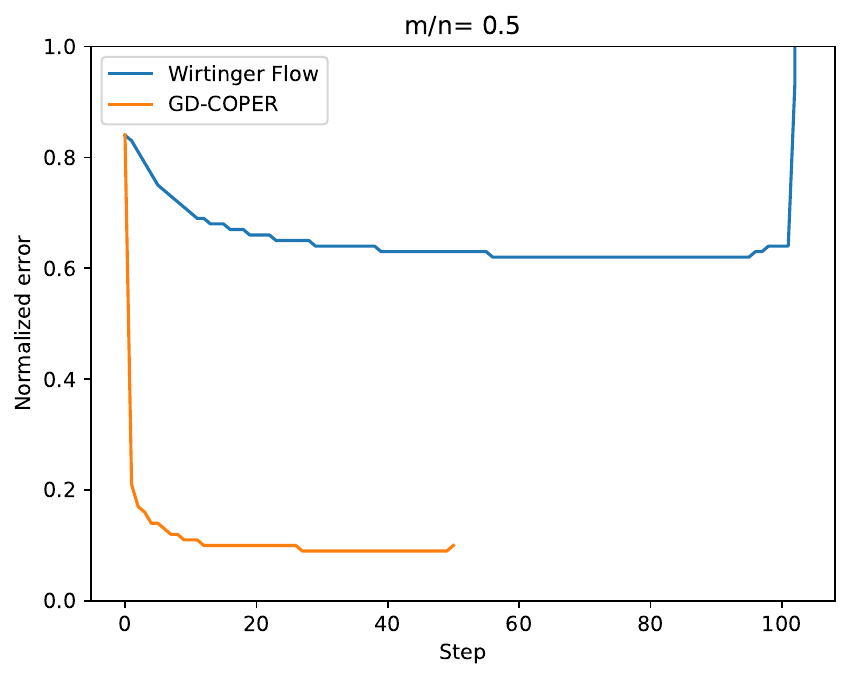}
\end{subfigure}
\begin{subfigure}[b]{.5 \textwidth}
\includegraphics[width =  \textwidth]{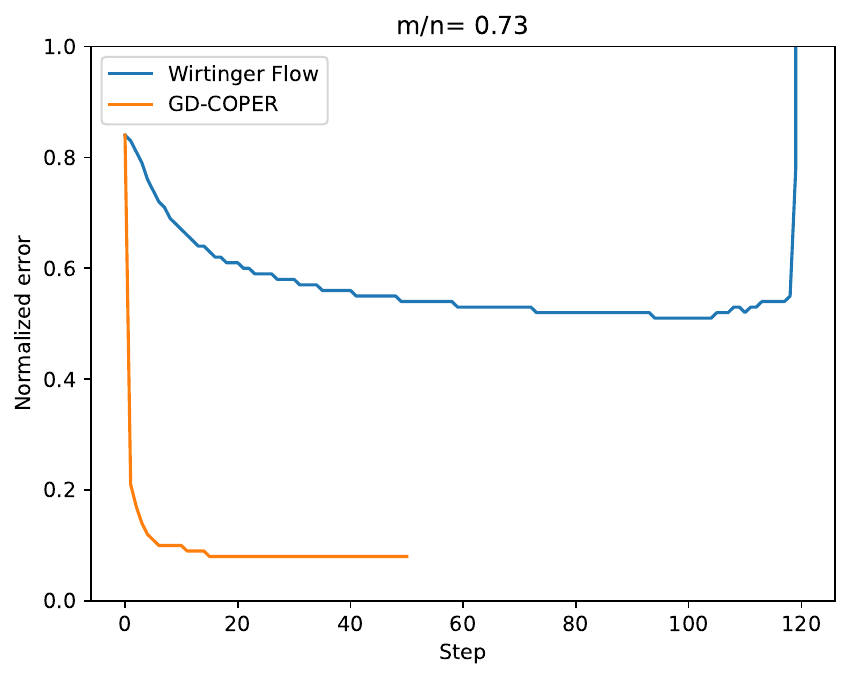}
\end{subfigure}
\begin{subfigure}[b]{.5 \textwidth}
\includegraphics[width =  \textwidth]{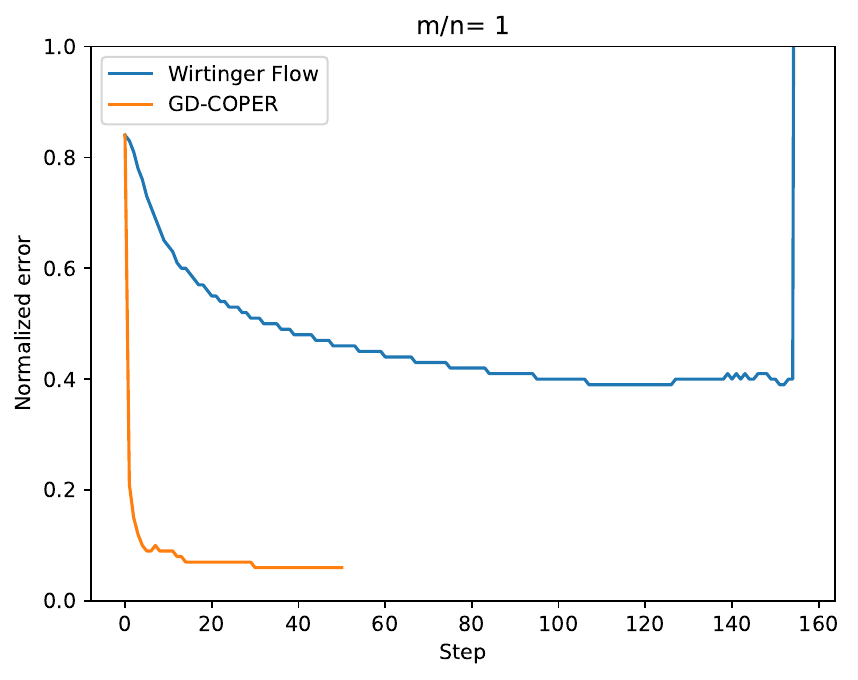}
\end{subfigure}
\begin{subfigure}[b]{.5 \textwidth}
\includegraphics[width =  \textwidth]{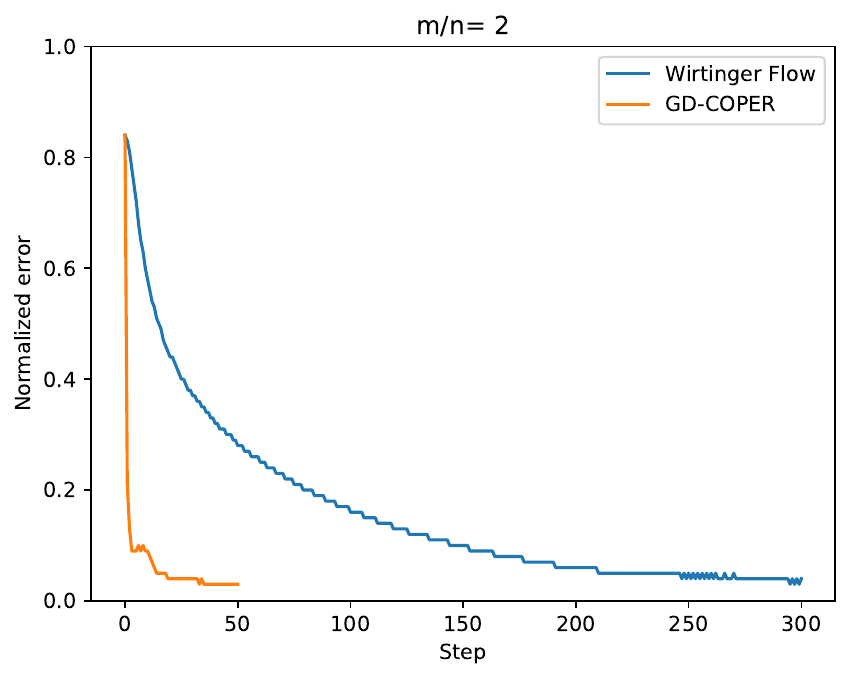}
\end{subfigure}

\caption{Normalized mean square error as a function of the iteration number for four different values of $m/n$. The original image is the same as the one chosen for Figure \ref{fig:final steps}.}
 \label{figure:normalized error}
\end{figure}

\end{enumerate}

%\begin{figure}[h!]
%\centering
%\begin{subfigure}[b]{.2 \textwidth}
%\includegraphics[width = \textwidth]{original_Balls.jpg}
%\caption{Balls}
%\end{subfigure}
%\begin{subfigure}[b]{.2 \textwidth}
%\includegraphics[width = \textwidth]{original_House.jpg}
%\caption{House}
%\end{subfigure}
%\begin{subfigure}[b]{.2 \textwidth}
%\includegraphics[width = \textwidth]{original_Girl_2.jpg}
%\caption{Girl 2}
%\end{subfigure}
%\begin{subfigure}[b]{.2 \textwidth}
%\includegraphics[width = \textwidth]{original_Girl_3.jpg}
%\caption{Girl 3}
%\end{subfigure}
%\caption{Signal images} \label{signal images}
%\end{figure}

\subsubsection{Coded diffraction model}
\label{sec:Coded diff model}
In this section, we evaluate the performance of our algorithm on the more practical coded-diffraction measurements. Again, we work with the seven images we introduced in the last section. However, given the fact that in the case of the coded diffraction patterns the measurement matrix is not explicitly stored we will use images in their original sizes, $256 \times 256$. We compare the performance of GD-COPER with that of the Wirtinger flow for different $m/n$ ratios. Tables \ref{table:dct results_1} and  \ref{table:dct results_2} summarize our simulation results. Again we should emaphasize that both the GD-COPER and Wirtinger flow are initialized with an all-white image. We can draw the following conclusion from Tables \ref{table:dct results_1} and  \ref{table:dct results_2}:

\begin{enumerate}
\item Again for lower values of the sampling ratio $m/n$ such as $m/n\leq 5$ Wirtinger flow is not capable of finding a good estimate. However, GD-COPER obtains an accurate estimate for $m/n\leq 5$, and even for $m/n=1$. If we compare these simulations with the ones we had for Gaussian matrices, it seems to be the case that the discrepancy between the performance of the Wirtinger flow and GD-COPER has increased in the coded-diffraction simulation. Part of this is a result of the fact that our simulations have been performed on larger images for which JPEG2000 is more efficient.  

\item As we increase the number of masks, usually after $L=10$ the performance of GD-COPER saturates, while Wirtinger flow continues to improve. There are two effects that cause the saturation of the GD-COPER: (i) Given that the compression is applied at every iteration, even though it is in its loss-less mode, it still imposes some quantization to the estimates. (ii) Suboptimal parameter tuning. We believe that the performance saturation of the algorithm does not cause a major issue in practice since it happens at very high values of PSNR, e.g. $40$ dB. However, as a result of the saturation, we see that in most cases, when $m/n> 15$, then the Wirtinger flow outperforms GD-COPER (for the sake of brevity we have not included the results of $m/n> 15$ in our tables). Hence, if extremely accurate estimates of the signal are required (e.g. PSNR= $50$ dB) and we have enough masks, then the Wirtinger flow should be preferred over GD-COPER.

\item The computational complexity of GD-COPER is comparable with that of the Wirtinger flow. Note that each iteration of GD-COPER is computationally more demanding than that of the Wirtinger flow. However, GD-COPER requires less steps to obtain a good estimate of the signal.

\end{enumerate}

\begin{table}[h!]
\centering
\caption{Comparison between the Wirtinger flow and the GD-COPER with the coded diffraction patterns for different values of $m/n$. The true images of the simulations are shown in the first column. }
\begin{tabular}{| c | *{3}{c|} | *{2}{c|}}
\hline
\multirow{2}{*}{Target} & \multirow{2}{*}{$\frac{m}{n}$}
 & \multicolumn{2}{c||}{GD-COPER} & \multicolumn{2}{c|}{Wirtinger Flow}
\\  \cline{3-6}
& & PSNR & \scriptsize{Run time} & PSNR & \scriptsize{Run time}
\\ \hline
\multirow{14}{*}{\includegraphics[width=.1 \textwidth]{original_Balls.jpg}} 
 & 1 & 27.8 & 13.0 & DVG & 1.2
\\ \cline{2-6}
 & 2 & 34.7 & 16.9 & DVG & 2.0
\\ \cline{2-6}
 & 3 & 36.2 & 18.1 & DVG & 2.7
\\ \cline{2-6}
 & 4 & 39.7 & 19.8 & DVG & 4.2
\\ \cline{2-6}
 & 5 & 42.1 & 14.2 & DVG & 4.1
\\ \cline{2-6}
 & 6 & 38.5 & 14.6 & DVG & 4.3
\\ \cline{2-6}
 & 7 & 42.7 & 15.4 & DVG & 5.7
\\ \cline{2-6}
 & 8 & 44.5 & 15.6 & DVG & 6.2
\\ \cline{2-6}
 & 9 & 38.9 & 16.1 & 23.6 & 18.9
\\ \cline{2-6}
 & 10 & 49.1 & 15.1 & 17.8 & 12.7
\\ \cline{2-6}
 & 15 & 38.6 & 17.3 & 13.0 & 23.0
\\ \hline
\multirow{14}{*}{\includegraphics[width=.1 \textwidth]{original_Couple.jpg}} 
 & 1 & 19.4 & 14.3 & 4.1 & 2.8
\\ \cline{2-6}
 & 2 & 28.6 & 19.5 & 7.2 & 5.1
\\ \cline{2-6}
 & 3 & 33.4 & 17.6 & 10.1 & 7.4
\\ \cline{2-6}
 & 4 & 34.5 & 14.4 & 13.1 & 5.9
\\ \cline{2-6}
 & 5 & 39.0 & 14.9 & 16.2 & 7.4
\\ \cline{2-6}
 & 6 & 40.2 & 15.0 & 18.9 & 8.0
\\ \cline{2-6}
 & 7 & 44.0 & 14.8 & 22.4 & 9.1
\\ \cline{2-6}
 & 8 & 45.9 & 15.3 & 25.2 & 10.0
\\ \cline{2-6}
 & 9 & 45.6 & 15.1 & 28.0 & 11.4
\\ \cline{2-6}
 & 10 & 47.4 & 15.7 & 31.8 & 12.9
\\ \cline{2-6}
 & 15 & 50.9 & 19.6 & 44.1 & 29.8
\\ \hline
\end{tabular}
\begin{tabular}{| c | *{3}{c|} | *{2}{c|}}
\hline
\multirow{2}{*}{Target} & \multirow{2}{*}{$\frac{m}{n}$}
 & \multicolumn{2}{c||}{GD-COPER} & \multicolumn{2}{c|}{Wirtinger Flow}
\\  \cline{3-6}
& & PSNR & \scriptsize{Run time} & PSNR & \scriptsize{Run time}
\\ \hline
\multirow{14}{*}{\includegraphics[width=.1 \textwidth]{original_girl_1.jpg}} 
& 1 & 22.2 & 14.2 & 5.1 & 2.9
\\ \cline{2-6}
 & 2 & 30.8 & 18.0 & 8.1 & 5.5
\\ \cline{2-6}
 & 3 & 34.8 & 16.9 & 11.1 & 6.3
\\ \cline{2-6}
 & 4 & 38.0 & 15.7 & 14.1 & 6.3
\\ \cline{2-6}
 & 5 & 38.9 & 18.1 & 17.2 & 8.0
\\ \cline{2-6}
 & 6 & 43.3 & 14.8 & 19.9 & 7.9
\\ \cline{2-6}
 & 7 & 43.3 & 15.0 & 23.5 & 9.1
\\ \cline{2-6}
 & 8 & 47.4 & 15.1 & 26.5 & 10.1
\\ \cline{2-6}
 & 9 & 45.3 & 16.3 & 29.4 & 13.4
\\ \cline{2-6}
 & 10 & 45.8 & 15.7 & 33.0 & 12.8
\\ \cline{2-6}
 & 15 & 49.3 & 17.1 & 47.8 & 22.8
\\ \hline
\multirow{14}{*}{\includegraphics[width=.1 \textwidth]{original_Girl_2.jpg}} 
 & 1 & 25.2 & 13.8 & DVG & 1.4
\\ \cline{2-6}
 & 2 & 32.5 & 18.1 & DVG & 3.2
\\ \cline{2-6}
 & 3 & 35.9 & 17.2 & 13.4 & 6.8
\\ \cline{2-6}
 & 4 & 38.7 & 14.8 & 16.4 & 5.8
\\ \cline{2-6}
 & 5 & 38.2 & 15.8 & 19.5 & 7.8
\\ \cline{2-6}
 & 6 & 43.1 & 14.9 & 22.4 & 8.1
\\ \cline{2-6}
 & 7 & 45.1 & 15.0 & 25.8 & 9.2
\\ \cline{2-6}
 & 8 & 42.9 & 15.2 & 28.9 & 10.0
\\ \cline{2-6}
 & 9 & 44.7 & 15.3 & 31.7 & 11.8
\\ \cline{2-6}
 & 10 & 50.0 & 16.0 & 35.2 & 13.1
\\ \cline{2-6}
 & 15 & 43.9 & 17.4 & 50.8 & 24.5
\\ \hline
\end{tabular}
 \label{table:dct results_1}
\end{table}

\begin{table}[h!]
\centering
\caption{Comparison between the Wirtinger flow and the GD-COPER with the coded diffraction patterns for different values of $m/n$. The true images in the simulations are shown in the leftmost column. }
\begin{tabular}[t]{| c | *{3}{c|} | *{2}{c|}}
\hline
\multirow{2}{*}{Target} & \multirow{2}{*}{$\frac{m}{n}$}
 & \multicolumn{2}{c||}{GD-COPER} & \multicolumn{2}{c|}{Wirtinger Flow}
\\  \cline{3-6}
& & PSNR & \scriptsize{Run time} & PSNR & \scriptsize{Run time}
\\ \hline
\multirow{14}{*}{\includegraphics[width=.1 \textwidth]{original_Girl_3.jpg}}
& 1 & 29.3 & 14.2 & DVG & 1.4
\\ \cline{2-6}
 & 2 & 34.0 & 18.5 & DVG & 2.7
\\ \cline{2-6}
 & 3 & 36.8 & 17.6 & DVG & 4.7
\\ \cline{2-6}
 & 4 & 38.0 & 15.1 & 17.6 & 6.1
\\ \cline{2-6}
 & 5 & 40.7 & 15.9 & 20.4 & 8.0
\\ \cline{2-6}
 & 6 & 44.2 & 14.6 & 22.8 & 8.2
\\ \cline{2-6}
 & 7 & 42.1 & 14.9 & 28.1 & 9.2
\\ \cline{2-6}
 & 8 & 40.7 & 16.2 & 30.5 & 9.8
\\ \cline{2-6}
 & 9 & 42.2 & 16.0 & 33.8 & 11.8
\\ \cline{2-6}
 & 10 & 49.9 & 16.3 & 37.6 & 13.1
\\ \cline{2-6}
 & 15 & 41.8 & 16.9 & 52.1 & 23.0
\\  \hline
\multirow{14}{*}{\includegraphics[width=.1 \textwidth]{original_House.jpg}} 
 & 1 & 27.2 & 13.8 & DVG & 1.4
\\ \cline{2-6}
 & 2 & 32.3 & 16.6 & DVG & 2.1
\\ \cline{2-6}
 & 3 & 35.8 & 16.9 & DVG & 3.9
\\ \cline{2-6}
 & 4 & 36.4 & 17.1 & 15.6 & 7.3
\\ \cline{2-6}
 & 5 & 38.7 & 15.1 & 17.9 & 6.7
\\ \cline{2-6}
 & 6 & 39.4 & 14.9 & 20.1 & 7.8
\\ \cline{2-6}
 & 7 & 42.9 & 15.1 & 27.1 & 8.9
\\ \cline{2-6}
 & 8 & 47.5 & 15.6 & 30.5 & 10.1
\\ \cline{2-6}
 & 9 & 40.9 & 18.2 & 33.2 & 18.9
\\ \cline{2-6}
 & 10 & 47.6 & 15.9 & 36.9 & 13.0
\\ \cline{2-6}
 & 15 & 48.6 & 17.7 & 53.3 & 23.0
\\  \hline
\end{tabular}
\begin{tabular}[t]{| c | *{3}{c|} | *{2}{c|}}
\hline
\multirow{2}{*}{Target} & \multirow{2}{*}{$\frac{m}{n}$}
 & \multicolumn{2}{c||}{GD-COPER} & \multicolumn{2}{c|}{Wirtinger Flow}
\\  \cline{3-6}
& & PSNR & \scriptsize{Run time} & PSNR & \scriptsize{Run time}
\\ \hline
\multirow{14}{*}{\includegraphics[width=.1 \textwidth]{original_tree.jpg}} 
 & 1 & 23.1 & 13.8 & DVG & 1.4
\\ \cline{2-6}
 & 2 & 28.0 & 17.9 & DVG & 2.6
\\ \cline{2-6}
 & 3 & 32.0 & 17.9 & DVG & 3.7
\\ \cline{2-6}
 & 4 & 34.3 & 18.4 & 16.2 & 7.1
\\ \cline{2-6}
 & 5 & 38.1 & 16.9 & 19.0 & 9.2
\\ \cline{2-6}
 & 6 & 38.5 & 15.3 & 21.2 & 8.0
\\ \cline{2-6}
 & 7 & 42.0 & 15.2 & 22.6 & 9.2
\\ \cline{2-6}
 & 8 & 44.6 & 15.6 & 29.2 & 9.9
\\ \cline{2-6}
 & 9 & 43.2 & 19.4 & 32.3 & 14.7
\\ \cline{2-6}
 & 10 & 43.0 & 16.3 & 36.1 & 13.1
\\ \cline{2-6}
 & 15 & 52.0 & 17.5 & 49.7 & 23.4
\\  \hline
\end{tabular}
 \label{table:dct results_2}
\end{table}

\subsection{Robustness of GD-COPER with respect to initialization}\label{sec:stability to initialization}
As we discussed in Section \ref{ssec:parametersetting}, the performance of GD-COPER is not sensitive to the initialization. In this section, we present some of our evidence that supports this claim. Given that our simulation results are similar for both coded diffraction patterns and Gaussian measurements, we only report our simulations for the coded diffraction patterns. In order to observe the impact of initialization we considered the following initialization: Let $\mvec{x}$ denote the underlying signal we want to recover, and let $\mvec{x}_o$ denote the vector that corresponds to an all-white image. A simple initialization that we can use in practice is $\mvec{x}_o$, while the best oracle-initialization is $\mvec{x}$. Hence, we can consider the family of initializations
\[
\mvec{x}_{\rm init} = \lambda \mvec{x}_o + (1-\lambda) \mvec{x},
\]
for $\lambda \in \{0, 0.1, 0.2, \ldots, 0.9, 1\}$. We expect the smaller values of $\lambda$ to return better initializations. Tables \ref{tab:testinit1} and \ref{tab:testinit2} evaluate the performance of GD-COPER for different initializations and different images. The other parameters of GD-COPER are set according to the strategy described in Section \ref{ssec:parametersetting}.  As is clear from our simulation results, the initialization schemes have much larger impacts on the Wirtinger flow compared to GD-COPER. In fact,  the GD-COPER is not very sensitive to the choice of initialization and in most cases, the difference between the best initialization and worst initialization is less than $2$ dB.  In contrast to GD-COPER, the performance of the Wirtinger flow is very sensitive to the choice of the initialization. For this reason, the spectral method is often used for the initialization of the Wirtinger flow algorithm. In the next section, we will show that the initialization of the Wirtinger flow algorithm with an all-white image is often better than the spectral initialization. However, we should emphasize that this phenomenon is only true for the real-valued signals, and has not been tested on complex-valued signals. 

\subsection{Spectral Initialization}
\label{sec: spectral}
Another claim we made in Section \ref{ssec:parametersetting} regarding the initialization was the fact that Spectral initialization does not seem to help the Wirtinger flow beyond what is offered by an all-white image initialization. We show part of our evidence regarding this claim. Tables \ref{table: WF spectral and white 1} - \ref{table: WF spectral and white 3} summarize some of our findings.  In these table the `n-init-err' shows the normalized mean square error of the initialization.  Note that in most cases the spectral methods does not offer a closer point than the all-white image except when we have $\frac{m}{n} \geq 7$.  Moreover, when we have many observations and the initial point offered by the Spectral method is closer than the white image, Wirtinger Flow usually performs better starting from the white image. This shows the initial distance is not the only important factor to the convergence of Wirtinger flow (this is an artifact of the fixed parameter tuning that has been proposed for the Wirtinger flow).  For instance, if the norm of the gradient at the starting steps, when the step-size defined in \eqref{eq: WF step-size} is large, remains high, then the algorithm may diverges.

\begin{table}
\centering
\caption{Wirtinger Flow performance with spectral and all-white initialization} \label{table: WF spectral and white 1}
\begin{tabular}[t]{| c | *{4}{c|} | *{3}{c|}}
\hline
\multirow{2}{*}{Target} & \multirow{2}{*}{$\frac{m}{n}$}
 & \multicolumn{3}{c||}{All-white} & \multicolumn{3}{c|}{Spectral}
\\  \cline{3-8}
& & {n-init-err} & PSNR & Run time & {n-init-err} & PSNR & Run time
\\ \hline
\multirow{11}{*}{\includegraphics[width=.2 \textwidth]{original_Balls.jpg}} 
  & 1 & 0.57&DVG & 1.7 & 1.39 & DVG & 2.4
\\ \cline{2-8}
 & 2 & 0.57&DVG & 1.4 & 1.39 & DVG & 4.4
\\ \cline{2-8}
 & 3 & 0.57&17.1 & 4.3 & 1.39 & DVG & 6.2
\\ \cline{2-8}
 & 4 & 0.57&20.3 & 5.5 & 1.37 & DVG & 7.4
\\ \cline{2-8}
 & 5 & 0.57&23.2 & 6.5 & 1.37 & DVG & 9.3
\\ \cline{2-8}
 & 6 & 0.57&26.9 & 7.7 & 1.38 & DVG & 11.2
\\ \cline{2-8}
 & 7 & 0.57&29.4 & 9.8 & 1.13 & DVG & 12.3
\\ \cline{2-8}
 & 8 & 0.57&32.8 & 13.5 & 0.89 & DVG & 16.2
\\ \cline{2-8}
 & 9 & 0.57&36.2 & 17.5 & 0.63 & 9.4 & 74.8
\\ \cline{2-8}
 & 10 & 0.57&39.0 & 44.5 & 0.64 & DVG & 32.5
\\ \cline{2-8}
 & 15 & 0.57&51.3 & 50.6 & 0.49 & 20.5 & 99.4
\\ \hline
\end{tabular}

\begin{tabular}[t]{| c | *{4}{c|} | *{3}{c|}}
\hline
\multirow{2}{*}{Target} & \multirow{2}{*}{$\frac{m}{n}$}
 & \multicolumn{3}{c||}{All-white} & \multicolumn{3}{c|}{Spectral}
\\  \cline{3-8}
& & {n-init-err} & PSNR & {Run time} & {n-init-err} & PSNR & {Run time}
\\ \hline
\multirow{11}{*}{\includegraphics[width=.2 \textwidth]{original_House.jpg}} 
  & 1 & 0.86&DVG & 2.8 & 1.39 & DVG & 4.8
\\ \cline{2-8}
 & 2 & 0.86&12.1 & 9.1 & 1.39 & DVG & 8.4
\\ \cline{2-8}
 & 3 & 0.86&15.1 & 11.4 & 1.39 & DVG & 10.9
\\ \cline{2-8}
 & 4 & 0.86&18.2 & 14.9 & 1.39 & DVG & 13.6
\\ \cline{2-8}
 & 5 & 0.86&21.1 & 20.0 & 1.41 & DVG & 15.1
\\ \cline{2-8}
 & 6 & 0.86&24.2 & 25.1 & 1.37 & DVG & 17.1
\\ \cline{2-8}
 & 7 & 0.86&27.4 & 28.2 & 1.06 & DVG & 19.5
\\ \cline{2-8}
 & 8 & 0.86&30.4 & 28.4 & 0.9 & DVG & 22.8
\\ \cline{2-8}
 & 9 & 0.86&33.4 & 31.6 & 1.33 & DVG & 26.0
\\ \cline{2-8}
 & 10 & 0.86&35.5 & 32.2 & 0.6 & 24.4 & 56.1
\\ \cline{2-8}
 & 15 & 0.86&56.7 & 43.7 & 0.48 & 11.0 & 81.7
\\ \hline
\end{tabular}

\begin{tabular}{| c | *{4}{c|} | *{3}{c|}}
\hline
\multirow{2}{*}{Target} & \multirow{2}{*}{$\frac{m}{n}$}
 & \multicolumn{3}{c||}{All-white} & \multicolumn{3}{c|}{Spectral}
\\  \cline{3-8}
& & {n-init-err} & PSNR & Run time & {n-init-err} & PSNR & Run time
\\ \hline
\multirow{11}{*}{\includegraphics[width=.2 \textwidth]{original_tree.jpg}} 
  & 1 & 0.98&DVG & 2.6 & 1.39 & DVG & 5.0
\\ \cline{2-8}
 & 2 & 0.98&DVG & 2.8 & 1.39 & DVG & 7.2
\\ \cline{2-8}
 & 3 & 0.98&14.0 & 9.6 & 1.39 & DVG & 8.9
\\ \cline{2-8}
 & 4 & 0.98&17.0 & 11.9 & 1.4 & DVG & 10.6
\\ \cline{2-8}
 & 5 & 0.98&20.0 & 15.9 & 1.38 & DVG & 12.6
\\ \cline{2-8}
 & 6 & 0.98&23.2 & 17.7 & 1.21 & DVG & 15.0
\\ \cline{2-8}
 & 7 & 0.98&26.1 & 21.8 & 1.31 & DVG & 17.0
\\ \cline{2-8}
 & 8 & 0.98&29.0 & 24.1 & 1.39 & DVG & 17.9
\\ \cline{2-8}
 & 9 & 0.98&32.2 & 26.2 & 0.65 & 20.4 & 30.8
\\ \cline{2-8}
 & 10 & 0.98&34.7 & 13.6 & 0.6 & 21.3 & 30.9
\\ \cline{2-8}
 & 15 & 0.98&57.1 & 21.9 & 0.48 & 21.2 & 55.3
\\ \hline
\end{tabular}
\end{table}

\begin{table}
\centering
\caption{Wirtinger Flow performance with spectral and all-white initialization} \label{table: WF spectral and white 2}
\begin{tabular}{| c | *{4}{c|} | *{3}{c|}}
\hline
\multirow{2}{*}{Target} & \multirow{2}{*}{$\frac{m}{n}$}
 & \multicolumn{3}{c||}{All-white} & \multicolumn{3}{c|}{Spectral}
\\  \cline{3-8}
& & \scriptsize{n-init-err} & PSNR & Run time & \scriptsize{n-init-err} & PSNR & Run time
\\ \hline
\multirow{11}{*}{\includegraphics[width=.2 \textwidth]{original_girl_1.jpg}} 
  & 1 & 2.84 & 5.1 & 2.6 & 1.39 & DVG & 3.1
\\ \cline{2-8}
 & 2 & 2.84 & 8.1 & 3.6 & 1.4 & DVG & 4.9
\\ \cline{2-8}
 & 3 & 2.84 & 11.0 & 4.7 & 1.39 & DVG & 7.1
\\ \cline{2-8}
 & 4 & 2.84 & 14.1 & 5.5 & 1.4 & DVG & 8.8
\\ \cline{2-8}
 & 5 & 2.84 & 17.3 & 6.9 & 1.38 & DVG & 10.6
\\ \cline{2-8}
 & 6 & 2.84 & 20.3 & 8.3 & 1.36 & DVG & 12.3
\\ \cline{2-8}
 & 7 & 2.84 & 22.9 & 9.7 & 1.36 & DVG & 14.2
\\ \cline{2-8}
 & 8 & 2.84 & 26.2 & 11.1 & 1.39 & DVG & 16.2
\\ \cline{2-8}
 & 9 & 2.84 & 28.4 & 12.8 & 1.1 & DVG & 17.8
\\ \cline{2-8}
 & 10 & 2.84 & 32.3 & 13.4 & 0.6 & 20.8 & 30.8
\\ \cline{2-8}
 & 15 & 2.84 & 45.0 & 22.2 & 0.48 & 39.7 & 53.4
\\ \hline
\end{tabular}

\begin{tabular}{| c | *{4}{c|} | *{3}{c|}}
\hline
\multirow{2}{*}{Target} & \multirow{2}{*}{$\frac{m}{n}$}
 & \multicolumn{3}{c||}{All-white} & \multicolumn{3}{c|}{Spectral}
\\  \cline{3-8}
& & \scriptsize{n-init-err} & PSNR & Run time & \scriptsize{n-init-err} & PSNR & Run time
\\ \hline
\multirow{11}{*}{\includegraphics[width=.2 \textwidth]{original_Girl_3.jpg}} 
  & 1 & 0.83 & DVG & 1.4 & 1.39 & DVG & 3.1
\\ \cline{2-8}
 & 2 & 0.83 & 12.6 & 3.8 & 1.39 & DVG & 5.1
\\ \cline{2-8}
 & 3 & 0.83 & 15.6 & 4.7 & 1.39 & DVG & 6.8
\\ \cline{2-8}
 & 4 & 0.83 & 18.6 & 5.6 & 1.39 & DVG & 8.7
\\ \cline{2-8}
 & 5 & 0.83 & 21.5 & 6.7 & 1.31 & DVG & 10.1
\\ \cline{2-8}
 & 6 & 0.83 & 24.8 & 8.1 & 1.36 & DVG & 12.2
\\ \cline{2-8}
 & 7 & 0.83 & 28.0 & 9.5 & 1.27 & DVG & 14.1
\\ \cline{2-8}
 & 8 & 0.83 & 30.9 & 10.7 & 0.7 & DVG & 21.6
\\ \cline{2-8}
 & 9 & 0.83 & 33.6 & 12.3 & 1.05 & DVG & 17.9
\\ \cline{2-8}
 & 10 & 0.83 & 36.3 & 13.0 & 0.85 & DVG & 19.6
\\ \cline{2-8}
 & 15 & 0.83 & 59.3 & 22.5 & 0.48 & 28.5 & 54.0
\\ \hline
\end{tabular}

\begin{tabular}{| c | *{4}{c|} | *{3}{c|}}
\hline
\multirow{2}{*}{Target} & \multirow{2}{*}{$\frac{m}{n}$}
 & \multicolumn{3}{c||}{All-white} & \multicolumn{3}{c|}{Spectral}
\\  \cline{3-8}
& & \scriptsize{n-init-err} & PSNR & Run time & \scriptsize{n-init-err} & PSNR & Run time
\\ \hline
\multirow{11}{*}{\includegraphics[width=.2 \textwidth]{original_Girl_2.jpg}} 
  & 1 & 1.25 & 7.4 & 2.5 & 1.4 & DVG & 3.0
\\ \cline{2-8}
 & 2 & 1.25 & 10.4 & 3.3 & 1.39 & DVG & 4.9
\\ \cline{2-8}
 & 3 & 1.25 & 13.5 & 4.9 & 1.39 & DVG & 6.5
\\ \cline{2-8}
 & 4 & 1.25 & 16.5 & 5.7 & 1.38 & DVG & 8.1
\\ \cline{2-8}
 & 5 & 1.25 & 19.3 & 7.1 & 1.4 & DVG & 10.2
\\ \cline{2-8}
 & 6 & 1.25 & 22.6 & 8.0 & 1.32 & DVG & 12.2
\\ \cline{2-8}
 & 7 & 1.25 & 25.9 & 9.9 & 1.02 & DVG & 13.7
\\ \cline{2-8}
 & 8 & 1.25 & 28.6 & 10.6 & 0.71 & 11.1 & 23.9
\\ \cline{2-8}
 & 9 & 1.25 & 31.9 & 22.3 & 0.77 & DVG & 22.1
\\ \cline{2-8}
 & 10 & 1.25 & 34.3 & 27.3 & 0.64 & DVG & 25.2
\\ \cline{2-8}
 & 15 & 1.25 & 55.2 & 42.8 & 0.49 & 24.2 & 79.6
\\ \hline
\end{tabular}
\end{table}

\begin{table}
\centering
\caption{Wirtinger Flow performance with spectral and all-white initialization} \label{table: WF spectral and white 3}
\begin{tabular}{| c | *{4}{c|} | *{3}{c|}}
\hline
\multirow{2}{*}{Target} & \multirow{2}{*}{$\frac{m}{n}$}
 & \multicolumn{3}{c||}{All-white} & \multicolumn{3}{c|}{Spectral}
\\  \cline{3-8}
& & \scriptsize{n-init-err} & PSNR & Run time & \scriptsize{n-init-err} & PSNR & Run time
\\ \hline
\multirow{11}{*}{\includegraphics[width=.2 \textwidth]{original_Couple.jpg}} 
  & 1 & 4.92 & 4.1 & 5.0 & 1.4 & DVG & 5.3
\\ \cline{2-8}
 & 2 & 4.92 & 7.1 & 6.1 & 1.39 & DVG & 8.1
\\ \cline{2-8}
 & 3 & 4.92 & 10.1 & 10.1 & 1.38 & DVG & 10.0
\\ \cline{2-8}
 & 4 & 4.92 & 13.2 & 9.6 & 1.39 & DVG & 13.0
\\ \cline{2-8}
 & 5 & 4.92 & 16.0 & 8.8 & 1.39 & DVG & 13.0
\\ \cline{2-8}
 & 6 & 4.92 & 19.0 & 13.0 & 1.36 & DVG & 17.2
\\ \cline{2-8}
 & 7 & 4.92 & 21.5 & 13.6 & 1.29 & DVG & 19.3
\\ \cline{2-8}
 & 8 & 4.92 & 24.1 & 21.0 & 0.86 & DVG & 23.6
\\ \cline{2-8}
 & 9 & 4.92 & 25.9 & 17.6 & 1.34 & DVG & 18.4
\\ \cline{2-8}
 & 10 & 4.92 & 28.4 & 13.7 & 0.62 & 30.9 & 35.1
\\ \cline{2-8}
 & 15 & 4.92 & 37.2 & 31.5 & 0.48 & 36.0 & 64.5
\\ \hline
\end{tabular}
\end{table}

\begin{table}[h!] 
\centering
\caption{The impact of initialization on the performance of GD-COPER and Wirtinger flow. ``n-init-error" is the normalized mean square error of the initialization. The initializations chosen in this simulation are in the form of $\mvec{x}_{\rm init} = \lambda \mvec{x}_o + (1-\lambda) \mvec{x}$, where $\mvec{x}_o$ is an all-white image and $\mvec{x}$ is the true signal.  }
\label{tab:testinit1}
\begin{tabular}{| c | c | c  *{3} {*{1}{|c} | *{1}{c|}}}
\hline
\multirow{2}{*}{Target} & \multirow{2}{*}{n-init-error} & \multirow{2}{*}{$\lambda$}
& \multicolumn{2}{c||}{$\frac{m}{n} = 1$} &  \multicolumn{2}{c||}{$\frac{m}{n} = 2$} &  \multicolumn{2}{c|}{$\frac{m}{n} = 3$} 
\\ \cline{4-9}
 & & & \multicolumn{1}{c|}{GD-C } & \multicolumn{1}{c||}{WF} & 
 \multicolumn{1}{c|}{GD-C } & \multicolumn{1}{c||}{WF} &
 \multicolumn{1}{c|}{GD-C} & \multicolumn{1}{c|}{WF} 
\\ \hline
\multirow{11}{*}{\includegraphics[width=.2 \textwidth]{original_Couple.jpg}}  
& 0.0 & 0.0 & 29.94 & inf & 32.55 & inf & 34.23 & inf
\\ \cline{2-9}
 & 0.49 & 0.1 & 29.46 & DVG & 32.03 & DVG & 33.79 & 26.78
\\ \cline{2-9}
 & 0.98 & 0.2 & 29.25 & DVG & 32.03 & DVG & 34.03 & 24.11
\\ \cline{2-9}
 & 1.48 & 0.3 & 28.36 & 14.55 & 32.19 & 17.57 & 33.96 & 20.59
\\ \cline{2-9}
 & 1.97 & 0.4 & 27.12 & 12.06 & 31.22 & 15.07 & 33.18 & 18.09
\\ \cline{2-9}
 & 2.46 & 0.5 & 25.0 & 10.13 & 30.63 & 13.13 & 33.59 & 16.15
\\ \cline{2-9}
 & 2.95 & 0.6 & 23.03 & 8.54 & 30.67 & 11.55 & 33.32 & 14.57
\\ \cline{2-9}
 & 3.44 & 0.7 & 21.41 & 7.2 & 29.66 & 10.21 & 33.21 & 13.22
\\ \cline{2-9}
 & 3.94 & 0.8 & 20.59 & 6.04 & 28.51 & 9.04 & 31.37 & 12.05
\\ \cline{2-9}
 & 4.43 & 0.9 & 20.36 & 5.01 & 27.69 & 8.01 & 30.89 & 11.01
\\ \cline{2-9}
 & 4.92 & 1.0 & 18.52 & 4.09 & 27.95 & 7.09 & 31.82 & 10.07
\\ \hline
\end{tabular}	

\begin{tabular}{| c | c | c  *{3} {*{1}{|c} | *{1}{c|}}}
\hline
\multirow{2}{*}{Target} & \multirow{2}{*}{n-init-error} & \multirow{2}{*}{$\lambda$}
& \multicolumn{2}{c||}{$\frac{m}{n} = 1$} &  \multicolumn{2}{c||}{$\frac{m}{n} = 2$} &  \multicolumn{2}{c|}{$\frac{m}{n} = 3$} 
\\ \cline{4-9}
 & & & \multicolumn{1}{c|}{GD-C } & \multicolumn{1}{c||}{WF} & 
 \multicolumn{1}{c|}{GD-C } & \multicolumn{1}{c||}{WF} &
 \multicolumn{1}{c|}{GD-C} & \multicolumn{1}{c|}{WF} 
\\ \hline
\multirow{11}{*}{\includegraphics[width=.2 \textwidth]{original_Girl_3.jpg}}  
 & 0.0 & 0.0 & 30.17 & inf & 34.68 & inf & 37.67 & inf
\\ \cline{2-9}
 & 0.08 & 0.1 & 29.53 & DVG & 34.59 & DVG & 37.32 & DVG
\\ \cline{2-9}
 & 0.17 & 0.2 & 29.6 & DVG & 33.67 & DVG & 37.35 & DVG
\\ \cline{2-9}
 & 0.25 & 0.3 & 29.65 & DVG & 33.6 & DVG & 37.74 & DVG
\\ \cline{2-9}
 & 0.34 & 0.4 & 29.32 & DVG & 33.7 & DVG & 37.51 & DVG
\\ \cline{2-9}
 & 0.42 & 0.5 & 28.18 & DVG & 34.19 & DVG & 36.64 & 18.73
\\ \cline{2-9}
 & 0.5 & 0.6 & 27.68 & DVG & 34.92 & DVG & 35.95 & 19.86
\\ \cline{2-9}
 & 0.59 & 0.7 & 28.37 & DVG & 35.08 & DVG & 35.92 & 18.66
\\ \cline{2-9}
 & 0.67 & 0.8 & 28.13 & DVG & 35.12 & 14.24 & 36.23 & 17.56
\\ \cline{2-9}
 & 0.76 & 0.9 & 29.21 & DVG & 34.79 & 13.43 & 36.04 & 16.54
\\ \cline{2-9}
 & 0.84 & 1.0 & 29.15 & DVG & 34.16 & 12.59 & 35.64 & 15.63
\\ \hline
\end{tabular}	
\end{table}

\begin{table}[h!] 
\centering
\caption{The impact of initialization on the performance of GD-COPER and Wirtinger flow. ``n-init-error" is the normalized mean square error of the initialization. The initializations chosen in this simulation are in the form of $\mvec{x}_{\rm init} = \lambda \mvec{x}_o + (1-\lambda) \mvec{x}$, where $\mvec{x}_o$ is an all-white image and $\mvec{x}$ is the true signal.  }
\label{tab:testinit2}
\begin{tabular}{| c | c | c  *{3} {*{1}{|c} | *{1}{c|}}}
\hline
\multirow{2}{*}{Target} & \multirow{2}{*}{n-init-error} & \multirow{2}{*}{$\lambda$}
& \multicolumn{2}{c||}{$\frac{m}{n} = 1$} &  \multicolumn{2}{c||}{$\frac{m}{n} = 2$} &  \multicolumn{2}{c|}{$\frac{m}{n} = 3$} 
\\ \cline{4-9}
 & & & \multicolumn{1}{c|}{GD-C } & \multicolumn{1}{c||}{WF} & 
 \multicolumn{1}{c|}{GD-C } & \multicolumn{1}{c||}{WF} &
 \multicolumn{1}{c|}{GD-C} & \multicolumn{1}{c|}{WF} 
\\ \hline
\multirow{11}{*}{\includegraphics[width=.2 \textwidth]{original_House.jpg}}  
& 0.0 & 0.0 & 27.84 & inf & 31.55 & inf & 35.11 & inf
\\ \cline{2-9}
 & 0.09 & 0.1 & 28.04 & DVG & 31.5 & DVG & 35.19 & DVG
\\ \cline{2-9}
 & 0.17 & 0.2 & 27.44 & DVG & 31.24 & DVG & 35.12 & DVG
\\ \cline{2-9}
 & 0.26 & 0.3 & 26.99 & DVG & 31.47 & DVG & 35.26 & DVG
\\ \cline{2-9}
 & 0.35 & 0.4 & 26.68 & DVG & 31.23 & DVG & 35.02 & DVG
\\ \cline{2-9}
 & 0.43 & 0.5 & 26.89 & DVG & 31.62 & DVG & 34.66 & 19.12
\\ \cline{2-9}
 & 0.52 & 0.6 & 26.5 & DVG & 32.18 & DVG & 33.89 & 18.97
\\ \cline{2-9}
 & 0.61 & 0.7 & 26.69 & DVG & 32.4 & DVG & 33.54 & 17.94
\\ \cline{2-9}
 & 0.7 & 0.8 & 26.56 & DVG & 31.97 & 13.86 & 33.71 & 17.13
\\ \cline{2-9}
 & 0.78 & 0.9 & 26.26 & DVG & 31.74 & 12.92 & 34.16 & 16.12
\\ \cline{2-9}
 & 0.87 & 1.0 & 26.71 & DVG & 32.0 & 12.11 & 34.6 & 15.21
\\ \hline
\end{tabular}	
\begin{tabular}{| c | c | c  *{3} {*{1}{|c} | *{1}{c|}}}
\hline
\multirow{2}{*}{Target} & \multirow{2}{*}{n-init-error} & \multirow{2}{*}{$\lambda$}
& \multicolumn{2}{c||}{$\frac{m}{n} = 1$} &  \multicolumn{2}{c||}{$\frac{m}{n} = 2$} &  \multicolumn{2}{c|}{$\frac{m}{n} = 3$} 
\\ \cline{4-9}
 & & & \multicolumn{1}{c|}{GD-C } & \multicolumn{1}{c||}{WF} & 
 \multicolumn{1}{c|}{GD-C } & \multicolumn{1}{c||}{WF} &
 \multicolumn{1}{c|}{GD-C} & \multicolumn{1}{c|}{WF} 
\\ \hline
\multirow{11}{*}{\includegraphics[width=.2 \textwidth]{original_tree.jpg}}  
& 0.0 & 0.0 & 23.65 & inf & 26.23 & inf & 27.53 & inf
\\ \cline{2-9}
 & 0.1 & 0.1 & 23.55 & DVG & 26.26 & DVG & 27.65 & DVG
\\ \cline{2-9}
 & 0.2 & 0.2 & 23.69 & DVG & 26.14 & DVG & 27.68 & DVG
\\ \cline{2-9}
 & 0.3 & 0.3 & 23.49 & DVG & 26.28 & DVG & 27.46 & DVG
\\ \cline{2-9}
 & 0.39 & 0.4 & 23.45 & DVG & 26.14 & DVG & 27.49 & DVG
\\ \cline{2-9}
 & 0.49 & 0.5 & 23.49 & DVG & 26.13 & DVG & 27.6 & DVG
\\ \cline{2-9}
 & 0.59 & 0.6 & 23.45 & DVG & 26.19 & DVG & 27.56 & DVG
\\ \cline{2-9}
 & 0.69 & 0.7 & 23.48 & DVG & 26.18 & DVG & 27.43 & 16.88
\\ \cline{2-9}
 & 0.79 & 0.8 & 22.82 & DVG & 26.44 & 12.72 & 27.53 & 15.9
\\ \cline{2-9}
 & 0.89 & 0.9 & 22.97 & DVG & 26.43 & 11.9 & 27.5 & 14.92
\\ \cline{2-9}
 & 0.99 & 1.0 & 22.62 & DVG & 26.27 & 11.03 & 27.56 & 14.0
\\ \hline
\end{tabular}
\label{tab:testinit3}
\end{table}

\section{Discussion of our assumptions}
\label{sec:assumption}

In the proof of the convergence of GD-COPER in Theorem \ref{thm convergence of algorithm}, we made two assumptions:
\begin{enumerate}
\item $\|\mvec{x}\|_2^2 =1$ and $\|\mvec{z}\|_2^2=1$ for all $\mvec{z} \in \mathcal{C}_r$.  

\item $\mathcal{P}_{\C_r} (\cdot) = \mathcal{D}_r (\mathcal{E}_r (\cdot))$, i.e. the application of the encoder and decoder of a compression algorithm is equivalent to projecting a signal on the closest code-word. 

\end{enumerate}
First, note that we can obtain a good estimate of the norm of the signal, and normalize the measurements and pretend that the signal satisfies $\|\mvec{x}\|_2^2 =1$. Below we present one approach to execute this normalization. Suppose that $ \mvec{y} = \envert{A \mvec{x}}$. We have 
\begin{equation*}
\e{\envert{{y}_k}^2}=  \e{{y}_k^*{y}_k} = \e{\mvec{x}^*  \mvec{a}_k^* \mvec{a}_k \mvec{x}} = 2 n \enVert{\mvec{x}}^2.
\end{equation*}
Hence, $ \frac{1}{2 n m } \enVert{\mvec{y}}^2 \xrightarrow{\mathbb{P}} \enVert{\mvec{x}}^2$, where the notation $ \xrightarrow{\mathbb{P}}$ denotes convergence in probability. Hence, if we divide our measurements by $\sqrt{ \frac{1}{2 n m } \enVert{\mvec{y}}^2}$, then we can  assume that $\|\mvec{x}\|_2=1$. Once we know that the magnitude of the signal is equal to one, we can modify any compression algorithm to have the property 
$\|\mvec{z}\|_2^2=1$ for all $\mvec{z} \in \mathcal{C}_r$, by dividing the output of the decoder by its magnitude. One question that we still have to address though is the following: Often times the estimate of the magnitude of the signal is random and may deviate from what we expect. Hence, we may end up having a signal $\mvec{x}$ whose magnitude satisfies $|\|\mvec{x}\|_2^2-1| \leq \gamma$ where $\gamma$ is a small number. What would be the impact of such an error in the performance of GD-COPER? In particular, one would hope that this error does not accumulate in the iterations of the algorithm. Our next theorem proves this claim.

\begin{thm} \label{thm:nonideal:magestimate}
Consider a fixed signal $\mvec{x} \in \Q$ that satisfies $|\|\mvec{x}\|_2^2-1| \leq \gamma$. Define $\mvec{z}_t  \in \C_r$ as in \eqref{eq:GD-COPER} with $ \mu = \frac{1}{8 m} $.  Suppose that for all $  \theta \in \mathds{R}$,   $\ee^ {i \theta} \mvec{x} \in \Q$. Define
$ \theta_t \triangleq \arg \min\limits_{\theta \in \mathds{R}} \enVert{\mvec{z}_t -  {\rm e}^{i \theta} \mvec{x}}$.  For all $ \epsilon \geq C_2 m^{-\frac{1}{3}} $, with probability at least $ 1 - C_3 {\rm e}^{- C_1 \sqrt{m \epsilon} + (3 \ln 2) r} $, where $C_1, C_2, C_3 > 0$ are absolute constants, for  $t=1,2,\ldots$, we have
\begin{equation}\label{eq:theoremGD 2}
\enVert{ \mvec{z}_{t + 1} -  {\rm e}^{i \theta_t} \mvec{x} } \leq \intoo{ \enVert{ \mvec{z}_t -  {\rm e}^{i \theta_t} \mvec{x}} + \epsilon } \enVert{ \mvec{z}_t -  {\rm e}^{i \theta_t} \mvec{x}} + 3 \delta_r+ \gamma. 
\end{equation}
\end{thm}
\begin{proof}
Note that we still assume that for all $\mvec{z} \in \C_r$, we have $\|\mvec{z}\|_2^2=1$, since this condition is straightforward to satisfy exactly (given that the output of decoder is available and hence we can directly normalize it). Since the proof of this theorem is similar to the proof of Theorem Theorem \ref{thm convergence of algorithm} we skip most of the steps, and only mention the ones that are different. By following the steps in the proof of Theorem \ref{thm convergence of algorithm} that led to \eqref{thm conv alg proof 1} we obtain
\begin{align}
 \tilde{\mvec{x}} - \mvec{s}_{t + 1}  = 
\tilde{\mvec{x}} -  {\rm e}^{i \theta_t} \mvec{x}  + (1 - ({\rm e}^{i \theta_t} \mvec{x})^* \mvec{z}_t) {\rm e}^{i \theta_t} \mvec{x}  + \frac{1}{8m} \intoo{ \nabla d_A(\mvec{z}_t) - \e{ \nabla d_A(\mvec{z}_t)} }.
\end{align}
In the proof of Theorem \ref{thm convergence of algorithm}, we claimed that $(1-({\rm e}^{i \theta_t} \mvec{x})^* \mvec{z}_t) = \frac{1}{2} \enVert{  {\rm e}^{i \theta_t} \mvec{x} - \mvec{z}_t }^2$. Clearly, this is not true any more. Instead we have  
\[
\frac{1}{2} \enVert{  {\rm e}^{i \theta_t} \mvec{x} - \mvec{z}_t }^2 = \frac{1}{2}  \| \mvec{x}\|^2 -\frac{1}{2} + (1-({\rm e}^{i \theta_t} \mvec{x})^* \mvec{z}_t). 
\]
Hence, we will conclude that
\begin{eqnarray} \nonumber
\lefteqn{\enVert{  { {\rm e}^{i \theta_t} \mvec{x}} - \mvec{z}_{t + 1} }
 \leq 
\enVert{{{\rm e}^{i \theta_t} \mvec{x}} - \tilde{\mvec{x}}} + \enVert{  \tilde{\mvec{x}} - \mvec{z}_{t + 1} }}
\\ 
& \nonumber \leq&
\delta_r + 2 \Re \intoo{ \mvec{v}_{t}^* ( \tilde{\mvec{x}} -  \mvec{s}_{t + 1}) }
\\ & \nonumber
\leq& \delta_r + 2 \enVert{\mvec{v}_t} \enVert{\tilde{\mvec{x}} - {\rm e}^{i \theta_t} \mvec{x}} + 2 (1 - ({\rm e}^{i \theta_t} \mvec{x})^* \mvec{z}_t) \enVert{\mvec{v}_t} \enVert{ {\rm e}^{i \theta_t} \mvec{x}} + \frac{1}{4m} \Re \intoo{\mvec{v}_t^* \intoo{\nabla d_A(\mvec{z}_t) - \e{\nabla d_A(\mvec{z}_t)}} }
\\
& \leq& 
\delta_r + 2 \delta_r + \enVert{{{\rm e}^{i \theta_t} \mvec{x}} - \mvec{z}_{t} }^2 +   |\| \mvec{x}\|^2 - 1 |+  \frac{1}{4 m} \Re  \intoo{\mvec{v}_{t}^*\intoo{ \nabla d_A(\mvec{z}_t) - \e{ \nabla d_A(\mvec{z}_t)} } } \nonumber \\
& \leq& 3 \delta_r + \gamma + \enVert{{{\rm e}^{i \theta_t} \mvec{x}} - \mvec{z}_{t} }^2 +  \frac{1}{4 m} \Re  \intoo{\mvec{v}_{t}^*\intoo{ \nabla d_A(\mvec{z}_t) - \e{ \nabla d_A(\mvec{z}_t)} } }. \nonumber
\end{eqnarray}
The rest of the proof is exactly the same as the proof of Theorem \ref{thm convergence of algorithm}, and is hence skipped. 
\end{proof}
We would like to emphasize that given the linear convergence of the GD-COPER, the accumulation of the error due to $\gamma$ will be negligible and the error in the estimation of the magnitude of $\| \mvec{x}\|^2$ does not have any major impact on the performance of GD-COPER. 

We now turn our attention to the second assumption, i.e. the assumption that $\mathcal{P}_{\C_r} (\cdot) = \mathcal{D}_r (\mathcal{E}_r (\cdot))$. We would like to first emphasize that ideally, this is what a compression algorithm should do. If an image compression algorithm maps an image to a codeword that is far from the original image, that is an indication of the fact that the compression algorithm is not good. However, it is also reasonable to consider situations in which multiple codewords are close to an image and the compression algorithm does not pick the one which is the closest to the image because of some non-ideal strategies that is chosen to reduce the computational complexity. Hence, again we can ask whether the GD-COPER algorithm is robust to such non-ideal compression algorithms?  In the rest of this section, we pursue the following two goals:
\begin{enumerate}
\item Provide a few examples to convince the readers that most of the standard compression algorithms try to mimic a projection onto the codewords.

\item Suppose that even though the compression algorithm is non-ideal and does not find the closest codeword, it is still capable of finding a codeword that is in the vicinity of the closest codeword. We aim to show that the performance of GD-COPER algorithm is robust to such non-idealities. 
\end{enumerate}

Let us start with an example that is the cornerstone of several important compression algorithms. Suppose $\Q \subset \intcc{0, 1}^n$ is the set of approximately sparse signals. For instance, for some $p<1$
\begin{equation*}
\Q = \cbr{\mvec{x} \in [0,1)^n, \; \norm{\mvec{x}}_p \leq \zeta }.
\end{equation*}
The main idea of many compression algorithms is to approximate the signals in $\Q$ with $k$-sparse signals and encode the $k$-sparse signal. For simplicity suppose that we are given the $k$.  Let  $r_t = k \lceil \log n \rceil + k (t + 1)$ denote the rate of our compression algorithm.  Let $\E_1: \Q \to \cbr{0, 1}^{k \lceil \log n \rceil}$ encode the location of $k$ largest  elements of $\mvec{x}$. Furthermore, to code the magnitudes of the non-zero coefficients we consider $\E_2 : \Q \to \cbr{0, 1}^{k (t + 1)}$ that consider the $k$ largest components of $\mvec{x}$ and codes each of them with $t+1$ bits (does a binary expansion and keeps the $t+1$ most significant bits).  

More precisely, if $\mvec{x} = \intoo{\mvec{x}_1, \mvec{x}_2, ..., \mvec{x}_n} \in \Q$, where $1 \leq i_1 < i_2 < ... < i_k \leq n $ are the location of its $k$-largest elements, then
\begin{equation}
\E_1 (\mvec{x}) = \intoo{B(i_1), ..., B(i_k)},
\end{equation}
where $B(i)$ denotes the binary expansion of positive integer $i$.  Note that since indices are less than or equal to $n$, $\log n$ bits are enough to code each of them.  Moreover, if $x_i = \sum_{j = 1}^\infty \epsilon_{i, j} 2^{-j}$ with $\epsilon_{i, j} \in \cbr{0, 1}$, denote the binary expansion of $x_i$, then 
\begin{equation}
\E_2 (\mvec{x}) = \intoo{(\epsilon_{i_1,1}, \epsilon_{i_1,2}, \ldots, \epsilon_{i_1,t+1}), ..., (\epsilon_{i_k,1}, \epsilon_{i_k,2}, \ldots, \epsilon_{i_k,t+1})}.
\end{equation}
Note that this type of coding is very close to what happens in e.g. JPEG and embedded zero tree wavelet (EZW) compression algorithms. Both compression algorithms first transform the image to a domain that is more compressible, e.g. Fourier and wavelet, and then code the location and magnitudes of the largest coefficients similar to what we did above.\footnote{There are some minor tweaks in the actual JPEG and EZW. Since coding the locations of the largest coefficients requires a large number of bits, they often use techniques such as counting zero runs or coding along the trees to reduce the number of bits. }  The decoder of the compression algorithms has access to the locations of the largest coefficients from the $k \log n$ bits that it received from the encoder. Hence, it can easily use  $(\epsilon_{i_1,1}, \epsilon_{i_1,2}, \ldots, \epsilon_{i_1,t+1}), ..., (\epsilon_{i_k,1}, \epsilon_{i_k,2}, \ldots, \epsilon_{i_k,t+1})$ to find the magnitudes of the signals at those locations. Define $\Gamma_k = \{\mvec{x} \in \intcc{0, 1}^n \ : \ \| \mvec{x}\|_0 \leq k \}$. One can easily confirm that
\begin{equation*}
\C_{r_t} = \D_{r_t} \intoo{\E_{r_t} \intoo{\Q}} = \cbr{ \mvec{y} \in \Gamma_k, \; y_i = \sum_{j = 1}^t \epsilon_{i, j} 2^{-j}, \quad \epsilon_{i, j} \in \cbr{0, 1} }.
\end{equation*}
It is straightforward to show that in these types of compression algorithms $\P_{\C_r} (\cdot) = \D_r \circ \E_r (\cdot)$. For the sake of completeness we include a brief proof below. Suppose that the choice of codeword for the projection is unique. For notational simplicity we drop the subscript $t$.  To prove this formally, let $\mvec{x} \in [0,1)^n$ be an arbitrary vector and let $\mvec{y} = \D_r (\E_r (\mvec{x}))$ and $\mvec{z} = \P_{\C_r} (\mvec{x})$.  We have to show $\mvec{y} = \mvec{z}$.  Since $\mvec{z} \in \C_r$, it has at most $k$ non-zero coordinates.  Firstly, we claim location of these non-zero coordinates have to match with the largest coordinates of $\mvec{x}$.  If this does not hold, one can swap two coordinates of $\mvec{z}$ and get smaller distance to $\mvec{x}$ by noting that if ${x}_i < {x}_j$ and ${z}_i > 0, {z}_j = 0$ then
\begin{equation*}
\intoo{{x}_i - {z}_i}^2 + {x}_j^2 < \intoo{{x}_j - {z}_i}^2 + {x}_i^2,
\end{equation*} 
which contradicts with $\mvec{z}$ being the projection of $\mvec{x}$.  Furthermore, if ${y}_i = \sum_{j = 1}^t \epsilon_{i, j} 2^{-j}$ and ${z}_i = \sum_{j = 1}^t \tilde{\epsilon}_{i, j} 2^{-j}$ then $\abs{{x}_i - {y}_i} \leq 2^{-t - 1}$ and $\abs{{x}_i - {z}_i} \leq 2^{-t - 1}$ implies $\abs{{y}_i - {z}_i} \leq 2^{-t}$ which yields $\abs{{x}_i - {y}_i} = \abs{ {x}_i - {z}_i}$. Note that there can be a case where ${y}_i \neq {z}_i$ while they have the same distance from ${x}_i$.  As an example, consider $t = 0, \; {x}_i = 0.5, \; {y}_i = 0, \; \mvec{z}_i = 1$. This yields for every $i$ that $ \abs{{x}_i - {y}_i} \leq \abs{{x}_i - {z}_i}$. Hence
\begin{equation*}
\norm{\mvec{x} - \mvec{y}} \leq \norm{\mvec{x} - \mvec{z}},
\end{equation*}
which means $\mvec{y} = \D ( \E (\mvec{x}) )$ is also a projection on $\C_r$.

Now, let us turn to another point we would like to make, that is, even if the compression algorithm is not an accurate projection, GD-COPER can still perform an accurate recovery. Towards this goal we assume that the operation of $ \D ( \E (\mvec{x}) )$ is not a projection, but has some error. In other words, we assume that
\[
\|\D ( \E (\mvec{x}) ) - \P_{\C_r} (\mvec{x})\| \leq \gamma. 
\]
Our next theorem proves that if $\gamma$ is not too large, then GD-COPER given by the following iteration can still perform well:
\begin{align*} 
\mvec{s}_{t + 1} &= \mvec{z}_t - \mu \nabla d_A ( \mvec{z}_t ),\nonumber\\
\quad \mvec{z}_{t + 1}& =  \mathcal{D}_r  (\mathcal{E}_r ( \mvec{s}_{t + 1} )).
\end{align*}

\begin{thm} \label{thm convergence non-ideal compression}
For a fixed signal $\mvec{x} \in \Q$, define $\mvec{z}_t  \in \C_r$ as in \eqref{eq:GD-COPER} with $ \mu = \frac{1}{8 m} $.  Suppose that for all $  \theta \in \mathds{R}$,   $\ee^ {i \theta} \mvec{x} \in \Q$. Define
$ \theta_t \triangleq \arg \min\limits_{\theta \in \mathds{R}} \enVert{\mvec{z}_t -  {\rm e}^{i \theta} \mvec{x}}$.  For all $ \epsilon \geq C_2 m^{-\frac{1}{3}} $, with probability at least $ 1 - C_3 {\rm e}^{- C_1 \sqrt{m \epsilon} + (3 \ln 2) r} $, where $C_1, C_2, C_3 > 0$ are absolute constants, for  $t=1,2,\ldots$, we have
\begin{equation}\label{eq:theoremGD 3}
\enVert{ \mvec{z}_{t + 1} -  {\rm e}^{i \theta_t} \mvec{x} } \leq \intoo{ \enVert{ \mvec{z}_t -  {\rm e}^{i \theta_t} \mvec{x}} + 2 \epsilon } \enVert{ \mvec{z}_t -  {\rm e}^{i \theta_t} \mvec{x}} + 3 \delta_r + 2 \gamma  +  \sqrt{2 \gamma (\delta_r + 1 + \epsilon)}. 
\end{equation}
\end{thm}
Before we prove this theorem, let us interpret it. Everything in the theorem is similar to what we had in Theorem \ref{thm convergence of algorithm}. The only difference, is the term $2 \gamma  +  \sqrt{2 \gamma (\delta_r + 1 + \epsilon)}$ added to the error. Again given the geometric convergence of the algorithm the total error after $T$ iterations does not accumulate much and remains at the same order. It is clear that if $\gamma$ is small, then the GD-COPER algorithm performs well. 

\begin{proof}[Proof of Theorem \ref{thm convergence non-ideal compression}]
Since the proof is very similar to the proof of Theorem \ref{thm convergence of algorithm} we do not repeat the entire proof and only emphasize on the aspects of this proof that change. 
Let $\tilde{\mvec{x}} = \P_{\C_r}(  {\rm e}^{i \theta_t} \mvec{x}) $, and define $\mvec{w}_{t+1} = \P_{\C_r}(\mvec{s}_{t + 1})$. In this case, we know that $ \mvec{z}_{t + 1} = \D( \E (\mvec{s}_{t + 1}))$ and we have
\begin{equation}\label{eq:closeproj}
\|\mvec{w}_{t+1} - \mvec{z}_{t+1}\| \leq \gamma. 
\end{equation}

Since $ \tilde{\mvec{x}} \in \C_r $, we have
\begin{align*} \nonumber
\enVert{ \mvec{s}_{t + 1} -  \tilde{\mvec{x}} }^2 
& \geq \enVert{ \mvec{s}_{t + 1} -  {\mvec{w}_{t + 1}} }^2   \\
&= \|\mvec{s}_{t+1}-\mvec{z}_{t+1}\|^2 + \|\mvec{z}_{t+1}-\mvec{w}_{t+1}\|^2 + 2 \Re \left( (\mvec{z}_{t+1}-\mvec{w}_{t+1})^*(\mvec{s}_{t+1}-\mvec{z}_{t+1})\right)
\\ & \geq \|\mvec{s}_{t+1}-\mvec{z}_{t+1}\|^2 + 2 \Re \left( (\mvec{z}_{t+1}-\mvec{w}_{t+1})^*(\mvec{s}_{t+1}-\mvec{z}_{t+1})\right),
\end{align*}
where to obtain the last inequality we used the Cauchy-Schwartz inequality, \eqref{eq:closeproj}, and the fact that both $\mvec{s}_{t+1}$ and $\mvec{z}_{t+1}$ have unit norms. Therefore, we have

\begin{align} \nonumber
\enVert{ \mvec{s}_{t + 1} -  \tilde{\mvec{x}} }^2 
& \geq \enVert{ \mvec{s}_{t + 1} -  {\mvec{z}_{t + 1}} }^2+2 \Re \left( (\mvec{z}_{t+1}-\mvec{w}_{t+1})^*(\mvec{s}_{t+1}-\mvec{z}_{t+1})\right)  \\
 & =
\enVert{\mvec{s}_{t + 1} -  \tilde{\mvec{x}} }^2 + \enVert{ \tilde{ \mvec{x}} - \mvec{z}_{t + 1}}^2 + 2 \Re \intoo{ ( \tilde{ \mvec{x}} - \mvec{z}_{t + 1})^* (\mvec{s}_{t + 1} - \tilde{\mvec{x}}) } \nonumber \\
&+  2 \Re \left( (\mvec{z}_{t+1}-\mvec{w}_{t+1})^*(\mvec{s}_{t+1}-\tilde{\mvec{x}})\right)+ 2 \Re \left( (\mvec{z}_{t+1}-\mvec{w}_{t+1})^*(\tilde{\mvec{x}}-\mvec{z}_{t+1})\right). \nonumber
\end{align}
Hence,
\begin{equation} \label{upper bdd of error by real inner product 2}
\enVert{ \tilde{\mvec{x}} - \mvec{z}_{t + 1}}^2 \leq 2 \Re \intoo{ ( \tilde{ \mvec{x}} - \mvec{z}_{t + 1})^* (\tilde{\mvec{x}} - \mvec{s}_{t + 1}) }  +  2 \Re \left( (\mvec{w}_{t+1}-\mvec{z}_{t+1})^*(\mvec{s}_{t+1}-\tilde{\mvec{x}})\right)+ 2 \Re \left( (\mvec{w}_{t+1}-\mvec{z}_{t+1})^*(\tilde{\mvec{x}}-\mvec{z}_{t+1})\right)
\end{equation}

Recall that  $ \e{ \nabla d_A (\mvec{z}) } 
= 8m ( \mvec{z} \mvec{z}^* - \mvec{x} \mvec{x}^* ) \mvec{z} $. Thus, 
\begin{align} \nonumber
 \tilde{\mvec{x}} - \mvec{s}_{t + 1} 
& =
 \tilde{\mvec{x}} -  {\rm e}^{i \theta_t} \mvec{x}  + {\rm e}^{i \theta_t} \mvec{x}  - \intoo{ \mvec{z}_t - \frac{1}{8 m} \e{ \nabla d_A( \mvec{z}_t) } + \frac{1}{8m} \intoo{\e{ \nabla d_A(\mvec{z}_t)} - \nabla d_A(\mvec{z}_t)} }
\\ & = \nonumber
 \tilde{\mvec{x}} -  {\rm e}^{i \theta_t} \mvec{x}  + {\rm e}^{i \theta_t} \mvec{x}  - \intoo{\mvec{z}_t - \mvec{z}_t + (\mvec{x}^* \mvec{z}_t) \mvec{x}} + \frac{1}{8m} \intoo{ \nabla d_A (\mvec{z}_t) - \e{ \nabla d_A(\mvec{z}_t)} }
\\ & = \label{thm conv alg proof 1_2}
\tilde{\mvec{x}} -  {\rm e}^{i \theta_t} \mvec{x}  + (1 - ({\rm e}^{i \theta_t} \mvec{x})^* \mvec{z}_t) {\rm e}^{i \theta_t} \mvec{x}  + \frac{1}{8m} \intoo{ \nabla d_A(\mvec{z}_t) - \e{ \nabla d_A(\mvec{z}_t)} }.
\end{align}
Note that 
$ \enVert{\tilde{\mvec{x}} - {\rm e}^{i \theta_t} \mvec{x}} \leq \delta_r$.
Also,  since $1 -({\rm e}^{i \theta_t} \mvec{x})^* \mvec{z}_t = \frac{1}{2} \enVert{  {\rm e}^{i \theta_t} \mvec{x} - \mvec{z}_t }^2
$
and $ \enVert{{\rm e}^{i \theta_t} \mvec{x}} = \enVert{\mvec{v}_{t}} = 1 $. Let 
\begin{align}
 \mvec{v}_{t} \triangleq \frac{ \tilde{\mvec{x}}  - \mvec{z}_{t + 1}}{\enVert{ \tilde{\mvec{x}}  - \mvec{z}_{t + 1} }}.
\end{align}
and 
\begin{align}
 \tilde{\mvec{{v}}}_{t} \triangleq \frac{ \mvec{w}_{t+1}  - \mvec{z}_{t + 1}}{\enVert{ \mvec{w}_{t+1}  - \mvec{z}_{t + 1} }}.
\end{align}
Then we have
\begin{align} \label{upper bdd of error by real inner product 3}
&\enVert{ \tilde{\mvec{x}} - \mvec{z}_{t + 1}}^2 \leq 2 \Re \intoo{ ( \tilde{ \mvec{x}} - \mvec{z}_{t + 1})^* (\tilde{\mvec{x}} - \mvec{s}_{t + 1}) }  +  2 \Re \left( (\mvec{w}_{t+1}-\mvec{z}_{t+1})^*(\mvec{s}_{t+1}-\tilde{\mvec{x}})\right)+ 2 \Re \left( (\mvec{w}_{t+1}-\mvec{z}_{t+1})^*(\tilde{\mvec{x}}-\mvec{z}_{t+1})\right) \nonumber \\
& \leq 2 \| \tilde{ \mvec{x}} - \mvec{z}_{t + 1}\| |\Re \intoo{ \mvec{v}_{t}^*(\tilde{\mvec{x}} - \mvec{s}_{t + 1})  }| + 2 \gamma |\Re \left( (\tilde{\mvec{v}}_{t+1})^*(\mvec{s}_{t+1}-\tilde{\mvec{x}})\right)| + 2 \gamma | \Re \left( (\tilde{\mvec{v}}_{t+1})^*(\tilde{\mvec{x}}-\mvec{z}_{t+1})\right)| \nonumber \\
&\leq 2 \| \tilde{ \mvec{x}} - \mvec{z}_{t + 1}\| \left(\delta_r +\frac{1}{2} \enVert{  {\rm e}^{i \theta_t} \mvec{x} - \mvec{z}_t }^2 + \left|\Re \intoo{ \mvec{v}_{t}^*( \frac{1}{8m} \intoo{ \nabla d_A(\mvec{z}_t) - \e{ \nabla d_A(\mvec{z}_t)} )  }}\right|\right) \nonumber \\
&+ 2 \gamma \left(\delta_r +\frac{1}{2} \enVert{  {\rm e}^{i \theta_t} \mvec{x} - \mvec{z}_t }^2 + \left|\Re \intoo{ \mvec{v}_{t}^*( \frac{1}{8m} \intoo{ \nabla d_A(\mvec{z}_t) - \e{ \nabla d_A(\mvec{z}_t)} )  }}\right|\right) + 2 \gamma \|\tilde{\mvec{x}}-\mvec{z}_{t+1}\|.  
\end{align}
%
%\begin{align} \nonumber
%\enVert{  { {\rm e}^{i \theta_0} \mvec{x}} - \mvec{z}_{t + 1} }
%& =
%\enVert{{{\rm e}^{i \theta_0} \mvec{x}} - \tilde{\mvec{x}}} + \enVert{  \tilde{\mvec{x}} - \mvec{z}_{t + 1} } 
%\\ 
%& \leq
%\delta_r + 2 \Re \intoo{ \mvec{v}_{t}^* ( \tilde{\mvec{x}} -  \mvec{s}_{t + 1}) }
%\\ & \nonumber
%\leq \delta_r + 2 \enVert{\mvec{v}_t} \enVert{\tilde{\mvec{x}} - {\rm e}^{i \theta_0} \mvec{x}} + 2 (1 - ({\rm e}^{i \theta_0} \mvec{x})^* \mvec{z}_t) \enVert{\mvec{v}_t} \enVert{ {\rm e}^{i \theta_0} \mvec{x}} + \frac{1}{4m} \Re \intoo{\mvec{v}_t^* \intoo{\nabla d_A(\mvec{z}_t) - \e{\nabla d_A(\mvec{z}_t)}} }
%\\
%& \leq \label{thm conv alg proof 2}
%\delta_r + 2 \delta_r + \enVert{{{\rm e}^{i \theta_0} \mvec{x}} - \mvec{z}_{t} }^2 + \frac{1}{4 m} \Re  \intoo{\mvec{v}_{t}^*\intoo{ \nabla d_A(\mvec{z}_t) - \e{ \nabla d_A(\mvec{z}_t)} } }.
%\end{align}
%
Define events $\cal G$  and $\tilde{\cal G}$ as follows
\begin{equation} \label{def event G}
{\cal G} \triangleq \cbr{ \frac{1}{4m} \Re \left(\mvec{v}^* \intoo{ \nabla d_A(\mvec{z}) - \e{ \nabla d_A(\mvec{z}) } } \right) \leq \epsilon \inf_{\theta \in \mathds{R}} \enVert{ {\rm e}^{i \theta} \mvec{x} - \mvec{z} }, \quad \mvec{v} = \frac{\tilde{\mvec{x}} - \mvec{z'}}{\enVert{\tilde{\mvec{x}} - \mvec{z'}}}, \qquad \forall \mvec{z}, \tilde{\mvec{x}} \in \C_r },
\end{equation}
\begin{equation} \label{def event tilde G}
{\tilde{\cal G}} \triangleq \cbr{ \frac{1}{4m} \Re \left(\tilde{\mvec{{v}}}^* \intoo{ \nabla d_A(\mvec{z}) - \e{ \nabla d_A(\mvec{z}) } } \right) \leq \epsilon \inf_{\theta \in \mathds{R}} \enVert{ {\rm e}^{i \theta} \mvec{x} - \mvec{z} }, \quad \tilde{\mvec{{v}}} = \frac{\mvec{z} - \mvec{z'}}{\enVert{\mvec{z} - \mvec{z'}}}, \qquad \forall \mvec{z}, \mvec{z'} \in \C_r }. 
\end{equation}
 
 Similar to the proof of Theorem \ref{thm convergence of algorithm}, we can bound the probabilities of these events. In particular, we have constants $C_1, C_2, C_3 > 0$ such that for every $\epsilon \geq C_2 m^{-\frac{1}{3}}$,
\begin{equation}\label{eq:concentrationsinglepart 2}
\p{ \envert{ \Re \intoo{  \mvec{v}^* \intoo{ \nabla d_A(\mvec{z}) - \e{\nabla d_A(\mvec{z})}} } } > 4 m \epsilon \inf_{\theta \in \mathds{R}} \enVert{{\rm e}^{i \theta} \mvec{x} - \mvec{z} } } \leq  C_3 {\rm e}^{ - C_1 \sqrt{m \epsilon}}. 
\end{equation}
Hence, combining \eqref{eq:concentrationsinglepart 2} with the union bound, for every $\epsilon \geq C_2 m^{-\frac{1}{3}}$, we have 
\begin{equation}
\p{\cal G} \geq 1 - 2^{3r} C_3 {\rm e}^{- C_1 \sqrt{m \epsilon}}. %\quad \forall \epsilon \geq C_2 m^{-\frac{1}{3}}.
\end{equation}
Similarly, we have
\begin{equation}
\p{\tilde{\cal G}} \geq 1 - 2^{3r} C_3 {\rm e}^{- C_1 \sqrt{m \epsilon}}. %\quad \forall \epsilon \geq C_2 m^{-\frac{1}{3}}.
\end{equation}
Therefore, conditioned on  $\cal G \cap \tilde{\cal G}$ we have  
\begin{align*}
\frac{1}{4 m} \Re  \intoo{\mvec{v}_{t}^*\intoo{ \nabla d_A(\mvec{z}_t) - \e{ \nabla d_A(\mvec{z}_t)} } } & \leq \epsilon \inf_{\theta \in \mathds{R}} \enVert{{\rm e}^{i \theta} \mvec{x} - \mvec{z}_t} 
= \epsilon \enVert{{\rm e}^{i \theta_t} \mvec{x} - \mvec{z}_t}.
\end{align*}
Hence, \eqref{upper bdd of error by real inner product 3} implies that, for all  $t\in\{ 1 \cdots,T\}$, 
\begin{align}
&\enVert{ \tilde{\mvec{x}} - \mvec{z}_{t + 1}}^2 \leq 2 \| \tilde{ \mvec{x}} - \mvec{z}_{t + 1}\| \left(\delta_r +\frac{1}{2} \enVert{  {\rm e}^{i \theta_t} \mvec{x} - \mvec{z}_t }^2 +  \epsilon \enVert{{\rm e}^{i \theta_t} \mvec{x} - \mvec{z}_t} \right) \nonumber \\
&+ 2 \gamma \left(\delta_r +\frac{1}{2} \enVert{  {\rm e}^{i \theta_t} \mvec{x} - \mvec{z}_t }^2 +   \epsilon \enVert{{\rm e}^{i \theta_t} \mvec{x} - \mvec{z}_t} \right) + 2 \gamma \|\tilde{\mvec{x}}-\mvec{z}_{t+1}\|  \nonumber \\
&\leq 2 \| \tilde{ \mvec{x}} - \mvec{z}_{t + 1}\| \left(\delta_r +\frac{1}{2} \enVert{  {\rm e}^{i \theta_t} \mvec{x} - \mvec{z}_t }^2 +  \epsilon \enVert{{\rm e}^{i \theta_t} \mvec{x} - \mvec{z}_t} + \gamma \right) + 2 \gamma (\delta_r + 1 + \epsilon). 
\end{align}
Given that we have a quadratic function of $\enVert{ \tilde{\mvec{x}} - \mvec{z}_{t + 1}}$ with one negative and one positive root, it is straightforward to see that $\enVert{ \tilde{\mvec{x}} - \mvec{z}_{t + 1}}$ should be smaller than the positive root. By bounding the positive root we obtain
\[
\enVert{ \tilde{\mvec{x}} - \mvec{z}_{t + 1}} \leq 2 \left(\delta_r +\frac{1}{2} \enVert{  {\rm e}^{i \theta_t} \mvec{x} - \mvec{z}_t }^2 +  \epsilon \enVert{{\rm e}^{i \theta_t} \mvec{x} - \mvec{z}_t} + \gamma \right) + \sqrt{2 \gamma (\delta_r + 1 + \epsilon)}. 
\]
Hence,
\begin{align*}\label{eq:last_nonideal compression}
\enVert{{\rm e}^{i \theta_t} \mvec{x} - \mvec{z}_{t+1}} &\leq \delta_r + 2 \left(\delta_r +\frac{1}{2} \enVert{  {\rm e}^{i \theta_t} \mvec{x} - \mvec{z}_t }^2 +  \epsilon \enVert{{\rm e}^{i \theta_t} \mvec{x} - \mvec{z}_t} + \gamma \right) + \sqrt{2 \gamma (\delta_r + 1 + \epsilon)}  \nonumber \\
&\leq 3 \delta_r + 2 \gamma  +  \sqrt{2 \gamma (\delta_r + 1 + \epsilon)} + \enVert{{\rm e}^{i \theta_t} \mvec{x} - \mvec{z}_t} \left(2 \epsilon + \enVert{{\rm e}^{i \theta_t} \mvec{x} - \mvec{z}_t} \right). 
\end{align*}

\end{proof}

\section{Proofs} 
\label{sec:proofs}
\subsection{Preliminaries}

\begin{lem} \label{lem correct phase_a}
$\inf\limits_{\theta \in \intco{0,2 \pi}} \enVert{{\rm e}^{i \theta}\mvec{x} - \mvec{y}}$ achieves its minimum at a value of $\theta$ that makes ${\rm e}^{-i \theta}\mvec{x}^* \mvec{y}$ a positive real number, and for that $\theta$ we have 
\begin{align*}
 \enVert{{\rm e}^{i \theta} \mvec{x} - \mvec{y}}^2 &  = \enVert{\mvec{x}}^2 + \enVert{\mvec{y}}^2 - 2 \envert{\mvec{x}^* \mvec{y}}
 \\ \nonumber &
 = \intoo{\enVert{\mvec{x}} - \enVert{\mvec{y}}}^2 + 2 \intoo{\enVert{\mvec{x}} \enVert{\mvec{y}} - \envert{\mvec{x}^* \mvec{y}}}.
 \end{align*}
\end{lem}
\begin{proof}
Let $ \mvec{z} = {\rm e}^{i \theta} \mvec{x} $
\begin{align*}
\enVert{\mvec{z} - \mvec{y}}^2 & = \intoo{\mvec{z} - \mvec{y}}^* \intoo{\mvec{z} - \mvec{y}} 
\\ \nonumber &
= \enVert{\mvec{z}}^2 + \enVert{\mvec{y}}^2 - 2 \Re(\mvec{z}^* \mvec{y}) 
\\ \nonumber &
\geq \enVert{\mvec{z}}^2 + \enVert{\mvec{y}}^2 - 2\envert{\mvec{z}^* \mvec{y}} 
\\ \nonumber &
=  \enVert{\mvec{x}}^2 + \enVert{\mvec{y}}^2 - 2 \envert{\mvec{x}^* \mvec{y}}.
\end{align*}
Note that equality holds only when $ \Re(\mvec{z}^* \mvec{y}) = \envert{\mvec{z}^* \mvec{y}}$, which proves our claim. 
\end{proof}

\begin{lem} \label{correct phase}
 For any two vectors $\mvec{x}$ and $\hat{\mvec{x}}$ in $ \mathds{C}^n$, we have
\begin{align*}
 \frac{1}{8} \intoo{\inf\limits_{\theta} \enVert{{\rm e}^{i \theta}\mvec{x} - \hat{\mvec{x}}}^2}^2  &
\leq \frac{1}{2} \intoo{\enVert{\mvec{x}}^2 - \enVert{\hat{\mvec{x}}}^2}^2 +  \intoo{\enVert{\mvec{x}}^2 \enVert{\hat{\mvec{x}}}^2 - \envert{\mvec{x}^* \hat{\mvec{x}}}^2}.
\end{align*}
\end{lem}

\begin{proof}

Note that according to Lemma \ref{lem correct phase_a} we have
\begin{align*}
\intoo{ \intoo{\enVert{\mvec{x}} - \enVert{\hat{\mvec{x}}}}^2 + 2 \intoo{\enVert{\mvec{x}} \enVert{\hat{\mvec{x}}} - \envert{\mvec{x}^* \hat{\mvec{x}}} } }^2 & \leq
 \intoo{1 + 1} \intoo{\intoo{\enVert{\mvec{x}} - \enVert{\hat{\mvec{x}}}}^4 + 4 \intoo{\enVert{\mvec{x}} \enVert{\hat{\mvec{x}}} - \envert{\mvec{x}^* \hat{\mvec{x}}} }^2}  
\\ &
\leq 2 \intoo{\enVert{\mvec{x}}^2 - \enVert{\hat{\mvec{x}}}^2}^2 + 8 \intoo{\enVert{\mvec{x}}^2 \enVert{\hat{\mvec{x}}}^2 - \envert{\mvec{x}^* \hat{\mvec{x}}}^2 }.
\end{align*}
Hence,
\begin{align*}
\frac{1}{8} \intoo{\inf_\theta \enVert{{\rm e}^{i \theta} \mvec{x} - \hat{\mvec{x}}}^2}^2
& \leq 
 \frac{1}{4} \intoo{\enVert{\mvec{x}}^2 - \enVert{\hat{\mvec{x}}}^2}^2 +  \intoo{\enVert{\mvec{x}}^2 \enVert{\hat{\mvec{x}}}^2 - \envert{\mvec{x}^* \hat{\mvec{x}}}^2}
 \\ \nonumber &
\leq \frac{1}{2} \intoo{\enVert{\mvec{x}}^2 - \enVert{\hat{\mvec{x}}}^2}^2 +  \intoo{\enVert{\mvec{x}}^2 \enVert{\hat{\mvec{x}}}^2 - \envert{\mvec{x}^* \hat{\mvec{x}}}^2}.
\end{align*}
\end{proof}

\begin{lem} \label{e^u2 Phi(u) < 1} Let $\Phi(x)$ denote the CDF of a standard normal variable. Then, for any $u > 0$,
$$ g(u) = {\rm e}^{\frac{1}{u^2}} \Phi \intoo{- \frac{\sqrt{2}}{u}} \leq 1. $$
\end{lem}

\begin{proof}
With a  change of variable $v = \frac{\sqrt{2}}{u}$, proving $g(u) \leq 1$ is equivalent to proving
$h(v) = {\rm e}^{-\frac{v^2}{2}} - \Phi(-v) \geq 0$ for all $v \geq 0$.  We have
\begin{align*}
h'(v) = \intoo{-{v} + \frac{1}{\sqrt{2 \pi}}} {\rm e}^{- \frac{v^2}{2}} \implies \begin{dcases}
h'(v) \geq 0 & v \leq \frac{1}{\sqrt{ 2 \pi}}
\\
h'(v) < 0 & v > \frac{1}{\sqrt{ 2 \pi}}
\end{dcases}.
\end{align*}
In addition, $h(0) = \frac{1}{2} > 0$ and $ h(\infty) = 0$.
\end{proof}

\begin{lem}[Chi squared concentration] \label{chi squared Upper bound}
For any $\tau \geq 0$, we have
$$ \p{\chi^2(m) > m (1 + \tau)} \leq {\rm e}^{- \frac{m}{2} \intoo{\tau - \ln(1 + \tau)}}.$$
\end{lem}

The proof of this lemma can be found in \cite{com2com}.

\subsection{Heavy-tailed concentration}

In this section, we discuss a few lemmas regarding the concentration of heavy-tailed random variables. A more complete discussion of such concentration results can be found in \cite{Bakhshizadeh2020concentration}. 

\begin{milad}

\begin{lem}[Bounded random variable MGF upper bound]
\label{lem bdd iid mgf upper bound}
Let $X$ be a random variable and $c, c'$ positive constants such that 
\begin{equation*}
\p{ \abs{X - \e{X}} > \tau } \leq c' \expp{ - c \sqrt{\tau} } \quad
\forall \tau \geq 0.
\end{equation*}
Then there exist constants $c_2, c_3$, depending only on the distribution of $X$, such that for all $L \geq c_3$ and $\lambda = \frac{c}{2 \sqrt{L}}$
\begin{equation*}
\log \e{ \expp{\lambda \intoo{X \mathbf{1}_{X \leq L} - \e{X}}} } \leq \frac{c_2}{2} \lambda^2.
\end{equation*}
\end{lem}

\begin{proof}
For simplicity of the notation let $X_L \triangleq X \mathbf{1}_{X \leq L} $ denotes the truncated version of the $X$.
By Taylor expansion of exponential function at $\e{X}$, one can get

\begin{align*}
\expp{\lambda X_L} = \expp{\lambda \e{X}} + \lambda \intoo{X_L - \e{X}} \expp{\lambda \e{X}} + \frac{\lambda^2}{2} \intoo{X_L - \e{X}}^2 \expp{\lambda Y},
\end{align*}
where $Y$ is a random variable whose value is between $\e{X}$ and $X_L$.  Therefore
\begin{align} \label{eq:mgf bound proof 1}
\e{ \expp{\lambda \intoo{X_L - \e{X}}} } = 1 + \lambda \intoo{\e{X_L} - \e{X}} + \frac{\lambda^2}{2} \e{\intoo{X_L - \e{X}}^2 \expp{\lambda \intoo{Y - \e{X}}}}.
\end{align}
Note that $\e{X_L - X} \leq 0$ and $\log(1 + x) \leq x \quad \forall x \geq 0$, hence by \eqref{eq:mgf bound proof 1} we obtain
\begin{equation*}
\log \e{ \expp{ \lambda \intoo{X_L - \e{X}}} } \leq \frac{\lambda^2}{2} \e{\intoo{X_L - \e{X}}^2 \expp{\lambda \intoo{Y - \e{X}}}},
\end{equation*}
which means if for some $c_3$ we show 
\begin{equation} \label{eq:mgf bound proof c2 bound}
\sup_{L \geq c_3, \; \lambda = \frac{c}{2 \sqrt{L}}
} \e{\intoo{X_L - \e{X}}^2 \expp{\lambda \intoo{Y - \e{X}}}} \leq c_2,
\end{equation}
the Lemma holds with $c_2, c_3$.

Note that since $Y$ is bounded between $\e{X}$ and $X_L$ we get 
\begin{align} \nonumber
& \e{\intoo{X_L - \e{X}}^2 \expp{\lambda \intoo{Y - \e{X}}}} \leq
\\ & \label{eq:mgf bound proof max of two terms}
{\e{X_L - \e{X}}^2 + \e{\intoo{X_L - \e{X}}^2 \expp{\lambda \intoo{X_L - \e{X}}}}}.
\end{align}
Since $X_L \xrightarrow{L \to \infty} X$ and it is dominated by $X$, by the dominated convergence Theorem we get $\e{X_L - \e{X}}^2 \to Var (X)$, hence for $L > c_3'$ we have
\begin{equation} \label{eq:mgf bound proof variance bound}
\e{X_L - \e{X}}^2 \leq 2 Var (X).
\end{equation}
Now, we bound the second term in \eqref{eq:mgf bound proof max of two terms}.

\begin{align*}
\e{\intoo{X_L - \e{X}}^2 \expp{\lambda \intoo{X_L - \e{X}}}} & =
\e{\intoo{X_L - \e{X}}^2 \expp{\lambda \intoo{X_L - \e{X}}} \mathbf{1}_{X_L < \e{X}}}
\\ & \;\; +
\e{\intoo{X_L - \e{X}}^2 \expp{\lambda \intoo{X_L - \e{X}}} \mathbf{1}_{X_L \geq \e{X}}}.
\end{align*}

Note that
\begin{equation} \label{eq:mgf bound proof bound term 2,1}
\e{\intoo{X_L - \e{X}}^2 \expp{\lambda \intoo{X_L - \e{X}}} \mathbf{1}_{X_L < \e{X}}} \leq \e{X_L -\e{X}}^2 \leq 2 Var (X),
\end{equation}
for $L > c_3'$.

Moreover, if $U \define X_L - \e{X}$

\begin{align} \nonumber
\e  {U^2 \expp{\lambda U} \mathbf{1}_{U \geq 0}}  & =
\int_0^\infty \p{U^2 \expp{\lambda U} > u, \; U \geq 0 } du
\\ & = \nonumber
\int_0^\infty \p{X_L - \e{X} > t } du, \qquad t^2 \expp{\lambda t} = u,
\\ & = \nonumber
\int_0^{L - \e{X}} \p{X - \e{X} > t} \expp{\lambda t} \intoo{2t + \lambda t^2} dt
\\ & \leq \label{eq:mgf bound proof int ineq 1}
\int_0^{L - \e{X}} c' \expp{- c \sqrt{t} + \frac{c}{2 \sqrt{L}} t} \intoo{2t + \lambda t^2} dt.
\end{align}
Note that for large enough $L$ and $0 \leq t \leq L - \e{X} $, we have 
$ - c \sqrt{t} + \frac{c}{2 \sqrt{L}} t \leq -\frac{c \sqrt{t}}{3}$.  More specifically, 
\begin{align*}
- \sqrt{t} + \frac{t}{2 \sqrt{L}} \leq -\frac{\sqrt{t}}{3}, \qquad
\forall L \geq \frac{-9 \e{X}}{7}, \quad 0 \leq t \leq L - \e{X},
 \end{align*} 

hence by \eqref{eq:mgf bound proof int ineq 1} we obtain
\begin{align} \nonumber
\e{\intoo{X_L - \e{X}}^2 \expp{\lambda \intoo{X_L - \e{X}}} \mathbf{1}_{X_L \geq \e{X}}} & \leq
c' \int_0^{L - \e{X}} \expp{- \frac{c \sqrt{t}}{3}} (2t + \lambda t^2) dt
\\ & \leq \nonumber
c' \int_0^{\infty} \expp{- \frac{c \sqrt{t}}{3}} (2t + \lambda t^2) dt
\\ & = \nonumber
c' \int_0^\infty  \exp(-z) \intoo{\frac{18}{c^2} z^2 + \frac{81 \lambda}{c^4} z^4} \frac{18 z}{c^2} dz
\\ & = \nonumber
\frac{18^2 c'}{c^4} \Gamma(4) + \frac{18 c' \times 81 \lambda}{c^6} \Gamma(6)
\\ & = \nonumber
\frac{18^2 c'}{c^4} \Gamma(4) + \frac{9 c' \times 81}{c^5 \sqrt{L}} \Gamma(6)
\\ & \leq \label{eq:mgf bound proof bound term2,2}
\frac{324 c'}{c^4} \Gamma(4) + \frac{729 c'}{c^5 \sqrt{c_3}} \Gamma(6), \qquad \forall L \geq c_3.
\end{align}
Hence. if we set $c_2 = 4 Var(X) + \frac{324 c'}{c^4} \Gamma(4) + \frac{729 c'}{c^5 \sqrt{c_3}} \Gamma(6)$ and $c_3 \geq c_3'$, \eqref{eq:mgf bound proof variance bound}, \eqref{eq:mgf bound proof bound term 2,1}, \eqref{eq:mgf bound proof bound term2,2} yield \eqref{eq:mgf bound proof c2 bound} which concludes the proof.
\end{proof}

\end{milad}

\begin{milad}
\begin{lem}[Heavy tail concentration] \label{lem heavy tail concentration real}
Let $ \cbr{Y_k}_{k \in \mathds{N}} $ be i.i.d. random variables.  Assume there are constants $c > 0, \; c' \geq 1$ such that $ \p{\envert{Y_{k} - \e{Y_k}} > \tau} \leq c' {\rm e}^{- c \sqrt{\tau}}$, for all $\tau > 0$ .  Then, there exist a positive constant $C_3> 0$, such that, for every $ \epsilon >  C_3 m^{- \frac{1}{3}}$, 
\begin{equation}
\p{ \envert{ \frac{1}{m} \sum_{k = 1}^m  Y_k - \e{ Y_k }} > \epsilon } \leq 4 {\rm e}^{ - \frac{c}{2} \sqrt{ m \epsilon} }.
\end{equation}

\end{lem}

\begin{proof}
Following the notion in the proof of the Lemma \ref{lem bdd iid mgf upper bound}, let $Y_k^L = Y_k \mathbf{1}_{Y_k \leq L}$.  Then,

\begin{align} \nonumber
\p{\frac{1}{m} \sum_{k = 1}^m Y_k - \e{Y_k} > \epsilon} & \leq
\p{\frac{1}{m} \sum_{k = 1}^m Y_k^L - \e{Y_k} > \epsilon} + \p{ \exists k, \; Y_k > L }
\\ \nonumber &
\leq
\expp{- \lambda \epsilon} \e{\expp{\frac{\lambda}{m} (Y_1^L - \e{Y_1})}}^m + m \p{Y_1 - \e{Y_1} > L - \e{Y_1}}
\\  & \leq \label{eq:concentration proof probability bound}
\expp{- \lambda \epsilon + \frac{c_2}{2} \frac{\lambda^2}{m^2} m} + m c' \expp{-c \sqrt{L - \e{Y_1}}},
\end{align}
for $\frac{\lambda}{m} = \frac{c}{2 \sqrt{L}}$.  Note that the last inequality is obtained by the Lemma \ref{lem bdd iid mgf upper bound}.  Set $L = m \epsilon$ and hence $\lambda =\frac{c \sqrt{m}}{2 \sqrt{\epsilon}} $, then by \eqref{eq:concentration proof probability bound} we have
\begin{equation} \label{eq:concentration proof probability bound 2}
\p{\frac{1}{m} \sum_{k = 1}^m Y_k - \e{Y_k} > \epsilon} \leq
\expp{- c \sqrt{m \epsilon} + \frac{c_2 c^2}{8} \frac{1}{\epsilon}} + m c' \expp{-c \intoo{\sqrt{m \epsilon} - \sqrt{\abs{\e{Y_1}}}}}.
\end{equation} 
Note that if $\epsilon \geq \intoo{\frac{c_2 c}{4}}^{\frac{2}{3}} m^{-\frac{1}{3}}$, we have $\frac{c_2 c^2}{8} \frac{1}{\epsilon} \leq \frac{c}{2} \sqrt{m \epsilon}$, hence, 
\begin{equation} \label{eq:concntration proof bound term 1}
\expp{- c \sqrt{m \epsilon} + \frac{c_2 c^2}{8} \frac{1}{\epsilon}} \leq \expp{- \frac{c}{2} \sqrt{m \epsilon}}.
\end{equation}
Furthermore,
\begin{align} \label{eq:concentration proof bound term 2}
m c' \expp{-c \intoo{\sqrt{m \epsilon} - \sqrt{\abs{\e{Y_1}}}}} = \expp{ \log c' + \log m + c \sqrt{\abs{\e{Y_1}}} - c \sqrt{m \epsilon}} \leq \expp{- \frac{c}{2} \sqrt{m \epsilon}},
\end{align}
whenever
\begin{equation} \label{eq:concentration proof bound for log term}
\log c' + \log m + c \sqrt{\abs{\e{Y_1}}} \leq \frac{c}{2} \sqrt{m \epsilon}.
\end{equation}
since $m \epsilon \geq C_3 m^{\frac{2}{3}}$ for $\epsilon \geq C_3 m^{- \frac{1}{3}}$ we have 
\begin{equation}
\frac{c}{2} \sqrt{m \epsilon} \geq \frac{c \sqrt{C_3}}{2} m^{\frac{1}{3}}.
\end{equation}

Given $m^{\frac{1}{3}}$ grows faster than $\log m$, by choosing large enough $C_3$ we can make \eqref{eq:concntration proof bound term 1} and \eqref{eq:concentration proof bound term 2} hold for all integer $m$, thus we obtain
\begin{equation} \label{eq:concentration proof right side bound}
\p{\frac{1}{m} \sum_{k = 1}^m Y_k - \e{Y_k} > \epsilon} \leq 2 \expp{- \frac{c}{2} \sqrt{m \epsilon}}.
\end{equation}

By repeating the exact same line of the proof for $- Y_k$ instead of $Y_k$ we can obtain
\begin{equation} \label{eq:concentration proof left side bound}
\p{\frac{1}{m} \sum_{k = 1}^m Y_k - \e{Y_k} < - \epsilon} \leq 2 \expp{- \frac{c}{2} \sqrt{m \epsilon}}.
\end{equation}

Combining \eqref{eq:concentration proof right side bound} and \eqref{eq:concentration proof left side bound} yields

\begin{equation}
\p{ \abs{\frac{1}{m} \sum_{k = 1}^m Y_k - \e{Y_k}} > \epsilon} \leq 4 \expp{- \frac{c}{2} \sqrt{m \epsilon}}.
\end{equation}

\end{proof}

\end{milad}

\subsection{Properties of $d_A(\cdot,\cdot)$} \label{Properties of d}

\begin{lem}\label{lem:lambdasVSvectors}
If $\lambda_1(\mvec{c})$ and $\lambda_2(\mvec{c})$ denote the two non-zero eigenvalues of $\mvec{x} \mvec{x}^* - \mvec{c} \mvec{c}^*$, then we have
\begin{enumerate}
\item $\lambda_1(\mvec{c}) + \lambda_2(\mvec{c}) = \enVert{\mvec{x}}^2 - \enVert{\mvec{c}}^2$. 
\item $\lambda_1(\mvec{c})^2 + \lambda_2(\mvec{c})^2 = \intoo{\enVert{\mvec{x}}^2 - \enVert{\mvec{c}}^2 }^2 + 2 \intoo{\enVert{\mvec{x}}^2 \enVert{\mvec{c}}^2 - \envert{\mvec{x}^*\mvec{c}}^2 }$.
\item $\lambda_1(\mvec{c}) \lambda_2(\mvec{c}) = \intoo{\envert{\mvec{x}^* \mvec{c}}^2 - \enVert{\mvec{x}}^2 \enVert{\mvec{c}}^2} \leq 0$
\end{enumerate}
\end{lem}
\begin{proof}
First note that 
\begin{align}
\label{sum lambdas}
\lambda_1(\mvec{c}) + \lambda_2(\mvec{c}) &= {\rm Tr} (\mvec{x}\mvec{ x}^* - \mvec{c}\mvec{ c}^*) = \enVert{\mvec{x}}^2 - \enVert{\mvec{c}}^2.
\end{align}
Similarly,
\begin{align}
\lambda_1(\mvec{c})^2 + \lambda_2(\mvec{c})^2  = &{\rm Tr} (\mvec{x} \mvec{x}^* -\mvec{ c}\mvec{ c}^*)^2 \nonumber
\\ \nonumber &
 = {\rm Tr} (\mvec{x}\mvec{ x}^* \mvec{x}\mvec{ x}^*) + {\rm Tr} ( \mvec{c}\mvec{ c}^* \mvec{c}\mvec{ c}^* ) 
- {\rm Tr} (\mvec{c}\mvec{c}^* \mvec{x}\mvec{ x}^*) - {\rm Tr} (\mvec{x}\mvec{x}^* \mvec{c}\mvec{ c}^*) 
\\ & \nonumber
= \enVert{\mvec{x}}^4 + \enVert{\mvec{c}}^4  - 2 \envert{\mvec{x}^* \mvec{c}}^2
\\  & \label{sum lambda squared} 
=\intoo{\enVert{\mvec{x}}^2 - \enVert{\mvec{c}}^2 }^2 + 2 \intoo{\enVert{\mvec{x}}^2 \enVert{\mvec{c}}^2 - \envert{x^*c}^2 }.
\end{align}
Finally,
\begin{align*}
2 \lambda_1(\mvec{c}) \lambda_2(\mvec{c}) &= (\lambda_1(\mvec{c}) + \lambda_2(\mvec{c}))^2 - (\lambda_1(\mvec{c})^2 + \lambda_2(\mvec{c})^2) 
\\ \nonumber &
= 2 \intoo{\envert{\mvec{x}^* \mvec{c}}^2 - \enVert{\mvec{x}}^2 \enVert{\mvec{c}}^2} 
\\ \nonumber &
\leq 0.
\end{align*}
\end{proof}

\begin{lem} \label{upper boumd for MGF} Let $Z= (\lambda_1(\mvec{c}) U + \lambda_2(\mvec{c})V )^2$, where $U$ and $V$ are independent $\chi^2(2)$. Then, for any $\alpha>0$, we have
\begin{align*}
f(\alpha) \triangleq 
 \e{{\rm e}^{- \alpha Z}} \leq \intoo{\frac{\pi}{\lambda_{\max}(\mvec{c})^2 \alpha}}^{\frac{1}{2}}.
\end{align*}
\end{lem}

\begin{proof}

\begin{align}
f(\alpha)  = \int_{x,y \geq 0} {\rm e}^{- \alpha \intoo{\lambda_1(\mvec{c}) x + \lambda_2(\mvec{c}) y}^2} \frac{{\rm e}^{-\frac{x}{2}}}{2} \frac{{\rm e}^{-\frac{y}{2}}}{2} dx dy.
\end{align}
Consider changing the variable $(x,y)$ in the above integral to $(u,v)$ defined as 
\[
(u,v) = \intoo{\lambda_1(\mvec{c}) x + \lambda_2(\mvec{c}) y,\frac{x + y}{2}}.
\]
The determinent of the Jacobian of this mapping is given by
\begin{align}
\envert{ \frac{\partial u,v}{\partial x,y} } 
& = \begin{vmatrix}
\lambda_1(\mvec{c}) & \lambda_2(\mvec{c})
\\ \nonumber
\frac{1}{2} & \frac{1}{2} 
\end{vmatrix}
= \frac{\lambda_1(\mvec{c}) - \lambda_2(\mvec{c})}{2}.
\end{align}
Furthermore,  
\[
v - \frac{u}{2 \lambda_2(\mvec{c})} = \intoo{\frac{1}{2} + { \frac{\lambda_1(\mvec{c})}{-2 \lambda_2(\mvec{c})}}} x.
\]
Since $\frac{\lambda_1(\mvec{c})}{-2 \lambda_2(\mvec{c})} > 0$, we have
\[
 x \geq 0 \iff v \geq \frac{u}{2 \lambda_2(\mvec{c})}.
\]
Similarly,
\[
 v \geq \frac{u}{2 \lambda_1(\mvec{c})}.
\]
Therefore, 
\begin{align}
\nonumber
f(\alpha) & = \frac{2}{4 (\lambda_1(\mvec{c}) - \lambda_2(\mvec{c}))}\int \int_{v \geq \frac{u}{2 \lambda_1(\mvec{c})}, v \geq \frac{u}{2 \lambda_2(\mvec{c})}} {\rm e}^{- \alpha u^2} {\rm e}^{-v} dv du
\\ \nonumber
& =
\frac{1}{2(\lambda_1(\mvec{c}) - \lambda_2(\mvec{c}))} \int_{u \geq 0} \int_{v = \frac{u}{2 \lambda_1(\mvec{c})}}^\infty {\rm e}^{- \alpha u^2} {\rm e}^{-v} dv du + \frac{1}{2(\lambda_1(\mvec{c}) - \lambda_2(\mvec{c}))} \int_{u < 0} \int_{v = \frac{u}{2 \lambda_2(\mvec{c})}}^\infty {\rm e}^{- \alpha u^2} {\rm e}^{-v} dv du 
\\ \nonumber
& =
 \frac{1}{2(\lambda_1(\mvec{c}) - \lambda_2(\mvec{c}))}  \int_{u = 0}^\infty  {\rm e}^{- \alpha u^2 - \frac{u}{2 \lambda_1(\mvec{c})}} du
 + \frac{1}{2(\lambda_1(\mvec{c}) - \lambda_2(\mvec{c}))} \int_{u = - \infty}^0 {\rm e}^{- \alpha u^2 - \frac{u}{2 \lambda_2(\mvec{c})}} du  
 \\ \nonumber
 & =
 \frac{{\rm e}^{\frac{1}{16 \lambda_1(\mvec{c})^2 \alpha}}}{2(\lambda_1(\mvec{c}) - \lambda_2(\mvec{c}))}   \int_{u = 0}^\infty  {\rm e}^{- \alpha (u + \frac{1}{4 \lambda_1 \alpha})^2} du 
+ \frac{{\rm e}^{\frac{1}{16 \lambda_2(\mvec{c})^2 \alpha}}}{2(\lambda_1(\mvec{c}) - \lambda_2(\mvec{c}))}  \int_{u = -\infty}^0  {\rm e}^{- \alpha (u + \frac{1}{4 \lambda_2(\mvec{c}) \alpha})^2} du
\\  \label{MGF of Z_1}
& = 
 \frac{\sqrt{\pi}}{2(\envert{\lambda_1(\mvec{c})} + \envert{\lambda_2(\mvec{c})}) \sqrt{\alpha}}  {\rm e}^{\frac{1}{16 \lambda_1(\mvec{c})^2 \alpha}} \Phi \intoo{-\frac{\sqrt{2}}{4 \envert{\lambda_1(\mvec{c})} \sqrt{\alpha}}} 
+  \frac{\sqrt{\pi}}{2(\envert{\lambda_1(\mvec{c})} + \envert{\lambda_2(\mvec{c})}) \sqrt{\alpha}} {\rm e}^{\frac{1}{16 \lambda_2(\mvec{c})^2 \alpha}} \Phi \intoo{- \frac{\sqrt{2}}{4 \envert{\lambda_2(\mvec{c})} \sqrt{\alpha}}},
\end{align}
where $
\Phi(x)={1\over \sqrt{2\pi}}\int_{x}^{\infty} {\rm e}^{-{1\over 2}u^2}.
$
According to Lemma \ref{e^u2 Phi(u) < 1} we have
\begin{align}\label{eq:upperboundmgflast}
{\rm e}^{\frac{1}{16 \lambda_1(\mvec{c})^2 \alpha}} \Phi \intoo{- \frac{\sqrt{2}}{4 \envert{\lambda_1(\mvec{c})} \sqrt{\alpha}}}  = g
\intoo{{4 \envert{\lambda_1(\mvec{c})} \sqrt{\alpha}}}  \leq 1. 
\end{align}
Hence, by combining \eqref{MGF of Z_1}, \eqref{eq:upperboundmgflast}, and the fact that $\frac{1}{\envert{\lambda_1(\mvec{c})} + \envert{\lambda_2(\mvec{c})} } \leq \frac{1}{\envert{\lambda_{\max}(\mvec{c})}}$ we can complete the proof. 
\end{proof}

\begin{r2}
Theorem below is showing how the distance function $d_A$ concentrates when we have sufficient measurements.  
\end{r2}

\begin{thm}[Concentration of $d_A(\cdot,\cdot)$] \label{thm:concentration of d}
Let $\mathcal{C}_r$ denote the set of codewords at rate $r$, and $\mvec{x}$ denotes the signal of interest.  For a given $\mvec{c} \in \mathds{C}^n$, let $\lambda^2_{\min} ( \mvec{c} ) \leq \lambda^2_{\max} (\mvec{c})$ be squared of the two non-zero eigenvalues of $\mvec{x}\mvec{ x}^* - \mvec{c}\mvec{ c}^*$.  For any positive real numbers $\tau_1,\tau_2$,
\begin{equation} \label{eq 1:thm:concentration}
\p{d_A(\envert{A \mvec{x}} , \envert{A \mvec{c}}) > \lambda_{\max}^2(\mvec{c}) \tau_1,\; \forall \mvec{c} \in C_r}
\geq 1 -  2^r {\rm e}^{\frac{m}{2} \intoo{K + \ln \tau_1 - \ln m} },
\end{equation}
where $  K = \ln 2 \pi {\rm e}$ and
\begin{align} \label{eq 2:thm:concentration}
\p{d_A( \envert{A \mvec{x} },\envert{A \mvec{c}} ) < \lambda_{\max}^2(\mvec{c}) \intoo{4m (1 + \tau_2)}^2}
 \geq 1 - {\rm e}^{-2m \intoo{\tau_2 - \ln (1 + \tau_2)}}.
\end{align}
\end{thm}

\begin{proof}
Recall from (\ref{def d_A(Ax,Ac)}) that 
\begin{align}
d_A(\envert{A \mvec{x}}, \envert{A \mvec{c}}) = \sum_{k = 1}^m \intoo{{\mvec{a}_k}^* (\mvec{x} \mvec{x}^* - \mvec{c} \mvec{c}^*){\mvec{a}_k}}^2.\label{eq:d-Ax-Ac}
\end{align}
First, for fixed $ \mvec{x}$ and $ \mvec{c}$, we derive the distribution and the moment-generating function (mgf) of $d_A(\envert{A \mvec{x}}, \envert{A \mvec{c}})$.
Note that $\mvec{x} \mvec{x}^* - \mvec{c} \mvec{c}^*$ is a Hermitian matrix of rank at most two, and therefore  it can be written as
\begin{equation} \label{eq:xx*-decomposition}
\mvec{x} \mvec{x}^* - \mvec{c} \mvec{c}^* = Q^T 
\begin{pmatrix}
\lambda_1(\mvec{c})
\\
& \lambda_2(\mvec{c})
\\ 
& & \ddots
\\
& & & 0
\end{pmatrix} 
\bar{Q},
\end{equation}
where $Q^T \bar{Q} = I_n$.  Combining \eqref{eq:d-Ax-Ac} and \eqref{eq:xx*-decomposition}, we have
\begin{align*}
\sum_{k = 1}^m \intoo{{\mvec{a}_k}^* (\mvec{x} \mvec{x}^* - \mvec{c} \mvec{c}^*){\mvec{a}_k}}^2 
& =  \sum_{k = 1}^m \intoo{{\mvec{a}_k}^*  Q^T 
\begin{pmatrix}
\lambda_1(\mvec{c})
\\
& \lambda_2(\mvec{c})
\\ 
& & \ddots
\\
& & & 0
\end{pmatrix} 
\bar{Q} {\mvec{a}_k}}^2
\\ \nonumber &
= \sum_{k = 1}^m \intoo{\mvec{B}_k^*
\begin{pmatrix}
\lambda_1(\mvec{c})
\\
& \lambda_2(\mvec{c})
\\ 
& & \ddots
\\
& & & 0
\end{pmatrix} 
\mvec{B}_k}^2  
\\
& =  \sum_{k = 1}^m \intoo{ \lambda_1(\mvec{c}) \envert{B_{k,1}}^2 + \lambda_2(\mvec{c}) \envert{B_{k,2}}^2 }^2,
\end{align*}
where $\mvec{B}_k = \bar{Q} {\mvec{a}_k}$.
%Hence, using the last line of the above equation,
%\[
%\sum_{k = 1}^m \intoo{{\mvec{a}_k}^* (\mvec{x} \mvec{x}^* - \mvec{c} \mvec{c}^*){\mvec{a}_k}}^2 
%\textcolor{red}{ \sim  \sum_{k = 1}^m \intoo{ \lambda_1(\mvec{c}) \chi^2(2) + \lambda_2(\mvec{c}) \chi^2(2) }^2.}
% \]
%By Lemma  \ref{A u 2 distribution},
Since $\bar{Q}$ is an orthonormal matrix,   $B = \bar{Q} \bar{A}$ has the same distribution as $A$, and therefore the $\chi^2$ variables in the above sum are all independent. Let $Z_k = \intoo{ \lambda_1(\mvec{c}) \envert{B_{k,1}}^2 + \lambda_2(\mvec{c}) \envert{B_{k,2}}^2 }^2$. Then we have
\begin{gather} \label{Zk dist}
d_A( \envert{A\mvec{x}},\envert{A \mvec{c}} ) = \sum_{i = 1}^m Z_i,
\end{gather}
where  $Z_1,\ldots,Z_m$ are i.i.d.~as $(\lambda_1(\mvec{c}) U + \lambda_2(\mvec{c})V )^2$, where $U$ and $V$ are independent $\chi^2(2)$ random variables. 
Define $\mathbf{\lambda}_{\min}(\mvec{c}), \lambda_{\max}(\mvec{c})$ to denote $\lambda_1(\mvec{c}),\lambda_2(\mvec{c})$ with smaller and larger absolute value respectively, i.e.,
$$ \envert{{\lambda}_{\min}(\mvec{c})} = \min \cbr{\envert{\lambda_1(\mvec{c})},\envert{\lambda_2(\mvec{c})}}, \quad \envert{{\lambda}_{\max}(\mvec{c})} = \max \cbr{\envert{\lambda_1(\mvec{c})},\envert{\lambda_2(\mvec{c})}}. $$ 
To derive \eqref{eq 1:thm:concentration}, note that according to Lemma \ref{upper boumd for MGF} for any $\alpha>0$, we have
\begin{align*}
\p{d_A(\envert{A \mvec{x}} , \envert{A \mvec{c}}) \leq t}  & =
\p{{\rm e}^{-\alpha \sum\limits_{i = 1}^m  Z_i} \geq {\rm e}^{ - \alpha t}}
\\ \nonumber &
\leq {\rm e}^{\alpha t} \e{{\rm e}^{- \alpha Z_1}}^m
\\ \nonumber &
 \leq {\rm e}^{\alpha t} f(\alpha)^m 
 \\ \nonumber &
\leq {\rm e}^{\alpha t} \intoo{\frac{\pi}{\lambda_{\max}(c)^2 \alpha}}^{\frac{m}{2}},
\end{align*}
where $\alpha > 0$ is a free parameter.  Let $\alpha = \frac{m}{2 \lambda_{\max}^2(\mvec{c}) \tau_1}$ and $t = \lambda_{\max}^2(\mvec{c}) \tau_1$. Therefore, 
\begin{align*} 
\p{d_A(\envert{A \mvec{x}} , \envert{A \mvec{c}}) \leq \lambda_{\max}^2 (\mvec{c}) \tau_1} & \leq {\rm e}^{\frac{m}{2}} \intoo{\frac{2 \pi \tau_1}{m}}^{\frac{m}{2}}
\\ \nonumber &
\leq {\rm e}^{\frac{m}{2} \intoo{K + \ln \tau_1 - \ln m}},
\end{align*}
where $K = \ln 2 \pi e $. Hence, we have
\begin{align*} 
\p{d_A(\envert{A \mvec{x}} , \envert{A \mvec{c}}) > \lambda_{\max}^2(\mvec{c}) \tau_1} \geq 1 -  {\rm e}^{\frac{m}{2} \intoo{K + \ln \tau_1 - \ln m} },
\end{align*}
and with an union bound on $C_r$ we get
\begin{align}
\nonumber 
\p{d_A(\envert{A \mvec{x}} , \envert{A \mvec{c}}) > \lambda_{\max}^2(\mvec{c}) \tau_1 \quad \forall \mvec{c} \in C_r}
\geq 1 -  2^r {\rm e}^{\frac{m}{2} \intoo{K + \ln \tau_1 - \ln m}}.
\end{align}
To prove \eqref{eq 2:thm:concentration}, note that for $Z_i$ defined in  (\ref{Zk dist}), one has $Z_i \leq \intoo{\envert{\lambda_{\max}(\mvec{c})} \chi^2(4)}^2$, thus
\begin{align*}
\nonumber 
 \sum_{i = 1}^m Z_i & \leq \lambda_{\max}^2 \sum_{i = 1}^m \chi^4(4)
 \leq \lambda_{\max}^2 \intoo{\sum_{i= 1}^m \chi^2(4)}^2 
 \overset{d}{=} \lambda_{\max}^2 \intoo{\chi^2(4m)}^2,
\end{align*}
where the notation $ \overset{d}{=}$ implies that they have the same distributions. Therefore, by Lemma \ref{chi squared Upper bound} we have
\begin{align*} 
\p{d_A( \envert{A \mvec{x}},\envert{A \mvec{c}} )  \geq \lambda_{\max}^2(\mvec{c}) \intoo{4m (1 + \tau_2)}^2}
 & = \p{\sum_{i = 1}^m Z_i \geq \lambda_{\max}^2 \intoo{4m (1 + \tau_2)}^2}
\\ &
\leq \p{\chi^2(4m) \geq 4m(1 + \tau_2)}
\\ \nonumber &
\leq
 {\rm e}^{-2m \intoo{\tau_2 - \ln (1 + \tau_2)}}.
\end{align*}
Hence, for any $  \tau_2 > 0$, we have
\begin{align*} 
\p{d_A( \envert{A \mvec{x} },\envert{A \mvec{c}} ) < \lambda_{\max}^2(\mvec{c}) \intoo{4m (1 + \tau_2)}^2} 
 \geq 1 - {\rm e}^{-2m \intoo{\tau_2 - \ln (1 + \tau_2)}}.
\end{align*}
\end{proof}

\begin{rem}[Expectation of $d_A(.,.)$]
Note that \eqref{Zk dist} implies

\begin{equation}  \label{eq: d expectation}
\e{d \intoo{\envert{A \mvec{x}}, \envert{A \mvec{c}}} } = 8 m  \intoo{\lambda_1(\mvec{c})^2 + \lambda_2(\mvec{c})^2 + \lambda_1(\mvec{c}) \lambda_2(\mvec{c})}
\end{equation}

\begin{proof}
By \eqref{Zk dist} we obtain
\begin{align*}
\e{d \intoo{\envert{A \mvec{x}}, \envert{A \mvec{c}}} }  &= m \e{Z_1} 
\\ &
= m \intoo{ \lambda_1(\mvec{c})^2  \e{U^2}   + \lambda_2(\mvec{c})^2 \e{V^2} + 2 \lambda_1(\mvec{c}) \lambda_2(\mvec{c})\e{U V} } 
\\ &
= m \intoo{ 8 \lambda_1(\mvec{c})^2     + 8 \lambda_2(\mvec{c})^2 + 2 \times 4 \lambda_1(\mvec{c}) \lambda_2(\mvec{c})} 
\\ &
=  8 m  \intoo{\lambda_1(\mvec{c})^2 + \lambda_2(\mvec{c})^2 + \lambda_1(\mvec{c}) \lambda_2(\mvec{c})}
\end{align*}
\end{proof}
\end{rem}

\subsection{Concentration of the gradient}\label{ssec:concgradient}

\begin{lem} \label{lem concentration of gradient}
Let $\mvec{v} \in \mathds{C}^{n}$ with $\enVert{\mvec{v}} = 1$ and $ \mvec{z} \in \C_r $ be fixed.  Then there exist constants $C_1, C_2, C_3 > 0$ such that,
\begin{equation}
\p{ \envert{ \Re \intoo{  \mvec{v}^* \intoo{ \nabla d_A(\mvec{z}) - \e{\nabla d_A(\mvec{z})}} } } > m \epsilon \inf_{\theta \in \mathds{R}} \enVert{{\rm e}^{i \theta} \mvec{x} - \mvec{z} } } \leq  C_2 {\rm e}^{ - C_1 \sqrt{m \epsilon}},
\qquad \forall \; \epsilon \geq C_3 m^{-\frac{1}{3}}.
\end{equation}
\end{lem}

\begin{proof}
In this proof, we will use the notations we introduced in \eqref{eq:xx*-decomposition} in the proof of Theorem \ref{thm:concentration of d}. Since we have assumed that  for any codeword $\mvec{c}$, $\|\mvec{x}\|_2 = \|\mvec{c} \|_2=1$, according to Lemma \ref{lem:lambdasVSvectors}, $\lambda_1 (\mvec{c}) + \lambda_2(\mvec{c}) =0$. Hence,
\begin{align} \nonumber
\nabla d_A(\mvec{z}) 
& = 2 \sum_{k = 1}^m \intoo{ \envert{\mvec{a}_k^* \mvec{z}}^2 - \envert{ \mvec{a}_k^* \mvec{x} } ^2 } \mvec{a}_k \mvec{a}_k^* \mvec{z},
\\ & = \nonumber
2 \sum_{k = 1}^m \mvec{a}_k \mvec{a}_k^* \mvec{z} \intoo{\bar{Q} \mvec{{\mvec{a}_k}}}^* \begin{pmatrix}
- \lambda_1(\mvec{z})
\\
&  \lambda_1 (\mvec{z})
\\
& & \ddots
\\
& & & 0
\end{pmatrix}
\bar{Q} \mvec{a}_k
\\ & = \label{eq gradient in term of U_k }
2 \lambda_1 (\mvec{z}) \sum_{k = 1}^m \mvec{a}_k \mvec{a}_k^* \mvec{z} U_k,
\end{align}
where $\lambda_i(\mvec{z}), Q$ are as defined in \eqref{eq:xx*-decomposition}, and
\begin{equation} \label{def U_k}
U_k \triangleq  \intoo{ \envert{ \intoo{\bar{Q} \mvec{a}_k}_2 }^2 - \envert{ \intoo{\bar{Q }\mvec{a}_k}_1 }^2}.
\end{equation}
It is straightforward to check that
\begin{align}
\e{ \nabla d_A(\mvec{z}) } 
= 8m ( \mvec{z} \mvec{z}^* - \mvec{x} \mvec{x}^* ) z. 
\end{align}
We also have
$$ \lambda_1(\mvec{z}) = - \lambda_2 (\mvec{z}) = \lambda_{\max} (\mvec{z}). $$
By \eqref{eq gradient in term of U_k } we have,
\begin{align} \nonumber
\Re \intoo{ \mvec{v}^* \intoo{\nabla d_A(\mvec{z}) - \e{\nabla d_A(\mvec{z})}}}
& =  2 \lambda_{\max}(\mvec{z}) \sum_{k = 1}^m \Re \intoo{ ( \mvec{v}^* \mvec{a}_k) (\mvec{a}_k^* \mvec{z}) U_k - \e{( \mvec{v}^* \mvec{a}_k) (\mvec{a}_k^* \mvec{z}) U_k}}
\\ & = \nonumber
2 \lambda_{\max}(\mvec{z}) \sum_{k = 1}^m \Re \intoo{ ( \mvec{v}^* \mvec{a}_k) (\mvec{a}_k^* \mvec{z}) U_k} - \e{\Re \intoo{ ( \mvec{v}^* \mvec{a}_k) (\mvec{a}_k^* \mvec{z}) U_k}}
\\ & = \label{lem gradient concetration proof 1}
2 \lambda_{\max} (\mvec{z}) \sum_{k = 1}^m Y_k - \e{Y_k},
\end{align}
where $ Y_k =  \Re \intoo{ ( \mvec{v}^* \mvec{a}_k) (\mvec{a}_k^* \mvec{z}) U_k} $.  We claim $Y_k$ satisfies all assumptions of Lemma \ref{lem heavy tail concentration real}.  To prove this note that since $ \enVert{\mvec{v}} = \enVert{\mvec{z}} = 1 $, all $ \mvec{v}^* \mvec{a}_k, \; \mvec{a}_k^* \mvec{z}, \; (\bar{Q} \mvec{a}_k)_ 1, \; (\bar{Q} \mvec{a}_k)_ 2$ have the same distribution as $\N(0, 1) + i \N(0, 1)$.  Therefore, $Y_k$ can be written as
\begin{equation} \label{Y_k as product of 4 Guassian}
Y_k = \sum_{j = 1}^{16} W_{1, j, k} W_{2, j, k} W_{3, j, k} W_{4, j, k}, 
\quad W_{l, j, k} \sim \N(0, 1) \quad 1 \leq l \leq 4, \; 1 \leq j \leq 16, \; 1 \leq k \leq m.
\end{equation}

\begin{milad}
We should emphasize that $W_{1, j, k}, W_{2, j, k}, W_{3, j, k}, W_{4, j, k}$ may be dependent on each other but are independent of $W_{1, j, k'}, W_{2, j, k'}, W_{3, j, k'}, W_{4, j, k'}$, if $k \neq k'$. 
Hence, we have
%\begin{align} \nonumber 
%\enVert{Y_k}_p
%& \leq 
%\sum_{j = 1}^{16}  \enVert{W_{1, j, k} W_{2, j, k} W_{3, j, k} W_{4, j, k}}_p
%\\ & =  \nonumber
%\sum_{j = 1}^{16} \intoo{ \e{\envert{W_{1, j, k} W_{2, j, k} W_{3, j, k} W_{4, j, k}}^p } }^{\frac{1}{p}}
%\\ & \leq \nonumber
%\sum_{j = 1}^{16} \intoo{\frac{1}{4}  \e{ \envert{W_{1, j, k}} ^{4p} + \envert{ W_{2, j, k}}^{4 p}  + \envert{W_{3, j, k}}^{4p} + \envert{W_{4, j, k}}^{4p} } }^{\frac{1}{p}}
%\\ & \leq \nonumber
%16  \e{ |W_{1, j, k}|^{4p} }^{\frac{1}{p}}
%\\ & \leq \nonumber
%16  (c p^{2 p})^{\frac{1}{p}}
%\\ & \leq \label{lem gradient concentration proof 2}
%c' p^2.
%\end{align}

%Moreover,
\begin{align} \nonumber
\p{\envert{Y_k} > \tau}
& \leq \nonumber
\p{\exists j \in \cbr{1, \ldots, 16};  \quad \envert{W_{1, j, k} W_{2, j, k} W_{3, j, k} W_{4, j, k}} > \frac{\tau}{16} }
\\ & \leq \nonumber
\p{ \exists j \in \cbr{1,\ldots, 16},  l \in \cbr{1, \ldots, 4}; \quad  \envert{W_{l, j, k}} > \sqrt[4]{\frac{\tau}{16}}}
\\ & \leq \nonumber
16 \times 4 \times {\rm e}^{ - \frac{1}{c^2} \sqrt{\frac{\tau}{16}} }
\\ & \leq  \label{lem gradient concentration proof 3}
64 {\rm e}^{- c' \sqrt{{\tau}}}.
\end{align}

To have \eqref{lem gradient concentration proof 3}, one may choose $ c ' = \frac{1}{4 c^2}$,
where $c$ is a constant for which
$ \p{\envert{\N(0, 1)} > \tau} \leq {\rm e}^{- \frac{\tau^2}{c^2}} $.

Hence, by Lemma \ref{lem heavy tail concentration real}, there exist constants $C_3$ such that
\begin{equation}
\p{ \envert{\sum_{k = 1}^m Y_k - \e{Y_k} } > m \frac{\epsilon}{2} } \leq 4 {\rm e}^{- \frac{c'}{2} \sqrt{ m \epsilon}}, \quad \forall \; \epsilon \geq C_3 m^{- \frac{1}{3}}.
\end{equation}
Thus,
\begin{equation*}
\p { \envert{\Re \intoo{ \mvec{v}^* \intoo{\nabla d_A(\mvec{z}) - \e{\nabla d_A(\mvec{z})}}}} > m \epsilon \lambda_{\max} (\mvec{z}) } \leq C_2 {\rm e}^{- C_1 \sqrt{ m \epsilon}}, \quad \forall \; \epsilon \geq C_3 m^{- \frac{1}{3}}.
\end{equation*}

Furthermore, note that by \eqref{sum lambda squared} and Lemma \ref{correct phase} we have

\begin{align*}
\lambda_{\max} (\mvec{z})^2 \leq \lambda_1 (\mvec{z})^2 + \lambda_2 ( \mvec{z} )^2 =  2 (1 - \envert{\mvec{x}^* \mvec{z}}) = \inf_{\theta \in \mathds{R}} \enVert{{\rm e}^{i \theta} \mvec{x} - \mvec{z} }^2.
\end{align*}

\end{milad}

\end{proof}

\section{Conclusions} \label{sec:conclude}
In this paper, we have studied the problem of employing compression codes to solve the phase retrieval problem. Given a class of structured signals and a corresponding compression code, we have proposed COPER, which provably recovers  structured signals  in that class from their phaseless measurements using the compression code. Our results have shown that, in  noiseless phase retrieval, asymptotically,  the required sampling rate  for almost zero-distortion recovery, modulo the phase, is the same as noiseless compressed sensing. 

COPER is based on a combinatorial optimization problem. Hence, we have also introduced an iterative algorithm named gradient descent COPER (GD-COPER). We have shown that GD-COPER can return an accurate estimate of the signal in polynomial time (under mild assumptions on the compression code and the initialization of the algorithm). However, GD-COPER requires more measurements than COPER.  The simulation results not only confirms the excellent performance of GD-COPER, but also shows the GD-COPER can perform pretty well even with a far initial point from the target.  This confirms that the very mild condition we had in Corollary \ref{cor:relaxedinit} for the theoretical guarantee, also works in practice.

\end{document}